\documentclass[11pt]{article}

\usepackage{tikz}
\usetikzlibrary{shapes,arrows}
\usepackage{graphicx}
\usepackage{amssymb}
\usepackage[fleqn]{amsmath}
\usepackage{epsfig}
\usepackage{cite}
\usepackage{verbatim}
\usepackage{enumerate}
\usepackage{subfigure}
\usepackage{multicol}
\usepackage{tabularx}
\usepackage{array}
\usepackage[pdftex,bookmarks=true]{hyperref}

\newcommand{\be}{\begin{equation}}
\newcommand{\ee}{\end{equation}}

\begin{document}
\newtheorem{theorem}{\bf Theorem}
\newtheorem{property}{\bf Property}
\newtheorem{corollary}{\bf Corollary}
\newtheorem{lemma}{\bf Lemma}
\newtheorem{conjecture}{\bf Conjecture}
\def\proof{\noindent{\it Proof: }}
\def\QED{\mbox{\rule[0pt]{1.5ex}{1.5ex}}}

\newcommand\blfootnote[1]{%
  \begingroup
  \renewcommand\thefootnote{}\footnote{#1}%
  \addtocounter{footnote}{-1}%
  \endgroup
}
\title{On the Capacity of the Finite Field Counterparts  \\of Wireless Interference Networks}
\author{Sundar R. Krishnamurthy and Syed A. Jafar
}
\date{}

\maketitle
 \blfootnote{Sundar R. Krishnamurthy (email: srkrishn@uci.edu) and Syed A. Jafar (email: syed@uci.edu) are with the Center of Pervasive Communications and Computing (CPCC) in the Department of Electrical Engineering and Computer Science (EECS) at the University of California Irvine. This work is accepted for presentation in part at ISIT 2013.}

\begin{abstract}
This work explores how degrees of freedom (DoF) results from wireless networks can be translated into capacity results for their finite field counterparts that arise in network coding applications.  The main insight is that scalar (SISO) finite field channels over $\mathbb{F}_{p^n}$ are analogous to $n\times n$  vector (MIMO) channels in the wireless setting, but with an important distinction -- there is additional structure due to finite field arithmetic which enforces commutativity of matrix multiplication and limits the channel \emph{diversity} to $n$, making these channels similar to diagonal channels in the wireless setting. Within the limits imposed by the channel structure, the DoF optimal precoding solutions for wireless networks can be translated into capacity optimal solutions for their finite field counterparts. This is shown through the study of  the 2-user X channel and the 3-user interference channel. Besides bringing the insights from wireless networks into network coding applications, the study of finite field networks over $\mathbb{F}_{p^n}$ also touches upon important open problems in wireless networks (finite SNR, finite diversity scenarios) through interesting parallels between $p$ and SNR, and $n$ and diversity.
\end{abstract}

\section{Introduction}\label{section:introduction} 
Precoding based network alignment (PBNA) is a network communication paradigm inspired by linear network coding and interference alignment principles \cite{ Ramakrishnan_Das_Maleki_Markopoulou_Jafar_Vishwanath, Meng_Ramakrishnan_Markopoulou_Jafar, Das_Vishwanath_Jafar_Markopoulou}. While intermediate nodes only perform arbitrary linear network coding operations which transform the network into a one-hop linear finite field network, all the intelligence resides at the source and destination nodes where information theoretically optimal encoding (precoding) and decoding is performed to achieve the capacity of the resulting linear network. The two restricting assumptions --- restricting the intelligence to the source and destination nodes, and restricting to linear operations at intermediate nodes --- are motivated  by the reduced complexity of network optimization and also by the potential to apply  the insights and techniques developed for one-hop wireless networks. Indeed, the PBNA paradigm gives rise to settings that are analogous to 1-hop wireless networks, albeit over finite fields. To highlight this distinction, we simply refer to these networks as finite field networks. There is a finite field counterpart to every 1-hop wireless network and vice versa. A number of interesting interference alignment techniques have been developed for 1-hop wireless networks and shown to be optimal from a degrees of freedom (DoF) perspective. Translating the DoF optimal schemes for wireless networks into capacity optimal schemes for finite field networks is therefore a promising research avenue. For example, the CJ scheme originally conceived  for the $K$ user \emph{time-varying} wireless interference channel in \cite{Cadambe_Jafar_int} is applied to the 3 unicast problem by Das et al. in \cite{ Ramakrishnan_Das_Maleki_Markopoulou_Jafar_Vishwanath, Meng_Ramakrishnan_Markopoulou_Jafar, Das_Vishwanath_Jafar_Markopoulou}. While the CJ scheme has also been applied successfully to the constant channel setting in wireless networks by using the rational dimensions framework of Motahari et al. in \cite{Motahari_Gharan_Khandani}, the constant channel setting remains much less understood. In this work, we  study constant channel settings, but over the finite field $\mathbb{F}_{p^n}$. 

The main contributions of this work are general insights into the correspondence between degrees of freedom of wireless networks and capacity results for their finite field counterparts. In the wireless setting, constant scalar (SISO) channels are challenging  because they lack the diversity needed for linear interference alignment schemes. Constant finite-field channels over $\mathbb{F}_{p^n}$ however, can be naturally treated as non-trivial $n\times n$ MIMO channels. A single link over $\mathbb{F}_{p^n}$ has capacity $n\log(p)$, similar to $n$ channels of capacity $\log(p)$ each. There is an immediate analogy to $n$ parallel wireless channels which would have a first order capacity $\approx n\log(\mbox{SNR})$,  establishing a correspondence between $n$ and ``diversity" (number of parallel channels)  and between $p$ and SNR.  Indeed, while scalar channels in $\mathbb{F}_{p^n}$ can be treated as $n\times n$ MIMO channels over  the base field $\mathbb{F}_p$, these channels exist in a space with diversity limited to $n$, i.e., any $n+1$ of these $n\times n$ channel matrices are linearly dependent over $\mathbb{F}_p$. Also, because of their special structure these channel matrices satisfy the commutative property of multiplication (inherited from the commutative property of multiplication in $\mathbb{F}_{p^n}$). Contrast this with generic $n\times n$ MIMO channels in $\mathbb{F}_p$, which occupy a space of diversity $n^2$ and generally do not commute. The difference is  consistent with the interpretation of $\mathbb{F}_{p^n}$ channels as similar to diagonal channels which have diversity only $n$, and are also commutative. These insights are affirmed by translating the DoF results from fixed diversity wireless networks to their $\mathbb{F}_{p^n}$ counterparts. Especially in the 3 user interference channel, the role of $n$ as the channel diversity becomes clear. 

Other interesting aspects of this work are  finer insights into  linear interference alignment and  the techniques used to prove  resolvability of desired signals from interference. Whereas in wireless networks, linear interference alignment is  feasible for either almost all channel realizations or almost none of them and is relevant primarily to the slope of the capacity curve in the infinite SNR (DoF) limit, in the finite field setting the fraction of channels where linear alignment is feasible can be a non-trivial function of $p$, so that not only we have the $p\rightarrow \infty$ behavior which corresponds to the wireless DoF results, but also we have an explicit dependence of linear alignment feasibility on $p$ for finite values of $p$. By analogy to finite SNR, this is  intriguing for its potential implications, even if the analogy is admittedly tenuous at this point. Since these finer insights are a priority in this work, we will not rely only on $p\rightarrow\infty$ assumptions to establish the capacity of the finite field networks. Instead, our goal will be to identify the capacity for all $p$ as much as possible. Because of this focus on constant channels and finite $p$, the linear independence arguments required to show resolvability of desired and interfering signals,  become a bit more challenging for finite $p$, and require  a different, somewhat novel approach. Finally, while we focus primarily on the X channel and 3 user interference channel to reveal the key insights, the insights seem to be broadly applicable and sufficient for extensions beyond these settings. 

We begin with the X channel.
\\

\section{X Channel}\label{section:X2_1}
An X network is an all-unicast setting, i.e., there is an independent message from each source node to each destination node. In this work we  study an X network with 2 source nodes, 2 destination nodes, and 4 independent messages as illustrated in Fig. \ref{fig:X_1}, also known simply as the X channel. 
\begin{figure}[h]
\begin{center}
\includegraphics[scale=0.75]{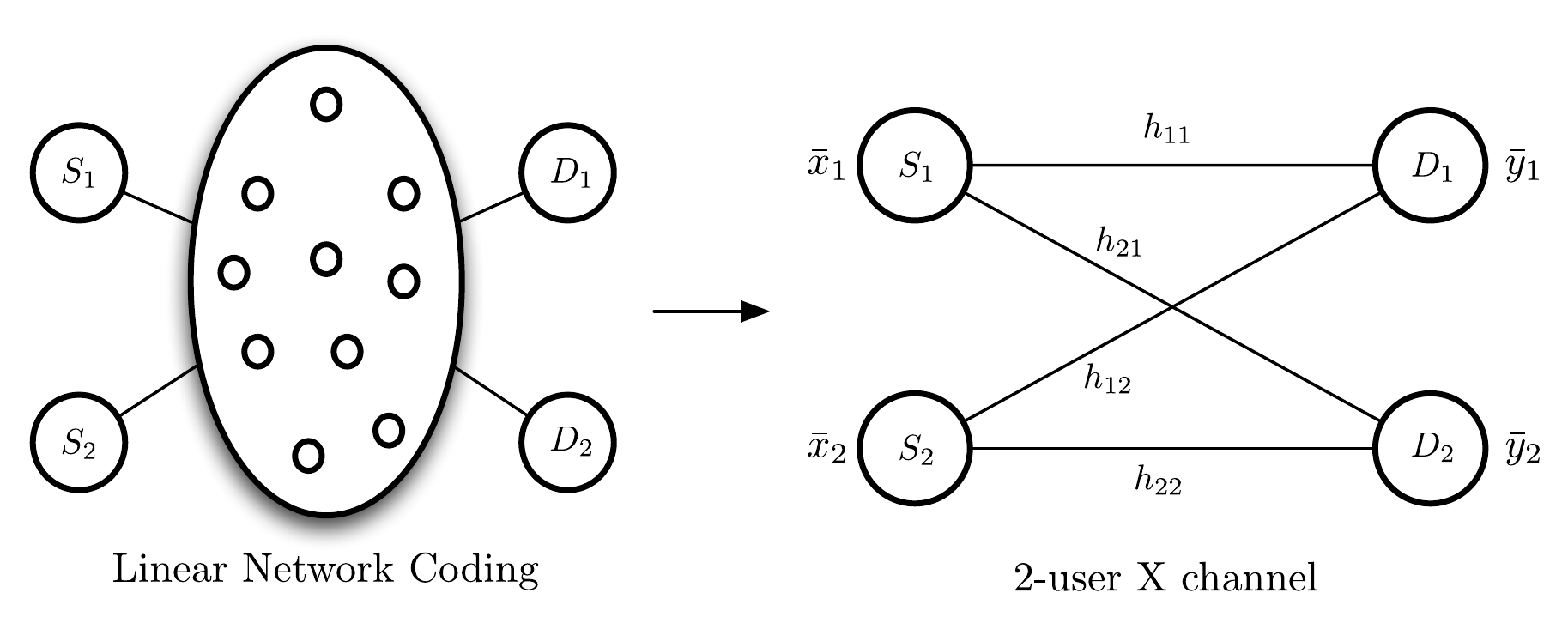}
\end{center}
\caption{Wired network modeled as 2-user X channel}\label{fig:X_1}
\end{figure}
\subsection{Prior Work}
The X channel, which contains broadcast, multiple access and interference channels as special cases, is one of the simplest, and also one of the earliest settings for interference alignment in wireless networks \cite{MMK, Jafar_Shamai}. With $A$ antennas at each node, and constant channels, the achievability of $\lfloor \frac{4A}{3}\rfloor$ DoF was shown by Maddah-Ali, Motahari and Khandani in  \cite{MMK}. Jafar and Shamai showed  in \cite{Jafar_Shamai} that $\frac{4A}{3}$ DoF were achievable when $M>1$ for constant channels, and also proved that this was the information theoretic outer bound for all $M$. For the scalar (SISO) case, i.e., $M=1$, Jafar and Shamai showed that $\frac{4}{3}$ DoF were achievable when the channels were time-varying. The DoF of the SISO case with constant and complex channels were settled in \cite{Cadambe_Jafar_Wang} by Cadambe, Jafar and Wang, who introduced asymmetric complex signaling, also known as improper Gaussian signaling and showed that it achieves the optimal value of $\frac{4}{3}$ for the  complex SISO X channel. The SISO case with constant and real coefficients was shown to achieve the optimal value of $\frac{4}{3}$ DoF  in \cite{Motahari_Gharan_Khandani} by Motahari, Gharan and Khandani, who introduced a real interference alignment framework based on rational-independence and diophantine approximation theory. Generalized degrees of freedom (GDoF) results for a symmetric SISO real constant X channel were obtained in \cite{Huang_Cadambe_Jafar} by Huang, Jafar and Cadambe, who also found a sufficient condition under which treating interference as noise is capacity optimal in the fully asymmetric case.  A capacity approximation for the real SISO constant X channel within a constant gap, subject to a small outage set, was obtained by Niesen and Maddah-Ali in \cite{Niesen_Maddah_Ali_X} using a novel deterministic channel model.  For X networks, i.e., with arbitrary number ($M$) of transmitters and arbitrary number ($N$) of receivers, Cadambe and Jafar show in \cite{Cadambe_Jafar_X} that the SISO setting with time-varying  channel coefficients has $\frac{MN}{M+N-1}$ DoF. The result is extended to the real constant SISO setting using the rational independence framework by Motahari et al. in \cite{Motahari_Gharan_Khandani}.  Partial characterizations of the DoF region are found by Wang in \cite{Zhengdao_X}. Cadambe and Jafar  show in \cite{Cadambe_Jafar_XFB} that the DoF value remains unchanged when relays and feedback are included. DoF of the time-varying MIMO X channel with $A>1$ antennas at each node are settled in \cite{Sun_Geng_Gou_Jafar} by Sun et al. who identify a one-sided decomposability property of $X$ networks, and show that the spatial scale invariance conjecture of Wang, Gou and Jafar \cite{Wang_Gou_Jafar_Subspace}  (that the DoF scale with the number of antennas) holds in this case. The DoF of a layered  multihop SISO X channel with 2 source nodes and 2 destination nodes are characterized in \cite{Wang_Gou_Jafar_MU} by Wang, Gou and Jafar, who show that the DoF can only take the values $1,\frac{4}{3}, \frac{3}{2}, \frac{5}{3}, 2$ and identify the networks that correspond to each value.  Note that all the DoF results mentioned above are meant in the `almost surely' sense, i.e., they hold for almost all channel realizations but in every case there are channels for which the DoF remain unknown. The problem is particularly severe for rational alignment and diophantine approximation based schemes for real constant channels, where while the DoF value applicable to almost all channels is known, the DoF of any given channel realization is unknown for almost all channel realizations. 

For wired networks, if intermediate nodes are intelligent, i.e., operations at intermediate nodes can be optimized, then the sum-capacity of an all-unicast network, i.e., an X network, has been shown to be achievable by routing \cite{Wang_Gou_Jafar_MU}. However, due to practical limitations, optimization of intermediate nodes may not be possible.  While the overhead and complexity of learning and optimizing individual coding coefficients at all intermediate nodes may be excessive, it is much easier to learn only the end-to-end channel coefficients, e.g.,  through network tomography, with no knowledge of the internal structure of the network or the individual coding coefficients at the intermediate nodes. This is the setting that we explore in this work.

\subsection{Finite Field X Channel Model}
Consider the finite field X channel 
\begin{eqnarray*}
\bar{ y}_1(t)&=&{ { h}_{11}}{\bar{x}_{1}}(t)+{ { h}_{12}}{\bar{x}_{2}}(t) \\
\bar{ y}_2(t)&=&{ { h}_{22}}{\bar{x}_{2}}(t)+{ { h}_{21}}{\bar{x}_{1}}(t)
\end{eqnarray*}
where, over the $t^{th}$ channel use, $\bar{x}_{i}(t)$ is the  symbol sent by source $i$, ${ h}_{ji}$ represents the channel coefficient between source $i$ and destination $j$ and $\bar{y}_j$ represents the received symbol at destination $j$. All symbols $\bar{x}_{i}(t), {h}_{ji}, \bar{y}_j(t)$ and addition and multiplication operations are in a finite field $\mathbb{F}_{p^n}$. The channel coefficients ${ h}_{ji}$ are constant and assumed to be perfectly known at all sources and  destinations. There are four independent messages, with  $W_{ji}$ denoting the message that originates at source $i$ and is intended for destination $j$.

A coding scheme over $T$ channel uses, that assigns to each message $W_{ji}$ a rate $R_{ji}$, measured in units of   $\mathbb{F}_{p^n}$ symbols  per channel use, corresponds to an encoding function at each source $i$ that maps the messages originating at that source into  a sequence of $T$ transmitted symbols, and a decoding function at each destination $j$ that maps the sequence of $T$ received symbols into decoded messages $\hat{W}_{ji}$.
\begin{eqnarray}
\mbox{Encoder 1: }&& (W_{11}, W_{21})\rightarrow {\bar x}_1(1){\bar x}_1(2)\cdots{\bar x}_1(T)\\
\mbox{Encoder 2: }&& (W_{12}, W_{22})\rightarrow {\bar x}_2(1){\bar x}_2(2)\cdots{\bar x}_2(T)\\
\mbox{Decoder 1: } &&{\bar y}_1(1){\bar y}_1(2)\cdots{\bar y}_1(T)\rightarrow (\hat{W}_{11}, \hat{W}_{12})\\
\mbox{Decoder 2: }&& {\bar y}_2(1){\bar y}_2(2)\cdots{\bar y}_2(T)\rightarrow (\hat{W}_{21}, \hat{W}_{22})
\end{eqnarray}
Each message $W_{ji}$ is uniformly distributed over $\{1,2,\cdots, \lceil p^{nTR_{ji}}\rceil\}$, $\forall  i,j\in\{1,2\}$. An error  occurs if $(\hat{W}_{11}, \hat{W}_{12}, \hat{W}_{21}, \hat{W}_{22})\neq(W_{11}, W_{12}, W_{21}, W_{22})$.  A rate tuple $(R_{11}, R_{12}, R_{21}, R_{22})$ is said to be achievable if there exist encoders and decoders such that the probability of error can be made arbitrarily small by choosing a sufficiently large $T$. The closure of all achievable rate pairs is the capacity region and the maximum value of $R_{11}+R_{12}+R_{21}+R_{22}$ across all rate tuples that belong to the capacity region, is the sum-capacity, that we will refer to as simply the capacity, denoted as $C$, for brevity. Since we are especially interested in linear interference alignment, we will also define $C_{\mbox{\tiny linear}}$ as the highest sum-rate possible through vector linear coding schemes (see, e.g., \cite{Maleki_Cadambe_Jafar}), also known as linear beamforming schemes, over the base field $\mathbb{F}_p$.

\subsection{Zero Channels}
First, let us deal with trivial cases where some of the channel coefficients are zero. 

\begin{theorem}\label{theorem:zero}
If one or more of the channel coefficients $h_{ji}$ is equal to zero, the capacity is given as follows.
\begin{enumerate}
\item If $h_{12}=h_{21}=0$ and $h_{11}, h_{22}\neq 0$, then $C=C_{\mbox{\tiny linear}}=2$.
\item If $h_{11}=h_{22}=0$ and $h_{12}, h_{21}\neq 0$, then $C=C_{\mbox{\tiny linear}}=2$.
\item If $h_{11}=h_{12}=h_{21}=h_{22}=0$, then $C=C_{\mbox{\tiny linear}}=0$.
\item In all other cases where at least one channel coefficient is zero, $C=C_{\mbox{\tiny linear}}=1$.
\end{enumerate}
\end{theorem}
\proof Cases 1, 2, 3 are trivial. The resulting channel for Case 4 is a MAC, BC or Z channel. MAC and BC have capacity 1 by min-cut max-flow theorem, and the proof for the Z channel follows from the corresponding DoF result presented in \cite{Jafar_Shamai} for the wireless setting. 

\subsection{X Channel Normalization}
\begin{figure}[h]
\begin{center}
\includegraphics[scale=0.8]{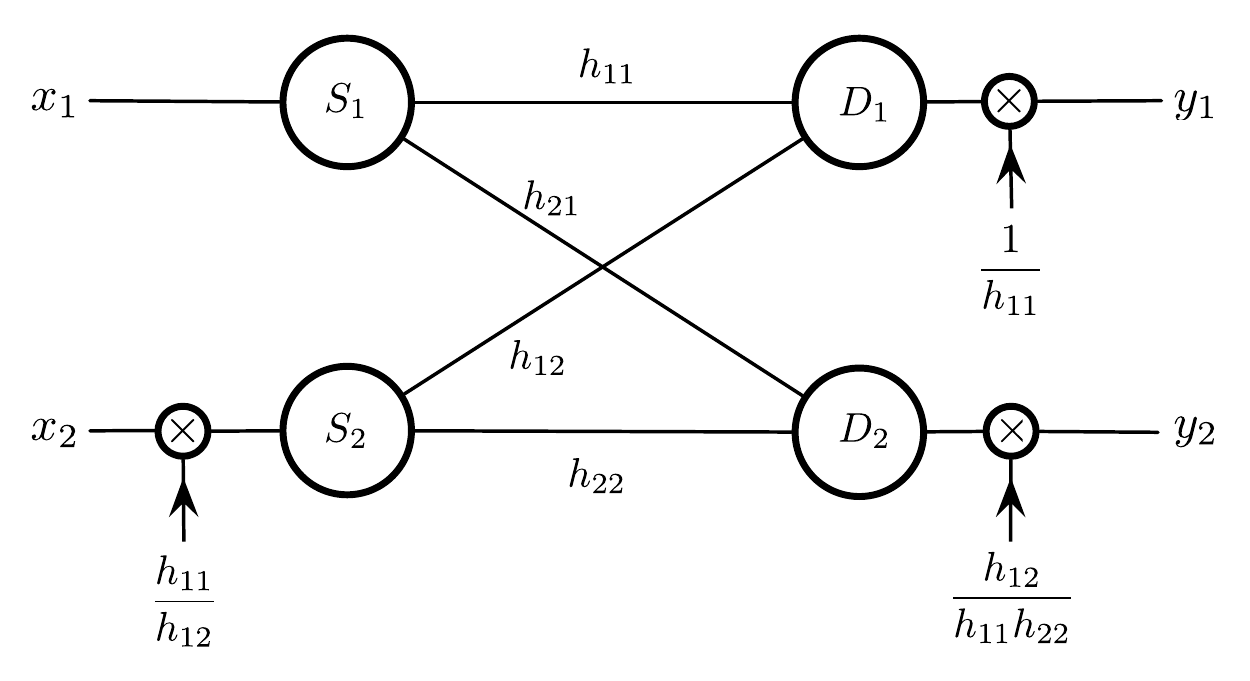}
\end{center}
\caption{Normalization in X channel}\label{fig:X_2}
\end{figure}
Based on Theorem \ref{theorem:zero}, henceforth we will assume that all channel coefficients are non-zero. We call this the fully connected X channel. Without loss of generality, let us normalize the channel coefficients by invertible operations at the sources and destinations shown in Fig. \ref{fig:X_2}. Since these are invertible operations, they do not affect the channel capacity:
\\
Destination 1 normalizes symbols by $h_{11} :  { y}_1 = \frac{\bar{ y}_1}{ h_{11}}$ 
\vspace{1mm}
\\
Destination 2 normalizes symbols by $\frac{h_{11}h_{22}}{h_{12}} : { y}_2 = \frac{\bar{ y}_2h_{12}}{h_{11}h_{22}}$ 
\vspace{1mm}
\\
Source 2 normalizes symbols by $\frac{h_{11}}{h_{12}}$ : $ \hspace{5mm} { x}_{2} = \frac{\bar{ x}_{2}h_{12}}{h_{11}}$ 
\vspace{1mm}
\\
Source 1 performs no normalization : $ \hspace{6mm} { x}_{1}=\bar{ x}_{1}$
\vspace{1mm}
\\
The normalized $X$ channel is represented as
\begin{eqnarray*}
{ y}_1&=&{{x}_{1}}+{{x}_{2}}\\
{ y}_2&=&{ {h}}{{x}_{1}}+{{x}_{2}}
\end{eqnarray*}
wherein we have reduced the channel parameters to a single channel coefficient ${h}$, defined as 
\begin{eqnarray}
{h} = \frac{h_{12}h_{21}}{h_{11}h_{22}}.
\end{eqnarray}
All symbols are still over $\mathbb{F}_{p^n}$.
\begin{figure}[h]
\begin{center}
\includegraphics[scale=0.8]{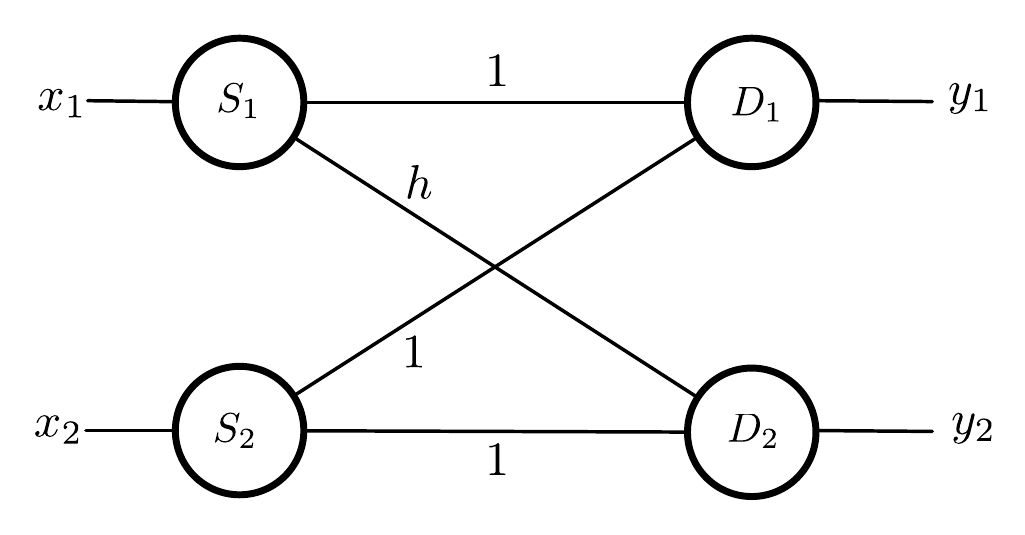}
\end{center}
\caption{Normalized X channel with Non-Zero Coefficients}\label{fig:X_3}
\end{figure}  

\subsection{Capacity of the Finite Field X Channel}
As mentioned in the review of prior work, the multiple input multiple output (MIMO) wireless X channel where each node is equipped with $n$ antennas has $\frac{4n}{3}$ DoF \cite{Jafar_Shamai, Cadambe_Jafar_Wang}. For almost all channel realizations in the wireless setting, the DoF are achieved through a linear vector space interference alignment scheme.  If $n$ is a multiple of $3$, no symbol extensions are needed and spatial beamforming is sufficient. For example, if each node is equipped with $3$ antennas, then it suffices to send 1 symbol per message, each along its assigned $3\times 1$ signal vector.  The vectors are chosen such that the two undesired symbols at each destination align in the same dimension leaving the remaining 2 dimensions free to resolve the desired signals. If $n$ is not a multiple of $3$ then $3$ symbol extensions (i.e., coding over $3$ channel uses) are needed to create a vector space within which a third of the dimensions are assigned to each message. When translating these insights into the finite field X channel with only scalar inputs and scalar outputs (SISO) we are guided by the main insight presented below.

\subsubsection{Insight: MIMO interpretation}
The main insight that forms the basis of this work is that \emph{a SISO network over $\mathbb{F}_{p^n}$ is analogous to a $n\times n$ MIMO network, albeit with a special structure imposed on the channel matrix due to finite field arithmetic. }

To appreciate this insight, let us briefly review the fundamentals. The finite field $\mathbb{F}_{p^n}$  can be used to generate an $n$-dimensional vector space as follows. Each element of $\mathbb{F}_{p^n}$ can be represented in the form
\begin{eqnarray}
z = x_{n-1}s^{n-1}+x_{n-2}s^{n-2}+\ldots+x_1s^{1}+x_0
\end{eqnarray}
wherein $z \in\mathbb{F}_{p^n}$, $x_{i} \in \mathbb{F}_{p}$. \\
\indent As an example consider  $\mathbb{F}_{3^3}$ which contains 27 elements $\{0,1,\ldots,26\}$ and each element $a \in \mathbb{F}_{3^3}$ is of the form $3^2a_2$+3$a_1$+$a_0$, wherein $a_2,a_1,a_0 \in \mathbb{F}_{3}$ with values from \{0,1,2\}. Hence every element can be written in a vector notation with coefficients $[a_2;a_1;a_0]$, e.g., $a=22$ can be written as $[2~;~1~;~1]$. 

Next, let us see how multiplication with the channel coefficient ${h}\in$ $\mathbb{F}_{3^3}$  is represented as a multiplication with a $3\times 3$ matrix with elements in $\mathbb{F}_{3}$. Consider the monic irreducible cubic polynomial $s^3+2s+1$ which is treated as zero in the field. The field itself consists of all polynomials with coefficients in $\mathbb{F}_{3}$, modulo $s^3+2s+1$.  Since $s^3+2s+1=0$ in $\mathbb{F}_{3^3}$, it follows that
\begin{eqnarray}
s^3 = -2s-1  = (3-2)s+(3-1)= s+2 \\
s^4 =s(s^3)=s(s+2)= s^2+2s
\end{eqnarray}
Since $h,x\in\mathbb{F}_{3^3}$ they can be represented as ${h} = h_2s^2+h_1s+h_0, \hspace{2mm} {x} =  x_2s^2+x_1s+x_0$ where  $h_i, x_i \in \mathbb{F}_{3^3}$. The product ${y}={hx} \in \mathbb{F}_{3^3}$ can be written as
\begin{eqnarray}
{y}={hx} &\equiv &(h_2s^2+h_1s+h_0)(x_2s^2+x_1s+x_0) \\ \nonumber
&=&s^4(h_2x_2)+s^3(h_2x_1+h_1x_2)+s^2(h_2x_0+h_0x_2+h_1x_1)+s(h_1x_0+h_0x_1)+(h_0x_0)
\end{eqnarray}
Equivalently, 
\begin{eqnarray}
\mathbf{{y}} = \mathbf{{H}{x}} = \begin{bmatrix}
h_2+h_0 & h_1 & h_2 \\
2h_2+h_1 & h_2+h_0 & h_1 \\
2h_1 & 2h_2 & h_0
\end{bmatrix}
\begin{bmatrix}
x_2 \\
x_1 \\
x_0
\end{bmatrix}
\end{eqnarray}
wherein $\mathbf{{x}}, \mathbf{{y}}$ are $3\times 1$ vector with entries from $\mathbb{F}_{3}$ and $\mathbf{{H}}$ is a $3\times 3$ matrix with its 9 entries from $\mathbb{F}_{3}$.
Here the equivalence of SISO channel over $\mathbb{F}_{3^3}$ and MIMO channel over $\mathbb{F}_{3}$ is established through the $3\times 3$ linear transformation, $\mathbf{{H}}$. Note also the structure inherent in the matrix representation  $\mathbf{{H}}$. While there are $3^9$ possible $3\times 3$ matrices over $\mathbb{F}_{3}$, there are only 27 valid $\mathbf{{H}}$ matrices, because $\mathbb{F}_{3^3}$ has only 27 elements.  This leads us to the main challenge that remains.

\subsubsection{Challenge: Channel Structure}
\emph{Given the main insight, the challenge that remains is dealing with the structural constraints on the MIMO channels that arise due to finite field arithmetic. } Structured channels are also encountered in the wireless setting ---   channels obtained by symbol extensions have a block diagonal structure \cite{Jafar_Shamai},   asymmetric complex signaling based schemes used for the SISO X channel have a unitary matrix structure \cite{Cadambe_Jafar_Wang}. Channel structure can  be destructive, e.g., loss of capacity in  rank deficient channels. However, channel structure can also be constructive, e.g., diagonal channel matrices enable the CJ scheme in \cite{Cadambe_Jafar_int}, and certain types of rank deficiencies have been shown to facilitate simpler alternatives to  interference alignment schemes \cite{Krishnamurthy_Jafar}. On the one hand, the MIMO channels, which arise by viewing $\mathbb{F}_{p^n}$ as an $n$ dimensional vector space over $\mathbb{F}_{p}$,  have a structure that is neither diagonal nor unitary. On the other hand,  diagonal channel matrices, unitary channel matrices, as well as the finite field channel matrices, all have the property that matrix multiplication is commutative, which can be a very useful property for interference alignment schemes. The impact of channel structure in the SISO constant finite field X channel setting is therefore an intriguing question.

\subsubsection{Main Result}
The  capacity result for the finite field X channel is presented in the following theorem.
\begin{theorem}\label{theorem:x_n}
For the fully connected X channel over $\mathbb{F}_{p^n}$, with $p>2$, if
\begin{eqnarray}
h = \frac{h_{12}h_{21}}{h_{11}h_{22}} \notin \mathbb{F}_p
\end{eqnarray}
then 
\begin{eqnarray}
C&=&C_{\mbox{\tiny linear}}=\frac{4}{3}
\end{eqnarray}
in units of $\mathbb{F}_{p^n}$ symbols per channel use. If $h\in\mathbb{F}_{p}$, then $C_{\mbox{\tiny linear}}=1$.
\end{theorem}
\bigskip
The setting where $h\in\mathbb{F}_p$ corresponds to the real constant SISO wireless X channel. Linear DoF collapse in this setting because even with symbol extensions, the channel matrices are simply scaled identity matrices so that the alignment of vector spaces is identical at both receivers, making it impossible to have signals align at one receiver where they are undesired and remain resolvable at the other receiver where they are desired. Since $h\in\mathbb{F}_p$ is the only exception where the capacity falls short of $4/3$, it is evident from Theorem \ref{theorem:x_n} that the capacity results for the 2 user finite field constant X over $\mathbb{F}_{p^n}$ closely mirror the corresponding DoF results for the real MIMO X channel where each user has $n$ antennas. Remarkably, even though the channels in the finite field setting are  highly structured, the structural constraints do not impact the capacity result. The significance of channel structure will become transparent when we study the 3 user interference channel later in this paper.

Note that there are $p^n-1$ possible non-zero  values for ${h}$, out of which all but $p-1$ have the capacity value of $\frac{4}{3}$ which is achieved by linear beamforming. The fraction of degenerate fully connected channel instances, for which $C_{\mbox{\tiny linear}}=1$, is therefore as follows.
\begin{eqnarray}
\frac{(p-1)}{(p^n-1)} = \frac{1}{1+p+p^2+\cdots+p^{n-1}} \label{eq:frac}
\end{eqnarray}
which approaches 0 as $p\rightarrow\infty$.
Note the similarity with the constant X channel in the wireless setting for which Cadambe et al. have shown in \cite{Cadambe_Jafar_Wang} for the complex case and Motahari et al. have shown in \cite{Motahari_Gharan_Khandani} for the real case, that  interference alignment scheme achieves $4/3$ DoF for \emph{almost all} channel realizations. Remarkably, in the finite field case the fraction of channels with linear capacity $\frac{4}{3}$ is non-trivial and still precisely computable. While a tangible connection seems elusive, it is an intriguing thought, whether interpreting $p$ and $n$ in (\ref{eq:frac})  as analogous to finite SNR and finite diversity in the wireless setting might lead to finer insights there that are not available directly from the coarse DoF metric.

\proof The information theoretic outer bound of $\frac{4}{3}$ follows immediately  from the DoF outer bound  for the wireless setting presented in \cite{Jafar_Shamai}, a combination of the Z channel bounds, with minor adjustments to account for finite field channels. The linear capacity bound of $1$ when $h\in\mathbb{F}_p$ is also straightforward because in this case, regardless of the number of channel extensions, all channel matrices are simply scaled identity matrices. Since the scaling factors are irrelevant for vector spaces, i.e., beamforming schemes, the linear capacity is not changed if we replace all channel gains with unity. But such a channel has only rank 1 (equivalently min-cut value of 1) per channel use, so its sum-rate is bounded by 1, which is therefore also an outer bound for linear capacity on the original channel. Achievability of rate 1 is trivial in a fully connected X channel. So this leaves us only to prove that  a sum rate of $\frac{4}{3}$ is achievable through vector linear schemes when $h\notin\mathbb{F}_p$. The achievability scheme is the simplest, i.e., no symbol extensions are required and only scalar linear coding (one stream per message) is sufficient, when $n$ is  $3$. For ease of exposition, the achievability proof for this case, i.e., for the X channel over $\mathbb{F}_{p^3}$ is presented first, in Section \ref{sec:p3} (an alternate proof for $\mathbb{F}_{p^3}$  is also presented in Appendix I). The achievability proof over $\mathbb{F}_{p^2}$, which requires a slightly different approach, is presented in Appendix II. The proofs over $\mathbb{F}_{p^3}$ and $\mathbb{F}_{p^2}$ are \emph{not restricted} to $p>2$. The achievability proof for the remaining general case, over $\mathbb{F}_{p^n}$,  $p>2$, is presented  in Section \ref{sec:pn}. \hfill\QED

\subsection{Achievability over $\mathbb{F}_{p^3}$}\label{sec:p3}

\begin{proof}
Consider the normalized X channel which can be characterized by single channel coefficient ${h} = \frac{h_{12}h_{21}}{h_{11}h_{22}}$ from $\mathbb{F}_{p^3}$.  We use superposition coding at the sources, wherein messages from source 1, ($W_{11}, W_{21}$) are independently encoded into symbols $x_{11}, x_{21}$, respectively, and added to obtain the transmitted symbol $x_1 = {{x}_{11}}+{{x}_{21}}$ and messages from source 2, ($W_{12}, W_{22}$) are similarly encoded as $x_2 = {{x}_{21}}+{{x}_{22}}$. Symbols $x_{ji}$ are from the subfield $\mathbb{F}_p$. The received symbols are expressed as
\begin{eqnarray*}
{ y}_1&=&{{x}_{11}}+{{x}_{12}}+{{x}_{22}}+{{x}_{21}} \\
{ y}_2&=&{ {h}}{{x}_{21}}+{ {h}}{{x}_{11}}+{{x}_{22}}+{{x}_{12}}
\end{eqnarray*}
wherein $h, y_j \in \mathbb{F}_{p^3}$.

As described earlier, $\mathbb{F}_{p^3}$ can be split into a 3-dimensional space over subfield $\mathbb{F}_{p}$ so that the output has 3 dimensions (each over $\mathbb{F}_{p}$) within which 2 desired symbols and 2 interference symbols are present at each destination. To achieve  capacity, the 2 interference symbols should be aligned at each destination such that they occupy only one dimension at that destination while remaining distinguishable at the other destination where they are desired. To this end, we will assign a precoding ``vector" ${v}_{ji}\in\mathbb{F}_{p^3}$ to each symbol $x_{ji}$. 
\begin{eqnarray*}
{ y}_1&=&{v_{11}{x}_{11}}+{v_{12}{x}_{12}}+{v_{22}{x}_{22}}+{v_{21}{x}_{21}} \\
{ y}_2&=&{v_{22}{x}_{22}}+{ {h}}{v_{21}{x}_{21}}+{ {h}}{v_{11}{x}_{11}}+{v_{12}{x}_{12}}
\end{eqnarray*}
Equivalently, using vector notation,
\begin{eqnarray*}
{ \bf y}_1&=&{{\bf v}_{11}{x}_{11}}+{{\bf v}_{12}{x}_{12}}+{{\bf v}_{22}{x}_{22}}+{{\bf v}_{21}{x}_{21}} \\
{ {\bf y}}_2&=&{{\bf v}_{22}{x}_{22}}+{ \mathbf{H}}{{\bf v}_{21}{x}_{21}}+{ \mathbf{H}}{{\bf v}_{11}{x}_{11}}+{{\bf v}_{12}{x}_{12}}
\end{eqnarray*}
 wherein ${\bf y}_j, {\bf v}_{ji}\in\mathbb{F}_{p}^{3\times 1}$ are  $3\times 1$ vectors with entries from $\mathbb{F}_{p}$ and ${\bf H}\in\mathbb{F}_p^{3\times 3}$ is a structured $3\times 3$ matrix with elements from $\mathbb{F}_{p}$, representing $h\in\mathbb{F}_{p^3}$.

\noindent Interference alignment conditions are expressed as
\begin{eqnarray}
\mbox{span}({{\bf v}_{22}}) &=& \mbox{span}({{\bf v}_{21}}) \hspace{1mm} \\
\mbox{span}({{\bf v}_{12}}) &=& \mbox{span}({ \mathbf{H}}{{\bf v}_{11}}) 
\end{eqnarray}

\noindent This is accomplished by setting
\begin{eqnarray}
{{\bf v}_{22}} &=& {{\bf v}_{21}} \hspace{6mm} \\
{{\bf v}_{12}} &=& { \mathbf{H}}{{\bf v}_{11}}
\end{eqnarray}
so that interference is aligned at each destination along one dimension. For ease of exposition, an instance of the problem and its solution are illustrated in Fig. \ref{fig:X_4} using scalar notation and again in Fig. \ref{fig:X_5} using vector notation. 
\begin{figure}[!ht]
\begin{center}
\includegraphics[scale=0.9]{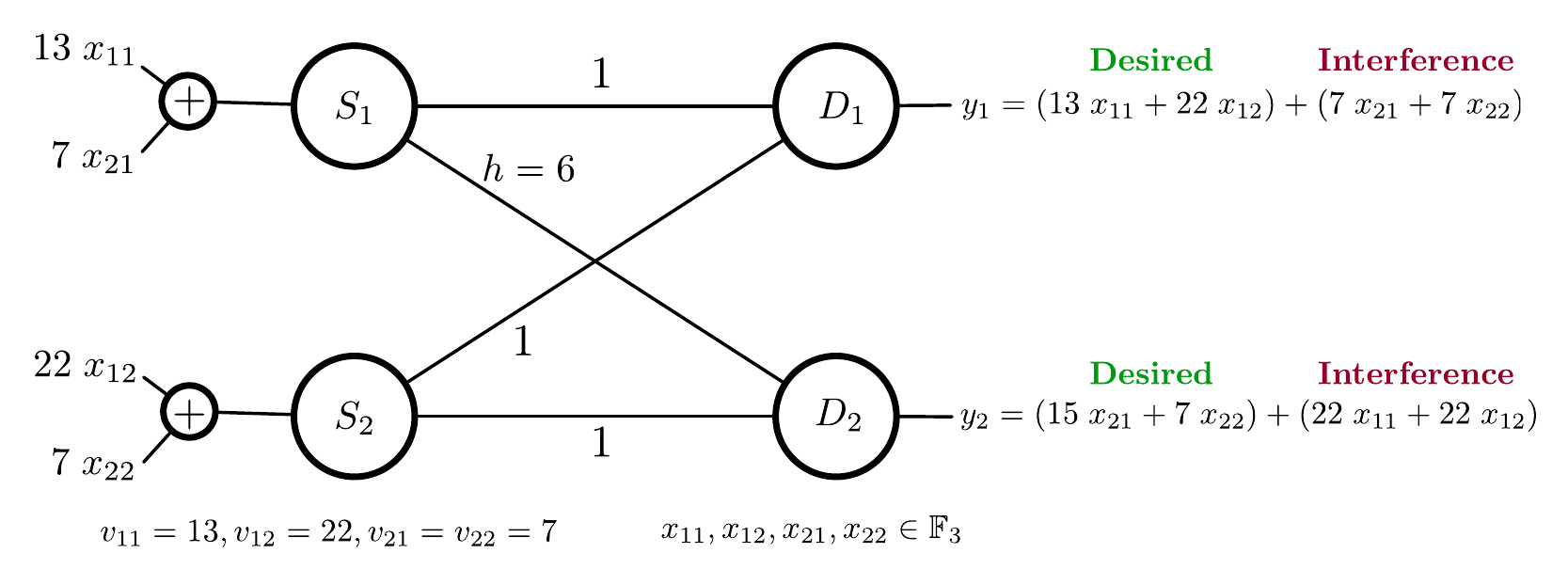}
\end{center}
\caption{An instance of the X channel over $\mathbb{F}_{3^3}$ and its capacity optimal solution represented in scalar notation.}\label{fig:X_4}
\begin{center}
\includegraphics[scale=0.9]{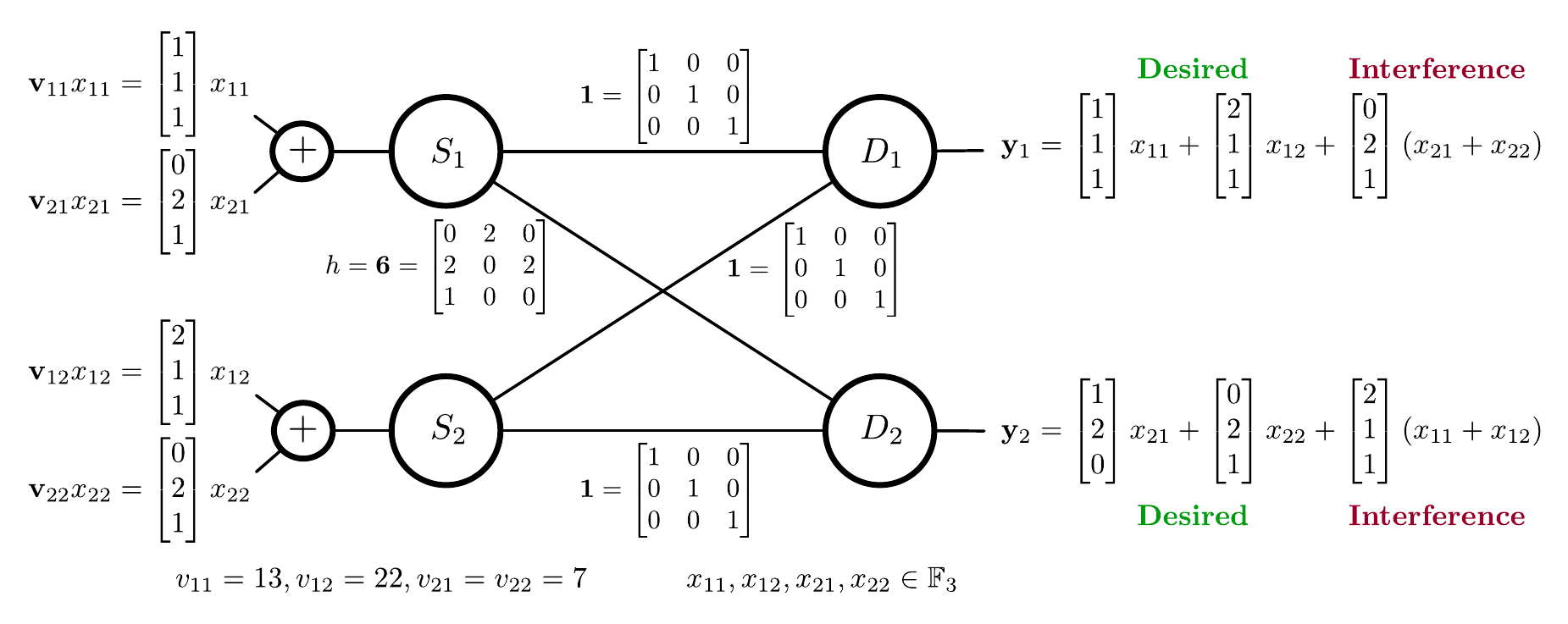}
\end{center}
\caption{The same example and solution as Fig. \ref{fig:X_4}, illustrated in vector notation.}\label{fig:X_5}
\end{figure}

 At the destinations, the spaces occupied by the two desired symbols and the aligned interference symbol are represented using matrices $S_1$ (for destination 1) and $S_2$ (for destination 2).
\begin{eqnarray}
S_1 = [{ v}_{11} \hspace{3mm} { v}_{12} \hspace{3mm} { v}_{21}] = [{ v}_{11} \hspace{3mm} {h}}{{ v}_{11} \hspace{3mm} { v}_{21}] \\
S_2 = [{ v}_{22} \hspace{3mm} {{h}}{ v}_{21} \hspace{3mm} {v}_{12}] = [{v}_{21} \hspace{3mm} {{h}}{v}_{21} \hspace{3mm} {h}}{{ v}_{11} ]
\end{eqnarray}

\indent When $h \notin \mathbb{F}_{p}$, we will now show that we can choose ${ v}_{11}$ and ${ v}_{21}$ such that elements of $S_1$ and $S_2$ are linearly independent over $\mathbb{F}_{p}$. Set $v_{21}=1$.  Then $S_1$ and $S_2$ can be written as
\begin{eqnarray}
S_1 = [{v}_{11}\hspace{3mm} {h}{ v}_{11} \hspace{3mm} 1]  \hspace{10mm} \& \hspace{10mm}
S_2 = [1 \hspace{3mm} {h} \hspace{3mm} {h}{ v}_{11} ]
\end{eqnarray}
\indent Consider $S_1$. Note that ${\bf v}_{11}$ and ${\bf H}{\bf v}_{11}$, or equivalently $v_{11}$ and $hv_{11}$, are linearly independent over $\mathbb{F}_{p}$ since $h \notin \mathbb{F}_{p}$, i.e., ${\bf H}$ is not a scaled identity matrix. Hence elements of $S_1$ are linearly independent if 1 is not a linear combination of $v_{11}$ and $hv_{11}$, or equivalently, if $\frac{1}{v_{11}}$ is not a linear combination (with coefficients from $\mathbb{F}_{p}$) of $1$ and $h$. This is guaranteed if
\begin{eqnarray}\label{eqn_a1_3}
v_{11}\notin A&\triangleq &\left\{\frac{1}{\alpha+\beta h}: \alpha, \beta\in\mathbb{F}_p, (\alpha,\beta)\neq(0,0)\right\}\cup\{0\}
\end{eqnarray}

\indent Similarly, consider $S_2$. Note that $1$ and $h$ are linearly independent over $\mathbb{F}_p$, since ${\bf H}$ is not a scaled identity matrix. Hence, elements of $S_2$ are linearly independent if $hv_{11}$ is not a linear combination of $1$ and $h$  over $\mathbb{F}_p$, or equivalently, if $v_{11}$ is not a linear combination of $\frac{1}{h}$ and $1$  over $\mathbb{F}_p$. This is guaranteed if
\begin{eqnarray}\label{eqn_a2_3}
v_{11}\notin B&\triangleq &\left\{\alpha+\frac{\beta}{ h}: \alpha, \beta\in \mathbb{F}_p,  (\alpha,\beta)\neq(0,0)\right\}\cup\{0\}
\end{eqnarray}
 Since $|A|\leq p^2$ and $|B|\leq p^2$, and  all $p$ constant polynomials are contained in both $A$ and $B$, we must have 
\begin{eqnarray}
|A\cup B|\leq 2p^2-p
\end{eqnarray}
Unless $A\cup B$ contains all $p^3$ elements of   $\mathbb{F}_{p^3}$ there is at least one choice of $v_{11}$ that satisfies both (\ref{eqn_a1_3}) and (\ref{eqn_a2_3}). In other words, the scheme works if 
\begin{eqnarray}
p^3&>&2p^2-p
\end{eqnarray}
which is true for all $p\geq 2$. Thus, we have proved the achievability of rate $\frac{1}{3}$ per message, and a sum-rate of $\frac{4}{3}$, which matches the capacity outer bound. Note that a   $\mathbb{F}_{p^3}$ symbol represents $\frac{1}{3}$ of an   $\mathbb{F}_p$ symbol and the capacity is measured in   $\mathbb{F}_{p^3}$ units because the original channel alphabet is from $\mathbb{F}_{p^3}$. Also note that the achievability proof applies to $p=2$ as well. An alternate proof for achievability of sum-rate of $\frac{4}{3}$ is presented in Appendix I. $\hfill\QED$
\end{proof}

Similar to splitting a field $\mathbb{F}_{p^3}$ to form a 3-dimensional space  in field of order p, other fields of order $p^n$ can be split to a $n$-dimensional field of order $p$. However, in order to achieve the optimal capacity of $\frac{4}{3}$, symbol extensions would be required when $n$ is not a multiple of 3. The capacity result for the general case is presented in the next section.

\subsection{Achievability over $\mathbb{F}_{p^n}$}\label{sec:pn}
\begin{proof}
Achievability proof for channels over field $\mathbb{F}_{p^2}$ is presented in Appendix II. Here, we discuss achievability proof for channels over field $\mathbb{F}_{p^n}, n > 3$. \\
Let us use 3 symbol extensions, so that we operate in a $3n$ dimensional vector space over $\mathbb{F}_p$. Each message $W_{ji}$  is encoded into $n$ streams represented by the elements of the column vector $x_{ji}\in\mathbb{F}_{p}^{n\times 1}$, and the $n$  streams are sent along the $n$ column vectors of the precoding matrix ${\bf v}_{ji}\in\mathbb{F}_{p}^{3n\times n}$, or equivalently, ${ v}_{ji}\in\mathbb{F}_{p^n}^{3\times n}$. Thus, the sum data rate is $\frac{4}{3}$ in units of $\mathbb{F}_{p^n}$ symbols per channel use, and it remains to be shown that the desired symbols are resolvable from interference at each destination.

\noindent Over each extended channel use, the received signals, $y_1, y_2 \in\mathbb{F}_{p^n}^{3\times 1}$ at each destination are expressed as:
\begin{eqnarray*}
{ y}_1&=&{v_{11}{x}_{11}}+{v_{12}{x}_{12}}+{v_{22}{x}_{22}}+{v_{21}{x}_{21}} \\
{ y}_2&=&{v_{22}{x}_{22}}+{ {h}}{v_{21}{x}_{21}}+{ {h}}{v_{11}{x}_{11}}+{v_{12}{x}_{12}}
\end{eqnarray*}
Equivalently, using vector notation the received signals, ${\bf y}_1, {\bf y}_2 \in\mathbb{F}_{p}^{3n\times 1}$ at each destination are expressed as:
\begin{eqnarray*}
{ \bf y}_1&=&{{\bf v}_{11}{x}_{11}}+{{\bf v}_{12}{x}_{12}}+{{\bf v}_{22}{x}_{22}}+{{\bf v}_{21}{x}_{21}} \\
{ {\bf y}}_2&=&{{\bf v}_{22}{x}_{22}}+{ \mathbf{H}}{{\bf v}_{21}{x}_{21}}+{ \mathbf{H}}{{\bf v}_{11}{x}_{11}}+{{\bf v}_{12}{x}_{12}}
\end{eqnarray*}
 wherein $\mathbf{H}\in\mathbb{F}_p^{3n\times 3n}$ is the channel matrix.
Interference alignment conditions are expressed as
\begin{eqnarray}
 \mbox{span}({{\bf v}_{22}}) &=& \mbox{span}({{\bf v}_{21}}) \hspace{1mm} \\
\mbox{span}({{\bf v}_{12}}) &=& \mbox{span}({ \mathbf{H}}{{\bf v}_{11}}) 
\end{eqnarray}
This is accomplished by setting
\begin{eqnarray}
{{\bf v}_{22}} &=& {{\bf v}_{21}} \hspace{6mm}\\
{{\bf v}_{12}} &=& { \mathbf{H}}{{\bf v}_{11}}
\end{eqnarray}
At each destination, $2n$ desired symbols and $n$ aligned interference symbols are represented using matrices $S_1\in\mathbb{F}_{p^n}^{3\times 3n}$ (for destination 1) and $S_2\in\mathbb{F}_{p^n}^{3\times 3n}$ (for destination 2).
\begin{eqnarray}
S_1 = [{ v}_{11} \hspace{3mm} { v}_{12} \hspace{3mm} { v}_{21}] = [{ v}_{11} \hspace{3mm} {h}}{{ v}_{11} \hspace{3mm} { v}_{21}] \\
S_2 = [{ v}_{22} \hspace{3mm} {{h}}{ v}_{21} \hspace{3mm} {v}_{12}] = [{v}_{21} \hspace{3mm} {{h}}{v}_{21} \hspace{3mm} {h}}{{ v}_{11} ]
\end{eqnarray}
\indent We will now show that when $h \notin \mathbb{F}_p$,  we can choose ${ v}_{11}$ and ${ v}_{21}$ such that the columns of $S_1$ and $S_2$ are linearly independent over $\mathbb{F}_p$. Let us choose 
\begin{eqnarray}
{ v}_{11}=g{ v}_{21}
\end{eqnarray}
with a non-zero $g \in \mathbb{F}_{p^n}$. For notational convenience, we will denote ${ v}_{21}$ as just $v\in\mathbb{F}_{p^n}^{3\times n}$. Then $S_1$ and $S_2$ can be written as
\begin{eqnarray}
S_1 &=& [{g}{ v} \hspace{3mm} {hg}{ v} \hspace{3mm} {v}]  \\
S_2 &=& [{ v} \hspace{3mm} {h}{ v}\hspace{3mm} {hg}{ v}]
\end{eqnarray}
wherein beamforming matrix $v$ has $n$ columns, denoted as $v_1,\ldots,v_n \in\mathbb{F}_{p^n}^{3\times 1}$.
\begin{eqnarray}
S_1 = [{g}{ v_1} \ldots {g}{ v_n} \hspace{3mm} {hg}{ v_1} \ldots {hg}{ v_n} \hspace{3mm} {v_1} \ldots {v_n}]  \\
S_2 = [{h}{ v_1} \ldots {h}{ v_n} \hspace{3mm} {hg}{ v_1} \ldots {hg}{ v_n}\hspace{3mm} {v_1} \ldots {v_n} ]
\end{eqnarray}

\noindent\begin{figure} [h] 
\tikzstyle{decision} = [diamond, draw, fill=blue!20, 
    text width=4.5em, text badly centered, node distance=4cm, inner sep=0pt]
\tikzstyle{blockinit} = [rectangle, draw, fill=blue!20, 
    text width=40em, text centered, rounded corners, minimum height=2em]
\tikzstyle{blocksmall} = [rectangle, draw, fill=red!20, 
    text width=3em, text centered, rounded corners, node distance=3cm, minimum height=2em]
    \tikzstyle{block} = [rectangle, draw,  fill=blue!20,
    text width=40em, text centered, rounded corners, node distance=2.1cm, minimum height=2em]
\tikzstyle{line} = [draw, -latex']
\tikzstyle{cloud} = [draw, ellipse,fill=red!20, node distance=6cm,
    minimum height=2em]

\begin{tikzpicture}[auto]
    \node [cloud] (begin) {begin};
    \node [decision, right of=begin] (decide) {$h\in\mathbb{F}_p$?};
      \node [blocksmall, right of=decide] (stop) {stop};
    \node [blockinit, below of=decide, node distance=2.0cm] (init) {Given $h$, choose $g$ such that $[g ~hg ~1]$ and $[h ~hg ~1]$ are each  linearly independent over $\mathbb{F}_p$};
        \node [block, below of=init, node distance=1.5cm] (choosev1) {Choose any non-zero $v_1\in\mathbb{F}^{3\times1}_{p^n}$, e.g., the  vector of all ones};
                \node [block, below of=choosev1, node distance=1.5cm] (choosev2) {Given $h, g, v_1$, choose $v_2\in\mathbb{F}^{3\times1}_{p^n}$ such that the 6 columns in $S_1$ and the 6 columns in $S_2$ that contain $v_1, v_2$, are each linearly independent over $\mathbb{F}_p$};
                                \node [block, below of=choosev2, node distance=2.5cm] (choosevk) {Given $h, g, v_1, v_2,\cdots,v_{k-1}$, choose $v_k\in\mathbb{F}^{3\times1}_{p^n}$ such that the $3k$ columns in $S_1$ and the $3k$ columns in $S_2$ that contain $v_1, v_2,\cdots, v_k$ are each linearly independent over $\mathbb{F}_p$};
    \node [block, below of=choosevk] (choosevn) {Given $h, g, v_1, v_2,\cdots,v_{n-1}$, choose $v_n\in\mathbb{F}^{3\times1}_{p^n}$ such that the $3n$ columns in $S_1$ and the $3n$ columns in $S_2$ that contain $v_1, v_2,\cdots, v_n$ are each linearly independent over $\mathbb{F}_p$};
        \node [blocksmall, below of=choosevn, node distance=1.5cm] (stop1) {stop};
    \path [line] (begin) -- (decide);
   \path [line] (decide) -- node {yes}(stop);
   \path [line] (decide) -- node {no}(init);
   \path [line] (init) -- (choosev1);
      \path [line] (choosev1) -- (choosev2);
            \path[line,dashed] (choosev2) -- (choosevk);
                  \path[line,dashed] (choosevk) -- (choosevn);
                   \path [line] (choosevn) -- (stop1);
\end{tikzpicture}
\caption{Algorithm for the construction of precoding vectors.}\label{fig:x_recursion}
\end{figure}
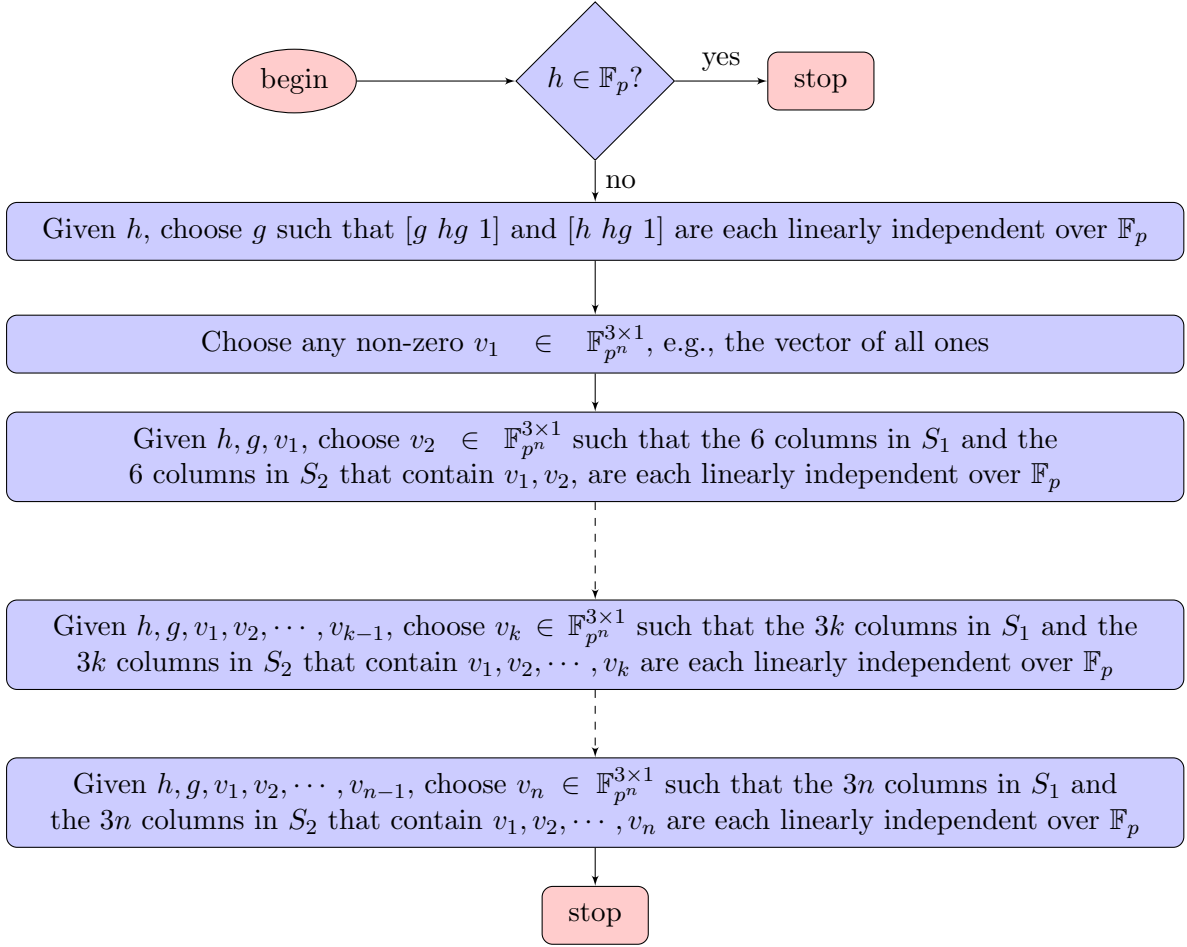

In Fig. \ref{fig:x_recursion}, we illustrate the recursive proof described hereafter.  

Choose $v_1$ as the all-ones vector. We first consider columns containing $v_1$. There are three such columns, and they need to be linearly independent in both $S_1$ and $S_2$. This requires that following vectors are linearly independent over $\mathbb{F}_p$.
\begin{eqnarray}
\text{From} \hspace{2mm} S_1: [{g}v_1 \hspace{3mm} {hg}v_1 \hspace{3mm} v_1] \\
\text{From} \hspace{2mm} S_2: [v_1 \hspace{3mm} {h}v_1 \hspace{3mm} {hg}v_1 ]
\end{eqnarray}
\indent Consider $S_1$. Note that $gv_1$ and $hgv_1$ are linearly independent over $\mathbb{F}_p$, since $h \notin \mathbb{F}_p$, i.e., $h$ is not a constant polynomial, and $g,v_1\neq 0$. Hence, elements of $S_1$ are linearly independent over $\mathbb{F}_p$ if 1 is not a linear combination of $g$ and $hg$ over $\mathbb{F}_p$, or equivalently, $\frac{1}{g}$ is not a linear combination of $1$ and $h$ over $\mathbb{F}_p$. This is guaranteed if 
\begin{eqnarray}\label{eqn_a1_n}
g\notin A&\triangleq &\left\{\frac{1}{\alpha+\beta h}: \alpha, \beta\in \mathbb{F}_p,  (\alpha,\beta)\neq(0,0)\right\}\cup\{0\}
\end{eqnarray}
\indent Similarly, consider $S_2$. Note that $v_1$ and $hv_1$ are linearly independent, since $h \notin \mathbb{F}_p$, i.e., $h$ is not a constant polynomial. Hence elements of $S_2$ are linearly independent over $\mathbb{F}_p$ if $hg$ is not a linear combination of $1$ and $h$ over $\mathbb{F}_p$, or equivalently, $g$ is not a linear combination of $1$ and $\frac{1}{h}$ over $\mathbb{F}_p$ . This is guaranteed if 
\begin{eqnarray}\label{eqn_a2_n}
g\notin B&\triangleq &\left\{\alpha+\frac{\beta}{ h}: \alpha, \beta\in \mathbb{F}_p,  (\alpha,\beta)\neq(0,0)\right\}\cup\{0\}
\end{eqnarray}
Since $|A|\leq p^{2}, |B|\leq p^2$ and $A$ and $B$ both contain all $p$ elements of $\mathbb{F}_p$, we must have $|A\cup B|\leq 2p^2-p$. Therefore, a choice of $g$ that satisfies both (\ref{eqn_a1_n}) and (\ref{eqn_a2_n}) is guaranteed to exist if 
\begin{eqnarray}
p^n>2p^2-p
\end{eqnarray}
which is true $\forall n\geq 3$.

If $v_k\neq 0$, the same choice of $g$ ensures that the following columns from $S_1$ and $S_2$ are linearly independent  over $\mathbb{F}_p$, $\forall k \in \{1,\ldots,n\}$. 
\begin{eqnarray}
\text{From} \hspace{2mm} S_1: [{g}v_k \hspace{3mm} {hg}v_k \hspace{3mm} v_k] \\
\text{From} \hspace{2mm} S_2: [v_k \hspace{3mm} {h}v_k \hspace{3mm} {hg}v_k ]
\end{eqnarray}
We  now present the recursive proof for linear independence over $\mathbb{F}_p$ of desired and interference symbols at destinations. At iteration $k$, column vector $v_{k+1}$ will be chosen based on previously chosen columns $v_1,\ldots,v_k$ and $g$. We already  chose $v_1$ to be the vector of ones. So now $v_2$ will be chosen such that following columns are linearly independent over $\mathbb{F}_p$ in $S_1$ and $S_2$ :
\begin{eqnarray}
\text{From} \hspace{2mm} S_1: [{g}v_1 \hspace{3mm} {hg}v_1 \hspace{3mm} v_1 \hspace{3mm} {g}v_2 \hspace{3mm} {hg}v_2 \hspace{3mm} v_2]\label{eq:recur1} \\
\text{From} \hspace{2mm} S_2: [{h}v_1 \hspace{3mm} {hg}v_1 \hspace{3mm} v_1 \hspace{3mm} {h}v_2 \hspace{3mm} {hg}v_2 \hspace{3mm} v_2]\label{eq:recur2}
\end{eqnarray}
Linear independence over $\mathbb{F}_p$ for (\ref{eq:recur1}) and (\ref{eq:recur2}) is guaranteed, respectively, if 
\begin{eqnarray}
v_2&\notin&A\triangleq\left\{\left(\frac{\alpha_1g+\alpha_2hg+\alpha_3}{\alpha_4g+\alpha_5hg+\alpha_6}\right)v_1:\alpha_1,\cdots,\alpha_6\in\mathbb{F}_p, (\alpha_4,\alpha_5,\alpha_6)\neq(0,0,0)\right\}\label{eq:recur3}\\
v_2&\notin&B\triangleq\left\{\left(\frac{\beta_1h+\beta_2hg+\beta_3}{\beta_4h+\beta_5hg+\beta_6}\right)v_1:\beta_1,\cdots,\beta_6\in\mathbb{F}_p, (\beta_4,\beta_5,\beta_6)\neq(0,0,0)\right\}\label{eq:recur4}
\end{eqnarray}
Now we note that
\begin{eqnarray}
 A\cap B&\supseteq&\left\{\left(\frac{\beta_{2}hg+\beta_{3}}{\beta_{5}hg+\beta_{6}}\right)v_1:\beta_1,\cdots,\beta_{6}\in\mathbb{F}_p, (\beta_{5},\beta_{6})\neq(0,0)\right\}\\
|A|&\leq&\frac{(p^3-1)p^3}{p-1}=p^5+p^4+p^3\\
|B|-|A\cap B|&\leq&\frac{(p^3-1)p^3}{p-1}-\frac{(p^2-1)p^2}{p-1} = p^5+p^4-p^2\\
|A\cup B| = |A|+|B|-|A\cap B| &\leq&2p^5+2p^4+p^3-p^2
\end{eqnarray}
Since there are $p^{3n}$ possible choices for $v_2$, there must exist at least one choice that satisfies both (\ref{eq:recur3}) and (\ref{eq:recur4}) if
\begin{eqnarray}
p^{3n}&>&2p^5+2p^4+p^3-p^2
\end{eqnarray}
which is true for all $p\geq 3$.

Similarly this recursion is carried out for choosing vectors $v_3,\ldots,v_{n-1}$. We will now describe the last stage of recursion, i.e., choosing vector $v_n$ for given $h,g,v_1,\ldots,v_{n-1}$. We want to design $v_n$ such that all $3n$ columns are linearly independent over $\mathbb{F}_p$ in $S_1$ and $S_2$ :
\begin{eqnarray}
\text{From} \hspace{2mm} S_1: [{g}v_1 \hspace{3mm} {hg}v_1 \hspace{3mm} v_1 \hspace{3mm} {g}v_2 \hspace{3mm} {hg}v_2 \hspace{3mm} v_2 \hspace{2mm}\ldots\hspace{2mm} {g}v_n \hspace{3mm} {hg}v_n \hspace{3mm} v_n] \\
\text{From} \hspace{2mm} S_2: [{h}v_1 \hspace{3mm} {hg}v_1 \hspace{3mm} v_1 \hspace{3mm} {h}v_2 \hspace{3mm} {hg}v_2 \hspace{3mm} v_2 \hspace{2mm}\ldots\hspace{2mm} {h}v_n \hspace{3mm} {hg}v_n \hspace{3mm} v_n]
\end{eqnarray}

\noindent The linear independence over $\mathbb{F}_p$ is guaranteed if
\begin{eqnarray}
v_n\notin A&\triangleq &\left\{\sum_{l=1}^{n-1}\left(\frac{\alpha_{3l-2}g+\alpha_{3l-1}hg+\alpha_{3l}}{\alpha_{3n-2}g+\alpha_{3n-1}hg+\alpha_{3n}}\right)v_l:\alpha_1,\cdots,\alpha_{3n}\in\mathbb{F}_p, (\alpha_{3n-2},\alpha_{3n-1},\alpha_{3n})\neq(0,0,0)\right\}\nonumber\\
&&\label{eq:recur5}\\
v_n\notin B&\triangleq& \left\{\sum_{l=1}^{n-1}\left(\frac{\beta_{3l-2}h+\beta_{3l-1}hg+\beta_{3l}}{\beta_{3n-2}h+\beta_{3n-1}hg+\beta_{3n}}\right)v_l:\beta_1,\cdots,\beta_{3n}\in\mathbb{F}_p, (\beta_{3n-2},\beta_{3n-1},\beta_{3n})\neq(0,0,0)\right\}\nonumber\\
&&\label{eq:recur6}\\
\Rightarrow A\cap B&\supseteq&\left\{\sum_{l=1}^{n-1}\left(\frac{\beta_{3l-1}hg+\beta_{3l}}{\beta_{3n-1}hg+\beta_{3n}}\right)v_l:\beta_1,\cdots,\beta_{3n}\in\mathbb{F}_p, (\beta_{3n-1},\beta_{3n})\neq(0,0)\right\}\label{eq:recur7}
\end{eqnarray}
Next we bound the cardinalities as follows.
\begin{eqnarray}
|A|&\leq&\frac{(p^3-1)p^{3n-3}}{p-1} = p^{3n-1}+p^{3n-2}+p^{3n-3}\\ \nonumber
|B|-|A\cap B|&\leq&\frac{(p^3-1)p^{3n-3}}{p-1}-\frac{(p^2-1)p^{2n-2}}{p-1}\\
&=&p^{3n-1}+p^{3n-2}+p^{3n-3}-p^{2n-1}-p^{2n-2}\\
|A\cup B| = |A|+|B|-|A\cap B| &\leq&2p^{3n-1}+2p^{3n-2}+2p^{3n-3}-p^{2n-1}-p^{2n-2}
\end{eqnarray}
Since there are $p^{3n}$ possible choices for $v_n$, there must exist at least one choice that satisfies both (\ref{eq:recur5}) and (\ref{eq:recur6}) if
\begin{eqnarray}
p^{3n}&>&2p^{3n-1}(1+\frac{1}{p}+\frac{1}{p^2})-p^{2n-1}-p^{2n-2}
\end{eqnarray}
which is easily shown to be true for all $p\geq 3$ as follows. If $p\geq 3$ then the RHS is bounded above by $2p^{3n-1}(1+\frac{1}{3}+\frac{1}{9})=\frac{26}{9}p^{3n-1}$ whereas the LHS is bounded below by $3p^{3n-1}$.
\hfill\QED
\end{proof} 
\vspace{2mm}
\\

\section{Interference Channel}
As noted previously, the impact of channel structure due to finite field operations in $\mathbb{F}_{p^n}$ is not evident in the capacity of the X channel as characterized in Theorem \ref{theorem:x_n}, because the capacity results for the $\mathbb{F}_{p^n}$ channels mimic the DoF results for the  generic $\mathbb{R}^{n\times n}$ real MIMO X channels in the wireless setting. In this section we will extend our study beyond the X channel, to the 3 user interference channel, where the distinction between a generic $\mathbb{R}^{n\times n}$ MIMO setting and the $\mathbb{F}_p^{n\times n}$ MIMO representations of the finite field $\mathbb{F}_{p^n}$ becomes evident. In particular, we will study the  linear sum-capacity, $C_{\mbox{\tiny linear}}$, of a finite field 3-user interference channel with 3 source nodes, 3 destination nodes and 3 independent messages as illustrated in Fig. \ref{fig:int_chan}. 

\begin{figure}[h]
\begin{center}
\includegraphics[scale=0.75]{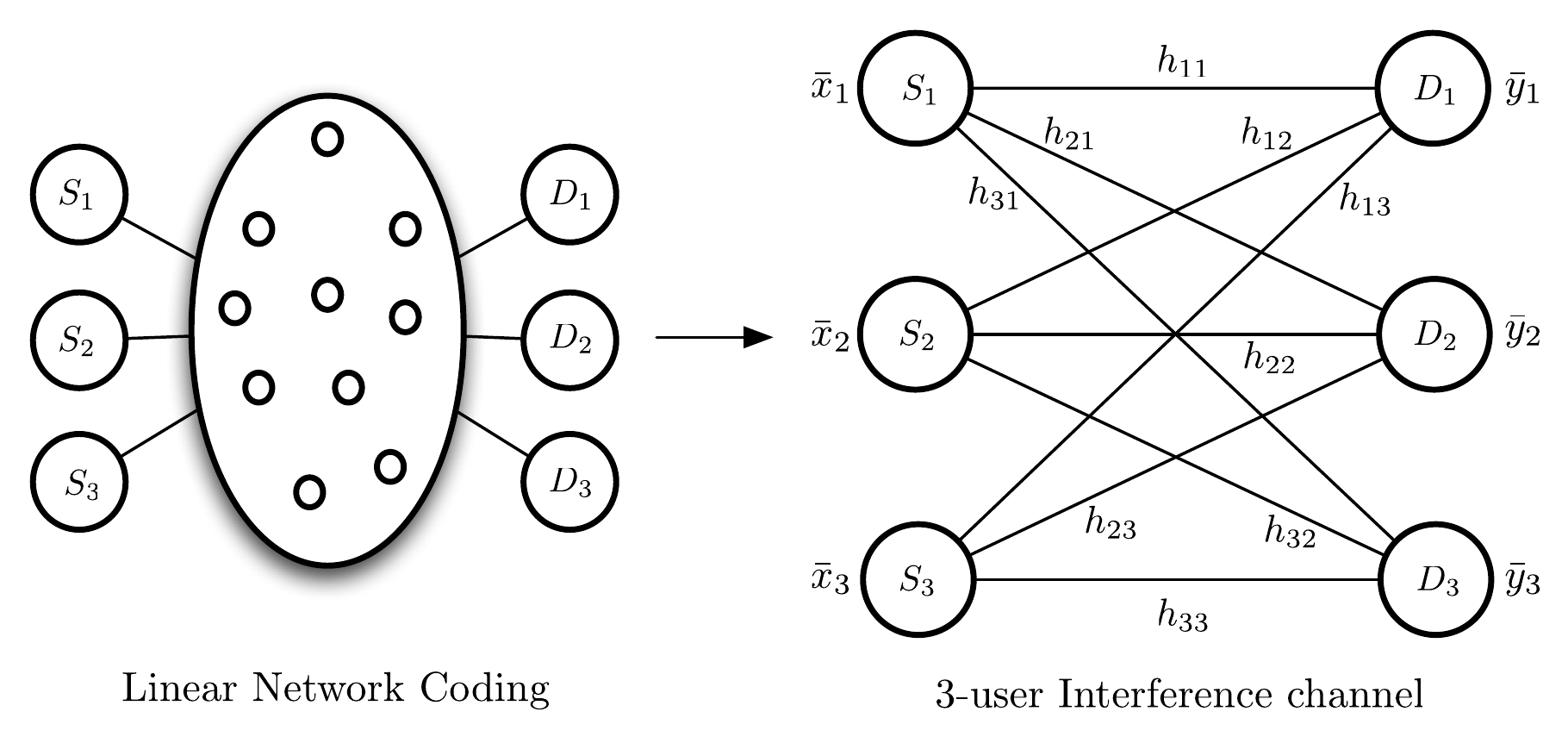}
\end{center}
\caption{Wired network modeled as 3-user interference channel}\label{fig:int_chan}
\end{figure}

\subsection{Prior Work}
The $K$ user interference channel, with $K>2$, has been extensively studied in recent years. Cadambe and Jafar showed in \cite{Cadambe_Jafar_int} that the $K$-user  fully connected interference channel with $M$ antennas at each node has $\frac{MK}{2}$ sum-DoF over a time-varying or frequency-selective channel, based on the CJ scheme. The DoF value of the 3 user constant complex MIMO interference channel with $M>1$ antennas at each node was also shown by Cadambe and Jafar, to be $\frac{3M}{2}$ using an eigenvector solution. The DoF of  asymmetric MIMO settings were characterized    in \cite{Gou_Jafar_MIMO,Ghasemi_Motahari_Khandani, Wang_Gou_Jafar_Subspace,Wang_Sun_Jafar} and the linear capacity of generic MIMO interference channels without symbol extensions was studied in \cite{Gomadam_Cadambe_Jafar, Yetis_Gou_Jafar_Kayran_TSP, Bresler_Cartwright_Tse, Razaviyayn_Lyubeznik_Luo, Gonzalez_Beltran_Santamaria, Ruan_Lau_Win,Wang_Gou_Jafar_Subspace,Bresler_Tse_Geometry}. 

For the complex constant 3 user SISO interference channel, Cadambe et al. showed in \cite{Cadambe_Jafar_Wang} that the linear DoF value is $\frac{6}{5}$ using asymmetric complex signaling scheme which precodes the real and imaginary parts of the signal separately. The constant complex SISO channel setting can be interpreted as having diversity  2.  Bresler and Tse characterized the DoF of the 3-user time-varying/frequency-selective interference channel as a function of the channel diversity, $L$, in \cite{Bresler_Tse_Diversity}. While DoF of $\frac{3}{2}$ can be achieved over channel with infinite diversity,  Bresler and Tse showed that the linear DoF of the 3-user interference channel with channel diversity $L$, is $\frac{3D}{2D+1}$ where $D = 2L-\lfloor L/2\rfloor-1$ is known as the alignment depth. Channel diversity, $L$, was shown to limit the extent to which interference signals can be aligned while maintaining the resolvability of the desired signals from interference.

In the context of network coding, the 3 unicast problem which is the counterpart of the 3 user interference channel, was studied   in \cite{Das_Vishwanath_Jafar_Markopoulou,Ramakrishnan_Das_Maleki_Markopoulou_Jafar_Vishwanath,Meng_Ramakrishnan_Markopoulou_Jafar} by Das et al., Ramakrishnan et al., and Meng et al., who introduced the Precoding-Based Network Alignment (PBNA) framework and found conditions under which half the source-destination min-cut was achievable for each user. The results were extended to networks with delay in \cite{Bavirisetti_Abhinav_Prasad_Rajan}. These works require time-varying channel coefficients due to a direct translation from the  CJ scheme originally designed for the time-varying interference channel. However, in this work we will focus only on the constant channel setting over $\mathbb{F}_{p^n}$, viewed as a constant $\mathbb{F}_p^{n\times n}$ MIMO setting. In particular, we wish to understand the significance of the channel structure.

\subsection{Finite Field Interference Channel Model}
Consider the finite field 3-user interference channel 
\begin{eqnarray*}
\bar{ y}_1(t)&=&{ { h}_{11}}{\bar{x}_{1}}(t)+{ { h}_{12}}{\bar{x}_{2}}(t)+{ { h}_{13}}{\bar{x}_{3}}(t) \\
\bar{ y}_2(t)&=&{ { h}_{21}}{\bar{x}_{1}}(t)+{ { h}_{22}}{\bar{x}_{2}}(t)+{ { h}_{23}}{\bar{x}_{3}}(t) \\
\bar{ y}_3(t)&=&{ { h}_{31}}{\bar{x}_{1}}(t)+{ { h}_{32}}{\bar{x}_{2}}(t)+{ { h}_{33}}{\bar{x}_{3}}(t)
\end{eqnarray*}
where, over the $t^{th}$ channel use, $\bar{x}_{i}(t)$ is the  symbol sent by source $i$, ${ h}_{ji}$ represents channel coefficient between source $i$ and destination $j$ and $\bar{y}_j$ represents the received symbol at destination $j$. All symbols $\bar{x}_{i}(t), {h}_{ji}, \bar{y}_j(t)$ and addition and multiplication operations are in a finite field $\mathbb{F}_{p^n}$. The channel coefficients ${ h}_{ji}$ are constant across $t$ channel uses and assumed to be perfectly known at all sources and  destinations. There are three independent messages, with $W_{i}$ denoting the message that originates at source $i$ and is intended for destination $i$.

A coding scheme over $T$ channel uses, that assigns to each message $W_{i}$ a rate $R_{i}$, measured in units of   $\mathbb{F}_{p^n}$  symbols  per channel use, corresponds to a encoding function at each source $i$ that maps the messages originating at that source into  a sequence of $T$ transmitted symbols, and a decoding function at each destination that maps the sequence of $T$ received symbols into decoded messages $\hat{W}_{i}$.
\begin{eqnarray}
\mbox{Encoder 1: }&& (W_{1})\rightarrow {\bar x}_1(1){\bar x}_1(2)\cdots{\bar x}_1(T)\\
\mbox{Encoder 2: }&& (W_{2})\rightarrow {\bar x}_2(1){\bar x}_2(2)\cdots{\bar x}_2(T)\\
\mbox{Encoder 3: }&& (W_{3})\rightarrow {\bar x}_3(1){\bar x}_3(2)\cdots{\bar x}_3(T)\\
\mbox{Decoder 1: } &&{\bar y}_1(1){\bar y}_1(2)\cdots{\bar y}_1(T)\rightarrow (\hat{W}_{1})\\
\mbox{Decoder 2: }&& {\bar y}_2(1){\bar y}_2(2)\cdots{\bar y}_2(T)\rightarrow (\hat{W}_{2})\\
\mbox{Decoder 3: }&& {\bar y}_3(1){\bar y}_3(2)\cdots{\bar y}_3(T)\rightarrow (\hat{W}_{3})
\end{eqnarray}
Each message $W_{i}$ is uniformly distributed over $\{1,2,\cdots, \lceil p^{nTR_{i}}\rceil\}$, $\forall  i\in\{1,2,3\}$. An error occurs if $(\hat{W}_{1}, \hat{W}_{2}, \hat{W}_{3})\neq(W_{1}, W_{2}, W_{3})$.  A rate tuple $(R_{1}, R_{2}, R_{3})$ is said to be achievable if there exist encoders and decoders such that the probability of error can be made arbitrarily small by choosing a sufficiently large $T$. The closure of all achievable rate pairs is the capacity region and the maximum value of $R_{1}+R_{2}+R_{3}$ across all rate tuples that belong to the capacity region, is the sum-capacity, C. Since we are interested in linear interference alignment, we will again define \emph{linear} capacity, $C_{linear}$, as the highest sum-rate possible through vector linear coding schemes  over the base field $\mathbb{F}_p$.

\subsection{Interference Channel Normalization} \label{interf_norm_sec}
As noted in the X channel, since the main insights come from the fully connected setting, we will assume that all channel coefficients are non-zero. Channel settings where some of the channels are zero are dealt with separately in the Appendix III. Without loss of generality, let us normalize the channel coefficients by invertible operations at the sources and destinations shown in Fig. \ref{fig:int_norm1}. Since these are invertible operations, they do not affect the channel capacity:
\\
Destination 1 normalizes symbols by $h_{12} :  \hspace{9mm} { y}_1 = \frac{\bar{ y}_1}{ h_{12}}$ 
\vspace{1mm}
\\
Destination 2 normalizes symbols by $\frac{h_{12}h_{23}}{h_{13}} : \hspace{5mm}  { y}_2 = \frac{\bar{ y}_2h_{13}}{h_{12}h_{23}}$ 
\vspace{1mm}
\\
Destination 3 normalizes symbols by $\frac{h_{12}h_{23}h_{31}}{h_{21}h_{13}} : { y}_3 = \frac{\bar{ y}_3h_{21}h_{13}}{h_{12}h_{23}h_{31}}$
\vspace{2mm}
\\
Source 1 normalizes symbols by $\frac{h_{13}h_{21}}{h_{12}h_{23}}$ : $ \hspace{12mm} { x}_{1} = \frac{\bar{ x}_{1}h_{12}h_{23}}{h_{13}h_{21}}$ 
\vspace{1mm}
\\
Source 2 performs no normalization : $ \hspace{16mm} { x}_{2}=\bar{ x}_{2}$ 
\vspace{1mm}
\\
Source 3 normalizes symbols by $\frac{h_{13}}{h_{12}}$ : $ \hspace{16mm} { x}_{3} = \frac{\bar{ x}_{3}h_{12}}{h_{13}}$
\begin{figure}[!ht]
\begin{center}
\includegraphics[scale=0.8]{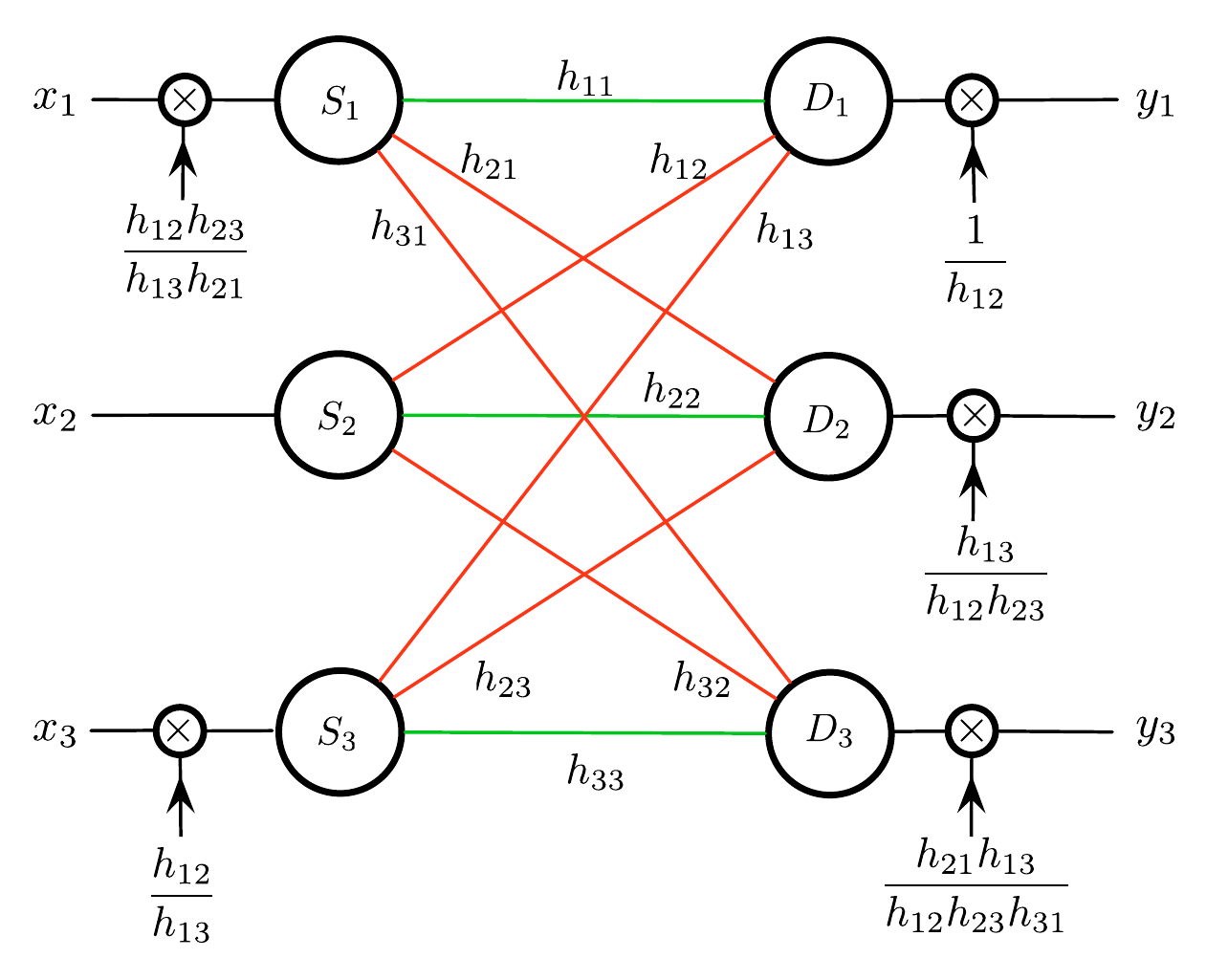}
\end{center}
\caption{Normalization in 3-user Interference Channel}\label{fig:int_norm1}
\end{figure}
\\
The normalized 3-user interference channel can be represented as 
\begin{eqnarray*}
{ y}_1&=&{\bar{h}_{11}{x}_{1}}+{{x}_{2}}+{{x}_{3}}\\
{ y}_2&=&{{x}_{1}}+{\bar{h}_{22}{x}_{2}}+{{x}_{3}} \\
{ y}_3&=&{{x}_{1}}+{\bar{h}{x}_{2}}+{\bar{h}_{33}{x}_{3}}
\end{eqnarray*}
\begin{figure}[h]
\begin{center}
\includegraphics[scale=0.8]{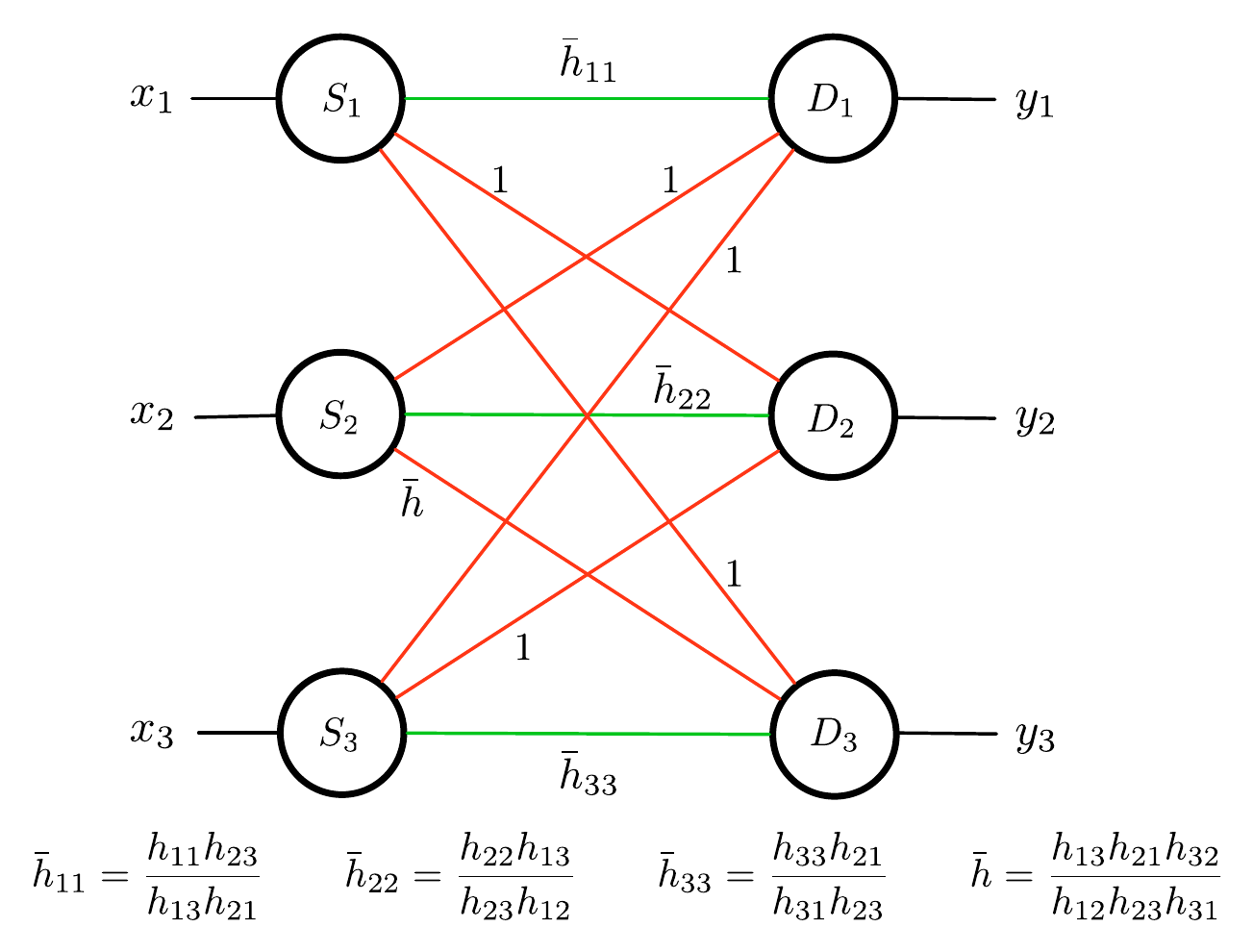}
\end{center}
\caption{Normalized 3-user Interference Channel}\label{fig:int_norm2}
\end{figure}
\\
wherein we have reduced channel parameters to four channel coefficients $\bar{h}_{11}, \bar{h}_{22}, \bar{h}_{33}, \bar{h}$, defined as 
\begin{equation}
\bar{h}_{11} = \frac{h_{11}h_{23}}{h_{13}h_{21}}, \hspace{8mm}
\bar{h}_{22} = \frac{h_{22}h_{13}}{h_{23}h_{12}}, \hspace{8mm}
\bar{h}_{33} = \frac{h_{33}h_{21}}{h_{31}h_{23}}, \hspace{8mm}
\bar{h} = \frac{h_{13}h_{21}h_{32}}{h_{12}h_{23}h_{31}}
\end{equation}

All symbols are still over $\mathbb{F}_{p^n}$ and we have the normalized interference channel illustrated in Fig. \ref{fig:int_norm2}.

\subsection{Linear-scheme Capacity of the Finite Field Interference Channel}
In the study of the X channel, we noted how scalar channels over $\mathbb{F}_{p^n}$ can be viewed as $n\times n$ MIMO channels over $\mathbb{F}_p$. Let us see if the same insight can be carried over to the 3 user interference channel. For the 3 user MIMO interference channel, an eigenvector based interference alignment solution that achieves the optimal DoF value, is introduced by Cadambe and Jafar in \cite{Cadambe_Jafar_int}. Let us see if the same solution applies in the finite field setting as well. As we will show, while the eigenvector solution holds in the wireless case for almost all channel realizations, because of channel structure in the finite field case, the solution holds only in certain `degenerate' settings, that are increasingly rare as the base field size increases, so that in the limit of infinite $p$, the eigenvector solution does not hold, almost surely.

\begin{theorem}\label{thm_3user_full32}
Fully connected 3-user interference channel over $\mathbb{F}_{p^n}$  has  capacity $C= C_{\mbox{\small linear}}=\frac{3}{2}$ for all $p>3$, if
\begin{eqnarray} 
\bar{h}_{kk}\notin \mathbb{F}_p, k \in\{1,2,3\} \\
\bar{h}\in \mathbb{F}_p
\end{eqnarray}
\end{theorem}
\begin{proof}
The outer bound of $\frac{3}{2}$ extends from \cite{Cadambe_Jafar_int} with only  minor adjustments to account for operating over finite fields. Achievable scheme is presented here. Let us denote the $n\times n$ linear transformation corresponding to product by $\bar{h}$ as $H$. i.e.,  $\bar{h}\in\mathbb{F}_{p^n}$ and $H \in \mathbb{F}_{p}^{n\times n}$. The achievable scheme involves beamforming vectors $\bar{V}_1,\bar{V}_2,\bar{V}_3 \in \mathbb{F}_{p}^{n\times 1}$ at the 3 sources such that interference is aligned at all destinations. Note that we need eigenvectors of $H$ (and also the eigenvalues) to be in $\mathbb{F}_p$. This implies that the eigen vector solution of \cite{Cadambe_Jafar_int} can be used only when $\bar{h}\in \mathbb{F}_p$ to achieve linear-scheme capacity of  $\frac{3}{2}$. Note that this is analogous to the asymmetric complex signaling setting studied in \cite{Cadambe_Jafar_Wang} where because the scalar complex channels become rotation matrices over reals, they do not have eigenvectors over reals unless the rotation is identity. Since $\bar{h}\in \mathbb{F}_p$, $H$ is a scaled identity matrix,  and every vector is an eigenvector of this matrix. Let us choose the same beamforming matrices at the 3 sources, $\bar{V} = \bar{V}_1 = \bar{V}_2 = \bar{V}_3$. This ensures that interference is aligned at all destinations for the normalized 3-user interference channel. At destination 3, interference from source 2 ($\bar{h}\bar{V}$) spans the same space as interference from source 1 ($\bar{V}$), since $\bar{h} \in \mathbb{F}_p$. Having aligned interference at the destinations, we now discuss construction of the beamforming matrix for odd and even $n$, such that desired and interference symbols are linearly independent at all destinations. \\

\noindent
\emph{Achievability for even $n=2l$:} \\
When n is even ($n=2l$), we choose $\bar{V} \in \mathbb{F}_{p^n}^{1\times l}$ and send $l$ input symbols per channel use ($x_1, \ldots, x_l \in \mathbb{F}_p$) from each source. 
Since $\bar{V} = \bar{V}_1 = \bar{V}_2 = \bar{V}_3$, it can be noted that interference will be aligned at all destinations in $l$ dimensional space. Let us denote the $l$ columns of $\bar{V}$ as $V_1 , V_2 , \ldots , V_l$. Then, signal space at the three destinations can be represented as
\begin{eqnarray}
S_k = [ \bar{h}_{kk}\bar{V} \hspace{5mm} \bar{V}] = [\bar{h}_{kk}V_1 , \bar{h}_{kk}V_2 , \ldots , \bar{h}_{kk}V_l , \hspace{2mm} V_1 , V_2 , \ldots , V_l] , \hspace{3mm} k \in \{1,2,3\}
\end{eqnarray}
We now describe how to choose columns of $\bar{V}$ such that desired and interference symbols are linearly independent at all destinations.
Let us choose $V_1$ to be 1. This implies that the 2 columns $[\bar{h}_{kk}V_1 \hspace{2mm}  V_1] = [\bar{h}_{kk} \hspace{2mm}  1]$ in $S_k$ are linearly independent over $\mathbb{F}_p$ since $\bar{h}_{kk}\notin \mathbb{F}_p, k \in\{1,2,3\}$.  
Now let us construct $V_2$ such that 4 columns of $S_k$ are linearly independent over $\mathbb{F}_p$ for $k \in\{1,2,3\}$.   
\begin{eqnarray}
\text{From $S_k$, } \hspace{3mm}  V_2\notin A_k&\triangleq&\left\{  \frac{(\alpha_1 \bar{h}_{kk} + \alpha_2)V_1}{\beta_1 \bar{h}_{kk}+\beta_2}: \alpha_1, \alpha_2, \beta_1, \beta_2\in\mathbb{F}_p, (\beta_1,\beta_2)\neq(0,0)\right\} \label{eqn_n32_1}
\end{eqnarray}
Now we note that
\begin{eqnarray}
|A_k| &\leq& \frac{(p^2-1)p^{2}}{p-1} = p^3+p^2, k \in\{1,2,3\} \\
|A_1\cup A_2 \cup A_3| &\leq& 3(p^3 + p^2)
\end{eqnarray}
There are $p^{n}$ choices for $V_2$, and since $p^n > 3(p^3 + p^2)$ for all $p > 3$, there exist choices for $V_2$ such that all 3 conditions of \eqref{eqn_n32_1} hold. Choosing $V_2$ from those, we note that 4 columns of $S_1,S_2,S_3$ are linearly independent over $\mathbb{F}_p$. We proceed recursively in a similar manner, for choosing columns $V_3, V_4, \ldots, V_{l-1}$ such that $6, 8, \ldots, 2(l-1)$ columns are linearly independent over $\mathbb{F}_p$ respectively, in all $S_k, k \in\{1,2,3\}$. \\
Let us now discuss the last iteration wherein we choose column $V_l$ such that all $n=2l$ columns are linearly independent over $\mathbb{F}_p$ in all $S_k, k \in\{1,2,3\}$, given that $2l-2$ columns are already linearly independent with appropriate choices of $V_1, V_2, \ldots, V_{l-1}$.
\begin{align} \nonumber
\text{From $S_k$, }  \hspace{3mm} V_l \notin A_k \triangleq \{  \frac{(\alpha_1 \bar{h}_{kk} + \alpha_2)V_1 + (\alpha_3 \bar{h}_{kk} + \alpha_4)V_2 + \cdots + (\alpha_{2l-3} \bar{h}_{kk} + \alpha_{2l-2})V_{l-1}}{\beta_1 \bar{h}_{kk}+\beta_2}: \\
\alpha_i, \beta_1, \beta_2\in\mathbb{F}_p, i\in \{1,\ldots,2l-2\}, (\beta_1,\beta_2)\neq(0,0)\} \label{eqn_n32_2}
\end{align}
Now we note that
\begin{eqnarray}
|A_k| &\leq& \frac{(p^2-1)p^{2l-2}}{p-1} = p^{2l-1}+p^{2l-2}, k \in\{1,2,3\} \\
|A_1\cup A_2 \cup A_3| &\leq& 3(p^{2l-1}+p^{2l-2})
\end{eqnarray}
There are $p^{n} = p^{2l}$ choices for $V_l$, and since $p^{2l} > 3(p^{2l-1}+p^{2l-2})$ for all $p > 3$, there exist choices for $V_l$ such that all 3 conditions of \eqref{eqn_n32_2} hold. Choosing $V_l$ from those, we note that all $n$ columns of $S_1,S_2,S_3$ are linearly independent over $\mathbb{F}_p$. Hence, desired and interference symbols are linearly independent at all destinations. Thus, sum rate of $\frac{3}{2}$ is achieved over $\mathbb{F}_{p^n}$ for all even $n$ with $p>3$, if $\bar{h}_{kk}\notin \mathbb{F}_p, k \in\{1,2,3\}$ and $\bar{h}\in \mathbb{F}_p$. \\
 
\noindent
\emph{Achievability for odd $n=2l+1$:} \\
Consider a 2 symbol extension of the channel with $2n$ dimensions of order $p$ at each destination. We choose $\bar{V} \in \mathbb{F}_{p^n}^{2\times n}$ and send $n$ input symbols over 2 channel uses ($x_1, \ldots, x_n \in \mathbb{F}_p$) from each source. 
Interference will be aligned at all destinations in an $n$ dimensional space. The signal space at the three destinations can be represented as
\begin{eqnarray}
S_k = [ \bar{h}_{kk}\bar{V} \hspace{5mm} \bar{V}] = [\bar{h}_{kk}V_1 , \bar{h}_{kk}V_2 , \ldots , \bar{h}_{kk}V_n , \hspace{2mm} V_1 , V_2 , \ldots , V_n] , \hspace{3mm} k \in \{1,2,3\}
\end{eqnarray}
Let us choose $V_1$ to be vector of ones. This implies that the 2 columns $[\bar{h}_{kk}V_1 \hspace{2mm}  V_1]$ in $S_k$ are linearly independent over $\mathbb{F}_p$ since $\bar{h}_{kk}\notin \mathbb{F}_p, k \in\{1,2,3\}$.  
Now let us construct $V_2$ such that 4 columns of $S_k$ are linearly independent over $\mathbb{F}_p$ for $k \in\{1,2,3\}$.   
\begin{eqnarray}
\text{From $S_k$, } \hspace{3mm}  V_2\notin A_k&\triangleq&\left\{  \frac{(\alpha_1 \bar{h}_{kk} + \alpha_2)V_1}{\beta_1 \bar{h}_{kk}+\beta_2}: \alpha_1, \alpha_2, \beta_1, \beta_2\in\mathbb{F}_p, (\beta_1,\beta_2)\neq(0,0)\right\} \label{eqn_n32_3}
\end{eqnarray}
Now we note that
\begin{eqnarray}
|A_k| &\leq& \frac{(p^2-1)p^{2}}{p-1} = p^3+p^2, k \in\{1,2,3\} \\
|A_1\cup A_2 \cup A_3| &\leq& 3(p^3 + p^2)
\end{eqnarray}
There are $p^{2n}$ choices for $V_2$, and since $p^{2n} > 3(p^3 + p^2)$ for all $p$, there exist choices for $V_2$ such that all 3 conditions of \eqref{eqn_n32_3} hold. Choosing $V_2$ from those, we note that 4 columns of $S_1,S_2,S_3$ are linearly independent over $\mathbb{F}_p$. We proceed recursively in a similar manner, for choosing columns $V_3, V_4, \ldots, V_{n-1}$ such that $6, 8, \ldots, 2(n-1)$ columns are linearly independent over $\mathbb{F}_p$ respectively, in all $S_k, k \in\{1,2,3\}$. \\
Let us now discuss the last iteration wherein we choose column $V_n$ such that all $n$ columns are linearly independent over $\mathbb{F}_p$ in all $S_k, k \in\{1,2,3\}$, given that $2n-2$ columns are already linearly independent with appropriate choices of $V_1, V_2, \ldots, V_{n-1}$.
\begin{align} \nonumber
\text{From $S_k$, }  \hspace{3mm} V_n \notin A_k \triangleq \{  \frac{(\alpha_1 \bar{h}_{kk} + \alpha_2)V_1 + (\alpha_3 \bar{h}_{kk} + \alpha_4)V_2 + \cdots + (\alpha_{2n-3} \bar{h}_{kk} + \alpha_{2n-2})V_{n-1}}{\beta_1 \bar{h}_{kk}+\beta_2}: \\
\alpha_i, \beta_1, \beta_2\in\mathbb{F}_p, i\in \{1,\ldots,2n-2\}, (\beta_1,\beta_2)\neq(0,0)\} \label{eqn_n32_4}
\end{align}
Now we note that
\begin{eqnarray}
|A_k| &\leq& \frac{(p^2-1)p^{2n-2}}{p-1} = p^{2n-1}+p^{2n-2}, k \in\{1,2,3\} \\
|A_1\cup A_2 \cup A_3| &\leq& 3(p^{2n-1}+p^{2n-2})
\end{eqnarray}
There are $p^{2n}$ choices for $V_l$, and since $p^{2n} > 3(p^{2n-1}+p^{2n-2})$ for all $p > 3$, there exist choices for $V_n$ such that all 3 conditions of \eqref{eqn_n32_4} hold. Choosing $V_n$ from those, we note that all $n$ columns of $S_1,S_2,S_3$ are linearly independent over $\mathbb{F}_p$. Hence, desired and interference symbols are linearly independent at all destinations. Thus, sum rate of $\frac{3}{2}$ is achieved over $\mathbb{F}_{p^n}$ for all odd $n$ with $p>3$, if $\bar{h}_{kk}\notin \mathbb{F}_p, k \in\{1,2,3\}$ and $\bar{h}\in \mathbb{F}_p$.

The fraction of channel realizations for which the conditions $\bar{h}_{kk}\notin \mathbb{F}_p, k \in\{1,2,3\}$ and $\bar{h}\in \mathbb{F}_p$ hold, is given by
\begin{eqnarray}
\frac{p}{p^n} \times (\frac{p^n-p}{p^n})^3.
\end{eqnarray}
which goes to $0$ as $p\rightarrow\infty$.
$\hfill\QED$
\end{proof}
\\
The implications of the structure of the channel become evident now. While we have $n\times n$ MIMO channels, they behave like channels with diversity $n$, e.g, like diagonal channels, where also the eigenvector solution does not work except over a measure 0 set. To strengthen this insight, we explore the 3-user interference channel further.

\subsubsection{Insight: Channel Diversity}
As noted for X networks earlier, the finite field $\mathbb{F}_{p^n}$ is analogous to a $n \times n$ MIMO network with special channel structure. The main insight that arises out of exploring the 3-user interference channel is that \emph{$n$ is analogous to channel diversity}. This is similar to saying that a scalar channel over $\mathbb{F}_{p^n}$ is analogous to $n$ parallel channels over $\mathbb{F}_p$. In the remainder of this work, we will focus only on linear capacity $C_{linear}$ and reinforce the parallels between $n$ and channel diversity.

\subsubsection{Main Result}

It is known from \cite{Bresler_Tse_Diversity} that the 3-user interference channel over $\mathbb{F}_{p^n}$ has channel diversity $n$, and so has linear capacity of $\frac{3D}{2D+1}$ when using linear beam forming schemes with alignment depth $D=2n-\lfloor n/2 \rfloor -1$.  The alignment depth, i.e., the length of the longest chain of one-to-one alignments, which is a function of channel diversity, is the primary limiting factor impacting both achievability and converse arguments. The achievable scheme is essentially the asymptotic interference alignment scheme of \cite{Cadambe_Jafar_int}. Outer bounds for linear schemes come from the argument that that the alignment depth cannot be more than $D$, and suppose it were, then desired signal would lie in span of the interference signal at the receivers. The result translates into the finite field setting as follows. We will focus mainly on the case where $n$ is odd (the cases where $n$ is even follow similarly and will be touched upon briefly).

\begin{theorem}\label{thm_3user_odd}
The 3-user interference channel over $\mathbb{F}_{p^n}$ with odd  $n=2l+1$ has linear  capacity  $C_{linear}=\frac{3l+1}{2l+1}$ if
\begin{eqnarray}
\bar{h}_{11}\notin A&\triangleq&\left\{\frac{\alpha_0+\alpha_1\bar{h}+\ldots+\alpha_{l-1}\bar{h}^{l-1}}{\beta_0+\beta_1\bar{h}+\ldots+\beta_l\bar{h}^l}: \alpha_k,\beta_m\in\mathbb{F}_p, (\beta_0,\ldots,\beta_l)\neq(0,\ldots,0)\right\} \\
\bar{h}_{22}\notin B&\triangleq&\left\{\frac{\alpha_0+\alpha_1\bar{h}+\ldots+\alpha_l\bar{h}^l}{\beta_0+\beta_1\bar{h}+\ldots+\beta_{l-1}\bar{h}^{l-1}}: \alpha_k,\beta_m\in\mathbb{F}_p, (\beta_0,\ldots,\beta_{l-1})\neq(0,\ldots,0)\right\} \\
\bar{h}_{33}\notin C&\triangleq&\left\{\frac{\alpha_0+\alpha_1\bar{h}+\ldots+\alpha_l\bar{h}^l}{\beta_0+\beta_1\bar{h}+\ldots+\beta_{l-1}\bar{h}^{l-1}}: \alpha_k,\beta_m\in\mathbb{F}_p, (\beta_0,\ldots,\beta_{l-1})\neq(0,\ldots,0)\right\} \\
&&\beta_l\bar{h}^l+\ldots+\beta_1\bar{h}+\beta_0 \neq 0 : \beta_0,\ldots,\beta_l\in\mathbb{F}_p, (\beta_0,\ldots,\beta_l)\neq(0,\ldots,0)
\end{eqnarray}
\end{theorem}

The outer bound on linear capacity is  presented in Section \ref{lin_outer_bound_sec}.
The achievable scheme is presented next.

\subsection{Achievability}
Over $\mathbb{F}_{p^{2l+1}}$, we will show that $3l+1$ symbols can be transmitted ($l+1$ symbols from source 1 and $l$ symbols each from sources 2 and 3), and all desired symbols are resolvable at the destinations.
Symbol extensions will not be necessary here. Note that $\bar{h}$ is equivalent to the $T$ matrix used in the CJ scheme \cite{Cadambe_Jafar_int}, since beamforming directions are identified with varying powers of $\bar{h}$. 

We will first discuss the achievable scheme over $\mathbb{F}_{p^3}$ and then show how it extends to all odd $n$, $\mathbb{F}_{p^{2l+1}}$. 

\subsubsection{Achievability over $\mathbb{F}_{p^3}$}
\begin{proof} 
Let us consider the normalized 3-user interference channel  over $\mathbb{F}_{p^3}$ so that $\bar{h}_{11}, \bar{h}_{22}, \bar{h}_{33}, \bar{h} \in \mathbb{F}_{p^3}$. We will show that linear schemes can achieve the rate of $\frac{4}{3}$. Consider the finite field network wherein source 1 sends 2 symbols, $x_1^1,x_1^2 \in \mathbb{F}_{p}$, while sources 2 and 3 send only one symbol each, $x_2,x_3 \in \mathbb{F}_{p}$. 

Because of the channel normalization, we use the same beamforming direction $v \in $ $\mathbb{F}_{p^3}$ for symbols sent by sources 2 and 3, so that interference is aligned at destination 1 ($v_2=v_3=v$). At source 1, we use 2 beam forming directions $\bar{h}v$ and $v$ so that, one symbol aligns at destination 2, and another aligns at destination 3 ($v_1^1=v, v_1^2=\bar{h}v$). With these choices for beamforming directions, the received symbols  can be represented as
\begin{eqnarray*}
{ y}_1&=&{\bar{h}_{11}({vx_1^1+\bar{h}vx_1^2}})+v{{x}_{2}}+v{{x}_{3}}\\
{ y}_2&=&{vx_1^1+\bar{h}vx_1^2}+{\bar{h}_{22}vx_2}+v{x_3} \\
{ y}_3&=&{vx_1^1+\bar{h}vx_1^2}+{\bar{h}v{x}_{2}}+{\bar{h}_{33}v{x}_{3}}
\end{eqnarray*}
Note that interference is aligned along $v$ at destinations 1 and 2, while interference at destination 3 is aligned along $\bar{h}v$.  There is additional unaligned interference at destinations 2 and 3, but they both have only a single input symbol to be decoded. We have 3 dimensions at each destination over $\mathbb{F}_p$, within which desired and interference symbols need to be resolved. In order to resolve desired symbols at the destinations, the signal spaces containing desired and interference symbols need to have linearly independent elements.
\begin{eqnarray}
S_1 = [\bar{h}_{11}\bar{h}v \hspace{5mm} \bar{h}_{11}v \hspace{5mm} v ] =  \bar{h}_{11}[\bar{h} \hspace{5mm} 1 \hspace{5mm} \frac{1}{\bar{h}_{11}} ] v  \\
S_2 = [\bar{h}_{22}v \hspace{5mm} \bar{h}v \hspace{5mm} v ] =  [\bar{h}_{22} \hspace{5mm} \bar{h} \hspace{5mm} 1] v  \\
S_3 = [\bar{h}_{33}v \hspace{5mm} \bar{h}v \hspace{5mm} v ] =  [\bar{h}_{33} \hspace{5mm} \bar{h} \hspace{5mm} 1] v 
\end{eqnarray}
\begin{figure}[t]
\begin{center}
\includegraphics[scale=0.85]{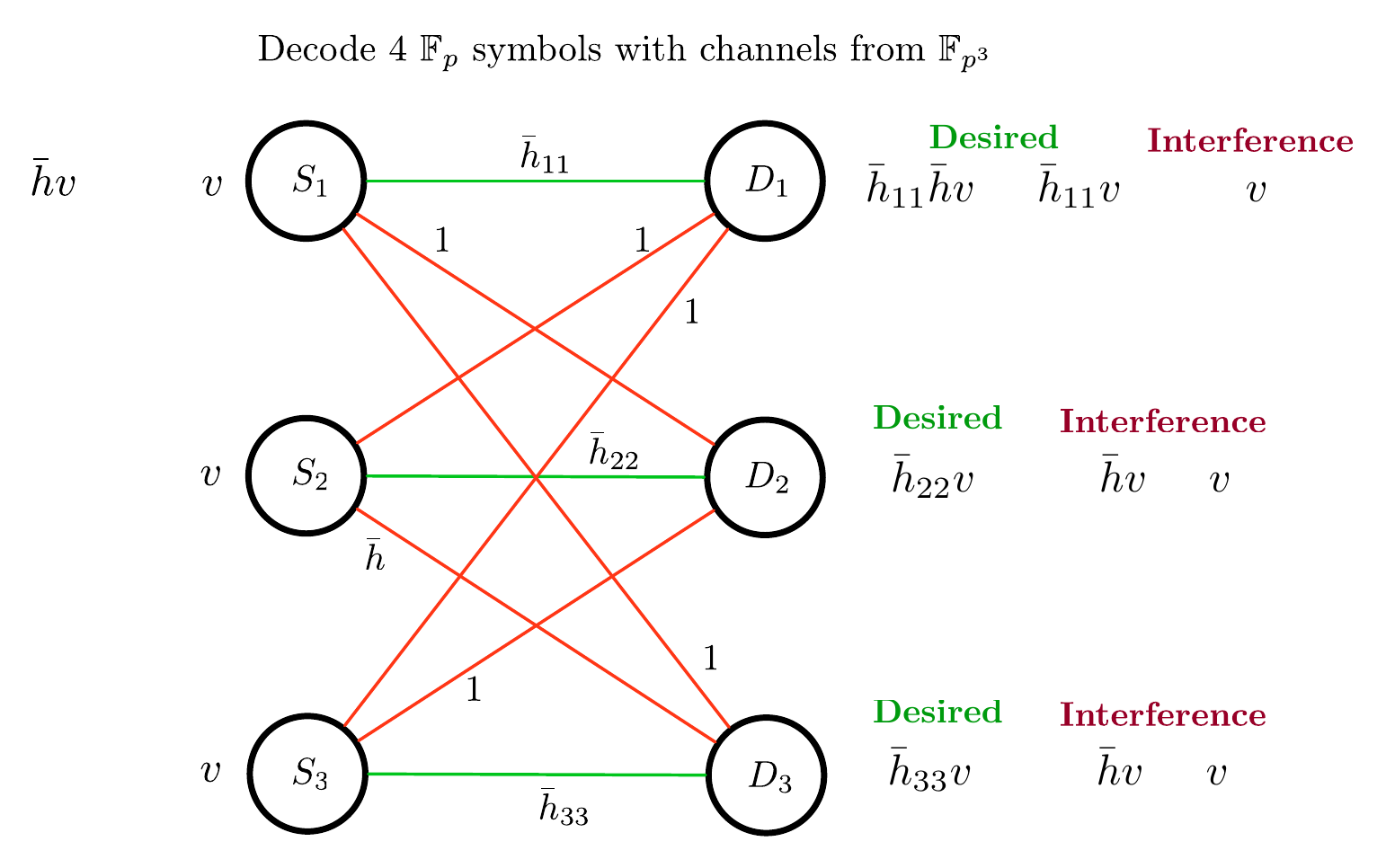}
\end{center}
\caption{3-user Interference channel over $\mathbb{F}_{p^3}$}\label{fig:int_chan_p3}
\end{figure}
When $\bar{h} \notin \mathbb{F}_p$,  $\bar{h}$ and 1 are linearly independent over $\mathbb{F}_p$. Hence, elements of $S_1$ can be linearly dependent only if $\frac{1}{\bar{h}_{11}}$ is a linear combination of $\bar{h}$ and 1. Similarly elements of $S_2$ and $S_3$ can be linearly dependent only if  
$\bar{h}_{22}$ or $\bar{h}_{33}$ is a linear combination of $\bar{h}$ and 1, respectively. Thus, the scheme works when the following conditions are satisfied.
\begin{eqnarray}
\bar{h}_{11}\notin A&\triangleq &\left\{\frac{1}{\beta_0+\beta_1\bar{h}}: \beta_0, \beta_1\in\mathbb{F}_p, (\beta_0,\beta_1)\neq(0,0)\right\}\cup\{0\} \\
\bar{h}_{22}\notin B&\triangleq &\left\{\alpha_0+\alpha_1\bar{h}: \alpha_0, \alpha_1\in\mathbb{F}_p\right\}\\
\bar{h}_{33}\notin C&\triangleq &\left\{\alpha_0+\alpha_1\bar{h}: \alpha_0, \alpha_1\in\mathbb{F}_p\right\}\\
\bar{h} \notin \mathbb{F}_p
\end{eqnarray}
Hence we can achieve the rate of $4$ $\mathbb{F}_p$ symbols per channel use, i.e., $\frac{4}{3}$ $\mathbb{F}_{p^3}$ symbols per channel use.  Fig. \ref{fig:int_chan_p3} illustrates the achievable scheme described for $\mathbb{F}_{p^3}$. $\hfill\QED$
\end{proof}

{\it Remark 1: }
We can rewrite the conditions  in terms of original channel coefficients as follows.

\begin{eqnarray}
\frac{1}{h_{11}} \notin A&\triangleq &\left\{\alpha_1\frac{h_{32}}{h_{12}h_{31}} + \beta_1\frac{h_{23}}{h_{13}h_{21}} : \alpha_1, \beta_1\in\mathbb{F}_p\right\}\\
h_{22} \notin B&\triangleq &\left\{\alpha_2\frac{h_{21}h_{32}}{h_{31}} + \beta_2\frac{h_{12}h_{23}}{h_{13}} : \alpha_2, \beta_2\in\mathbb{F}_p\right\}\\
h_{33}  \notin C&\triangleq &\left\{\alpha_3\frac{h_{13}h_{32}}{h_{12}} + \beta_3\frac{h_{31}h_{23}}{h_{21}} : \alpha_3, \beta_3\in\mathbb{F}_p\right\}
\end{eqnarray}
These conditions, which are obtained for the constant channel setting, are similar to the conditions for feasibility of PBNA derived in \cite{Meng_Ramakrishnan_Markopoulou_Jafar} for the time-varying setting, wherein 6 cofactors of off-diagonal channel coefficients are involved in the feasibility criteria. However, note that in this finite field channel, the combining coefficients $\alpha_k,\beta_k,k\in\{1,2,3\}$ can be from $\mathbb{F}_p$ whereas in \cite{Meng_Ramakrishnan_Markopoulou_Jafar}, only binary coefficients were involved. \\

{\it Remark 2:} Each of the direct channels $h_{ii}$ can take one of  $p^3$ values. At most $p^2$ of these can be linear combination of the cross channel functions. Hence, there are at least $p^3-p^2$ choices for each direct channel such that the linear independence conditions are met and so desired symbols are resolvable. The fraction of channel realizations for which $h_{ii}$ is not a linear combination of cross channel functions, is therefore at least
\begin{eqnarray}
\frac{p^3-p^2}{p^3} = 1-\frac{1}{p} \to 1 \text{ for large } p.
 \end{eqnarray}
The fraction of channels for which the scheme works, considering all conditions simultaneously  is therefore at least
 \begin{eqnarray} \nonumber
(\frac{p^3-p}{p^3}) \times (1-\frac{1}{p})^3 =  (1-\frac{1}{p^2}) \times (1-\frac{1}{p})^3 \to 1 \text{ for large } p. \hspace{4mm}     
 \end{eqnarray}
 Note that unlike the wireless case where the DoF results are proved in an almost surely sense, the guarantee on the fraction of channels for which the scheme works is much more interesting.

\subsubsection{Achievability over $\mathbb{F}_{p^n}$, $n=2l+1$}
\begin{proof}
Now let us show that the sum-rate of $\frac{3l+1}{2l+1}$ can be achieved  over $\mathbb{F}_{p^{2l+1}}$, which generalizes the proof for $\mathbb{F}_{p^3}$ discussed earlier, to any odd $n$. Suppose source 1 sends $l+1$ symbols, $x_1^1, x_1^2, \ldots x_1^{l+1} \in \mathbb{F}_p$, while sources 2 and 3 sends $l$ symbols each, $x_2^1, \ldots, x_2^l, x_3^1, \ldots, x_3^l \in \mathbb{F}_p$.  

\indent We use the same set of beamforming directions, $\bar{h}^{l-1}v, \ldots,\bar{h}v,  v$ with $v \in \mathbb{F}_{p^{2l+1}}$ for the $l$ symbols sent by sources 2 and 3, so that interference is aligned at destination 1 in span($[\bar{h}^{l-1}v \hspace{2mm} \ldots \hspace{2mm} \bar{h}v \hspace{2mm} v]$). At source 1, we use $l+1$ beamforming directions $\bar{h}^lv, \ldots, \bar{h}v, v$ so that, $l$ symbols align at destination 2, and $l$ symbols align at destination 3. With these choices of beamforming directions for input symbols, the received symbols at the destinations can be represented as

\begin{eqnarray*}
{ y}_1&=&{\bar{h}_{11}({\bar{h}^lvx_1^{l+1}+\ldots+\bar{h}vx_1^2+vx_1^{1}}})+\bar{h}^{l-1}v{{x}_2^l}+\ldots+v{{x}_2^{1}}+\bar{h}^{l-1}v{{x}_3^l}+\ldots+v{x}_3^1 \\
{ y}_2&=&{\bar{h}^lvx_1^{l+1}+\ldots+\bar{h}vx_1^2+vx_1^{1}}+{\bar{h}_{22}\bar{h}^{l-1}v{x_2^l}+\ldots+\bar{h}_{22}v{x_2^1}}+\bar{h}^{l-1}v{x_3^l}+\ldots+v{x}_3^1 \\
{ y}_3&=&{\bar{h}^lvx_1^{l+1}+\ldots+\bar{h}vx_1^2+vx_1^{1}}+{\bar{h}^{l-1}v{x_2^l}+\ldots+v{x_2^1}}+{\bar{h}_{33}\bar{h}^{l-1}v{x_3^l}+\ldots+\bar{h}_{33}v{x}_3^1}
\end{eqnarray*}
\indent In order to resolve desired symbols at the destinations, signal spaces containing desired and interference symbols need to have linearly independent entries.
 
\begin{eqnarray}
S_1 = [\bar{h}_{11}\bar{h}^lv  \hspace{2mm} \ldots \hspace{2mm} \bar{h}_{11}\bar{h}v \hspace{2mm} \bar{h}_{11}v \hspace{2mm} \bar{h}^{l-1}v \hspace{2mm} \ldots \hspace{2mm} \bar{h}v \hspace{2mm} v ] =   [\bar{h}_{11}\bar{h}^l  \hspace{2mm} \ldots \hspace{2mm} \bar{h}_{11}\bar{h} \hspace{2mm} \bar{h}_{11} \hspace{2mm} \bar{h}^{l-1} \hspace{2mm} \ldots \hspace{2mm} \bar{h} \hspace{2mm} 1 ]v   \\
S_2 = [\bar{h}_{22}\bar{h}^{l-1}v  \hspace{2mm} \ldots \hspace{2mm} \bar{h}_{22}\bar{h}v \hspace{2mm} \bar{h}_{22}v \hspace{2mm} \bar{h}^{l}v \hspace{2mm} \ldots \hspace{2mm} \bar{h}v \hspace{2mm} v ] =   [\bar{h}_{22}\bar{h}^{l-1}  \hspace{2mm} \ldots \hspace{2mm} \bar{h}_{22}\bar{h} \hspace{2mm} \bar{h}_{22} \hspace{2mm} \bar{h}^{l} \hspace{2mm} \ldots \hspace{2mm} \bar{h} \hspace{2mm} 1 ]v   \\
S_3 = [\bar{h}_{33}\bar{h}^{l-1}v  \hspace{2mm} \ldots \hspace{2mm} \bar{h}_{33}\bar{h}v \hspace{2mm} \bar{h}_{33}v \hspace{2mm} \bar{h}^{l}v \hspace{2mm} \ldots \hspace{2mm} \bar{h}v \hspace{2mm} v ] =   [\bar{h}_{33}\bar{h}^{l-1}  \hspace{2mm} \ldots \hspace{2mm} \bar{h}_{33}\bar{h} \hspace{2mm} \bar{h}_{33} \hspace{2mm} \bar{h}^{l} \hspace{2mm} \ldots \hspace{2mm} \bar{h} \hspace{2mm} 1 ]v  
\end{eqnarray}
\begin{figure}[t]
\begin{center}
\includegraphics[scale=0.79]{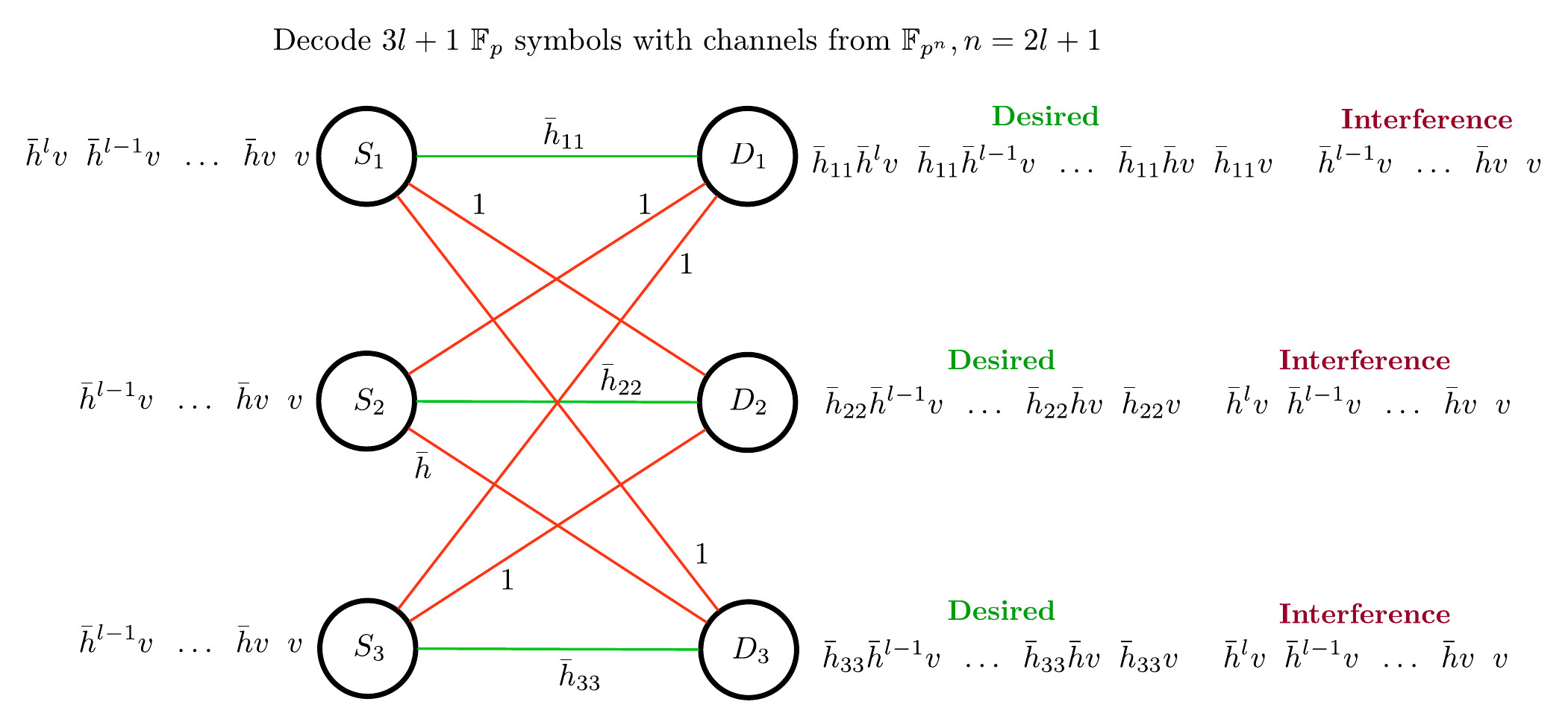}
\end{center}
\caption{3-user Interference channel over $\mathbb{F}_{p^n}$, $n=2l+1$}\label{fig:int_chan_pl}
\end{figure}
\indent The desired and interference symbols are resolvable and $3l+1$ symbols can be decoded at the destinations when the foliowing conditions are satisfied.
\begin{eqnarray}
\bar{h}_{11}\notin A&\triangleq&\left\{\frac{\alpha_0+\alpha_1\bar{h}+\ldots+\alpha_{l-1}\bar{h}^{l-1}}{\beta_0+\beta_1\bar{h}+\ldots+\beta_l\bar{h}^l}: \alpha_k,\beta_m\in\mathbb{F}_p, (\beta_0,\ldots,\beta_l)\neq(0,\ldots,0)\right\} \\
\bar{h}_{22}\notin B&\triangleq&\left\{\frac{\alpha_0+\alpha_1\bar{h}+\ldots+\alpha_l\bar{h}^l}{\beta_0+\beta_1\bar{h}+\ldots+\beta_{l-1}\bar{h}^{l-1}}: \alpha_k,\beta_m\in\mathbb{F}_p, (\beta_0,\ldots,\beta_{l-1})\neq(0,\ldots,0)\right\} \\
\bar{h}_{33}\notin C&\triangleq&\left\{\frac{\alpha_0+\alpha_1\bar{h}+\ldots+\alpha_l\bar{h}^l}{\beta_0+\beta_1\bar{h}+\ldots+\beta_{l-1}\bar{h}^{l-1}}: \alpha_k,\beta_m\in\mathbb{F}_p, (\beta_0,\ldots,\beta_{l-1})\neq(0,\ldots,0)\right\} \\
&&\beta_l\bar{h}^l+\ldots+\beta_1\bar{h}+\beta_0 \neq 0 : \beta_0,\ldots,\beta_l\in\mathbb{F}_p, (\beta_0,\ldots,\beta_l)\neq(0,\ldots,0)
\end{eqnarray}

Fig. \ref{fig:int_chan_pl} illustrates the achievable scheme described for $\mathbb{F}_{p^n}$ with $n=2l+1$. Note that a $\mathbb{F}_{p}$ symbol represents $\frac{1}{2l+1}$ of an $\mathbb{F}_{p^{2l+1}}$ symbol and rate is measured in $\mathbb{F}_{p^{2l+1}}$ units. Hence we have proved achievability of linear  capacity of $\frac{3l+1}{2l+1}$ for all odd $n=2l+1$.  $\hfill\QED$
\end{proof}
\\
\\

{\it Remark 3:} Each of the direct channels $h_{ii}$ can be from one of the $p^{2l+1}$ choices. 
The fraction of channel realizations for which direct channels satisfy the conditions is at least
\begin{eqnarray} \nonumber
&&\text{Fraction of channels with }h_{11} \text{ not in A} \geq \frac{p^{2l+1}-(p^{2l}+\ldots+p^l)}{p^{2l+1}} \\
&&= 1-\{\frac{1}{p}+\frac{1}{p^2}+\ldots+\frac{1}{p^{l+1}}\} \to 1 \text{ for large } p. \\ \nonumber
&&\text{Fraction of channels with }h_{22} \text{ or } h_{33} \text{ not in B or C} = \frac{p^{2l+1}-(p^{2l}+\ldots+p^{l+1})}{p^{2l+1}} \\
&& = 1-\{\frac{1}{p}+\frac{1}{p^2}+\ldots+\frac{1}{p^l}\} \to 1 \text{ for large } p.
 \end{eqnarray}
Also, following condition on cross channel $\bar{h}$ needs to be met
 \begin{eqnarray}
\beta_l\bar{h}^l+\ldots+\beta_1\bar{h}+\beta_0 \neq 0 : \beta_0,\ldots,\beta_l\in\mathbb{F}_p, (\beta_0,\ldots,\beta_l)\neq(0,\ldots,0)\label{eq:barh}
 \end{eqnarray}
The $l+1$ combining coefficients can repesent no more than $p^{l+1}$ distinct polynomials, and since each has degree $l$ or less, each polynomial can have at most $l$ zeros. Therefore, the number of possible $\bar{h}$ that can violate (\ref{eq:barh}) is no more than $lp^{l+1}$. So, the fraction of $\bar{h}$ values for which the scheme works is at least
\begin{eqnarray}
\frac{p^{2l+1}-lp^{l+1}}{p^{2l+1}}&=&1-\frac{l}{p^l}
\end{eqnarray}
which approaches 1 as either $p$ or $l$ approaches infinity. Putting everything together, the fraction of all channels for which the scheme works is at least
 \begin{eqnarray}
 (1-\frac{l}{p^{l}}) (1-\{\frac{1}{p}+\frac{1}{p^2}+\ldots+\frac{1}{p^{l+1}}\})
 (1-\{\frac{1}{p}+\frac{1}{p^2}+\ldots+\frac{1}{p^l}\})^2  \to 1 \text{ for large } p
 \end{eqnarray}

{\it Remark 4:} Using Lemma \ref{lemma:prime_lemma} in Appendix I, the condition on the cross channel in Theorem \ref{thm_3user_odd} can be simplified as $\bar{h} \notin \mathbb{F}_p$ for all prime $n$, since 
 \begin{eqnarray}
\beta_l\bar{h}^l+\ldots+\beta_1\bar{h}+\beta_0 \neq 0  \iff \bar{h} \notin \mathbb{F}_p
 \end{eqnarray}
So fraction of channel values for which scheme works with $n$ being prime, is at least
  \begin{eqnarray}
 (1-\frac{1}{p^{2l}}) (1-\{\frac{1}{p}+\frac{1}{p^2}+\ldots+\frac{1}{p^{l+1}}\})
 (1-\{\frac{1}{p}+\frac{1}{p^2}+\ldots+\frac{1}{p^l}\})^2  \to 1 \text{ for large } p. \hspace{4mm}     
 \end{eqnarray}
 
\subsubsection{Achievability over $\mathbb{F}_{p^2}$}
Having established the achievability proof over $\mathbb{F}_{p^n}$ for odd $n$, we will omit the general case of even $n$, except to mention that it can be translated from \cite{Bresler_Tse_Diversity} using the same  principles as illustrated for odd $n$ and  does not offer new insights. However, we will present the achievability proof for the case of $n=2$ because the corresponding result in \cite{Cadambe_Jafar_Wang} uses the asymmetric complex signaling approach which may be of interest.
As before, $\mathbb{F}_{p^2}$  can be viewed as a 2-dimensional vector space over subfield $\mathbb{F}_{p}$, much like the field of complex numbers can be viewed as a 2-dimensional vector space over reals, so that an achievable scheme similar to asymmetric complex signaling of \cite{Cadambe_Jafar_Wang} can be used. Hence, we translate the DoF result of \cite{Cadambe_Jafar_Wang} into the finite field setting as follows.   

\begin{theorem}\label{theorem:int_n2}
The 3-user interference channel over $\mathbb{F}_{p^2}$ has linear capacity, $C_{linear}=\frac{6}{5}$, if
\begin{eqnarray} \nonumber
\bar{h}_{11} = \frac{h_{11}h_{23}}{h_{13}h_{21}} \notin \mathbb{F}_p, \hspace{8mm} {\bar{h}}{\bar{h}}_{11} = \frac{h_{11}h_{32}}{h_{12}h_{31}} \notin \mathbb{F}_p \\ \nonumber
\bar{h}_{22} = \frac{h_{22}h_{13}}{h_{23}h_{12}} \notin \mathbb{F}_p, \hspace{8mm}  \frac{\bar{h}}{{\bar{h}}_{22}} = \frac{h_{21}h_{32}}{h_{22}h_{31}} \notin \mathbb{F}_p \\ \nonumber
\bar{h}_{33} = \frac{h_{33}h_{21}}{h_{31}h_{23}} \notin \mathbb{F}_p, \hspace{8mm} \frac{\bar{h}}{{\bar{h}}_{33}} = \frac{h_{32}h_{13}}{h_{33}h_{12}} \notin \mathbb{F}_p
\end{eqnarray}
\end{theorem}
\begin{proof} 
The outer bound follows from \cite{Cadambe_Jafar_Wang} in much the same fashion as the outer bound for the previous section follows from \cite{Bresler_Tse_Diversity}. Here we present only the achievability proof. Consider a 5 symbol extension of the normalized 3-user interference channel over $\mathbb{F}_{p^2}$. Over this 5 symbol extensions, 4 input symbols denoted by $x_k^1, x_k^2, x_k^3, x_k^4$ are precoded and transmitted at source $k$. Each input symbol $x_k^i, i\in \{1,2,3,4\},k\in\{1,2,3\}$ is from $\mathbb{F}_{p}$. Corresponding $5\times 1$ beam forming vectors are denoted using vectors $V_k^1, V_k^2, V_k^3, V_k^4 \in \mathbb{F}_{p^2}^{5\times 1}, k\in\{1,2,3\}$.
Each destination has 10 dimensions of order $p$ over the symbol extended channel. Desired symbols from corresponding source would occupy 4 dimensions and for resolvability, interference need to occupy only $6$ dimensions of order $p$. Hence at each destination, two of the 8 interference vectors from 2 unintended sources, need to be aligned. To this end, we make the following choices for certain beam forming vectors.
\begin{eqnarray}
V_1^3  = \bar{h} V_2^1, \hspace{8mm} V_1^4 = V_3^2, \hspace{8mm}
V_2^3  = V_3^1, \hspace{8mm} V_2^4 = \frac{1}{\bar{h}}V_1^2, \hspace{8mm}
V_3^3  = V_1^1, \hspace{8mm} V_3^4 = V_2^2
\end{eqnarray}

\noindent Desired and Interference signal space at the destinations can now be represented as follows.
\begin{eqnarray}
S_1 = [\bar{h}_{11}V_1^1 \hspace{2mm} \bar{h}_{11}V_1^2 \hspace{2mm} \bar{h}_{11}V_1^3 \hspace{2mm} \bar{h}_{11}V_1^4 \hspace{2mm} V_2^1 \hspace{2mm} V_2^2 \hspace{2mm} V_2^3 \hspace{2mm} V_2^4 \hspace{2mm} V_3^2 \hspace{2mm} V_3^3] \\
S_2 = [\bar{h}_{22}V_2^1 \hspace{2mm} \bar{h}_{22}V_2^2 \hspace{2mm} \bar{h}_{22}V_2^3 \hspace{2mm} \bar{h}_{22}V_2^4 \hspace{2mm} V_3^1 \hspace{2mm} V_3^2 \hspace{2mm} V_3^3 \hspace{2mm} V_3^4 \hspace{2mm} V_1^2 \hspace{2mm} V_1^3] \\
S_3 = [\bar{h}_{33}V_3^1 \hspace{2mm} \bar{h}_{33}V_3^2 \hspace{2mm} \bar{h}_{33}V_3^3 \hspace{2mm} \bar{h}_{33}V_3^4 \hspace{2mm} V_1^1 \hspace{2mm} V_1^2 \hspace{2mm} V_1^3 \hspace{2mm} V_1^4 \hspace{2mm} \bar{h}V_2^2 \hspace{2mm} \bar{h}V_2^3] 
\end{eqnarray}
Due to interference alignment, these matrices can be equivalently  re-written as  
\begin{eqnarray}
S_1 = [\bar{h}_{11}V_1^1 \hspace{2mm} \bar{h}_{11}V_1^2 \hspace{2mm} \bar{h}_{11}\bar{h}V_2^1 \hspace{2mm} \bar{h}_{11}V_3^2 \hspace{2mm} V_2^1 \hspace{2mm} V_2^2 \hspace{2mm} V_3^1 \hspace{2mm} \frac{1}{\bar{h}}V_1^2 \hspace{2mm} V_3^2 \hspace{2mm} V_1^1] \\
S_2 = [\bar{h}_{22}V_2^1 \hspace{2mm} \bar{h}_{22}V_2^2 \hspace{2mm} \bar{h}_{22}V_3^1 \hspace{2mm} \frac{\bar{h}_{22}}{\bar{h}}V_1^2 \hspace{2mm} V_3^1 \hspace{2mm} V_3^2 \hspace{2mm} V_1^1 \hspace{2mm} V_2^2 \hspace{2mm} V_1^2 \hspace{2mm} \bar{h}V_2^1] \\
S_3 = [\bar{h}_{33}V_3^1 \hspace{2mm} \bar{h}_{33}V_3^2 \hspace{2mm} \bar{h}_{33}V_1^1 \hspace{2mm} \bar{h}_{33}V_2^2 \hspace{2mm} V_1^1 \hspace{2mm} V_1^2 \hspace{2mm} \bar{h}V_2^1 \hspace{2mm} V_3^2 \hspace{2mm} \bar{h}V_2^2 \hspace{2mm} \bar{h}V_3^1] 
\end{eqnarray}
In order to resolve desired signals at all destinations, the columns of these 3 matrices need to be linearly independent over $\mathbb{F}_p$. The following six conditions are required.
\begin{eqnarray} \nonumber
\bar{h}_{11} \notin \mathbb{F}_p, \hspace{8mm} {\bar{h}}{\bar{h}}_{11} \notin \mathbb{F}_p, \hspace{8mm} 
\bar{h}_{22} \notin \mathbb{F}_p, \hspace{8mm}  \frac{\bar{h}}{{\bar{h}}_{22}} \notin \mathbb{F}_p, \hspace{8mm} 
\bar{h}_{33} \notin \mathbb{F}_p, \hspace{8mm} \frac{\bar{h}}{{\bar{h}}_{33}} \notin \mathbb{F}_p
\end{eqnarray}

We will now choose beam forming vectors $V_k^i, i\in \{1,2\},k\in\{1,2,3\}$, such that all three matrices $S_k$ have their 10 columns linearly independent.
\begin{figure}[h]
\begin{center}
\includegraphics[scale=0.78]{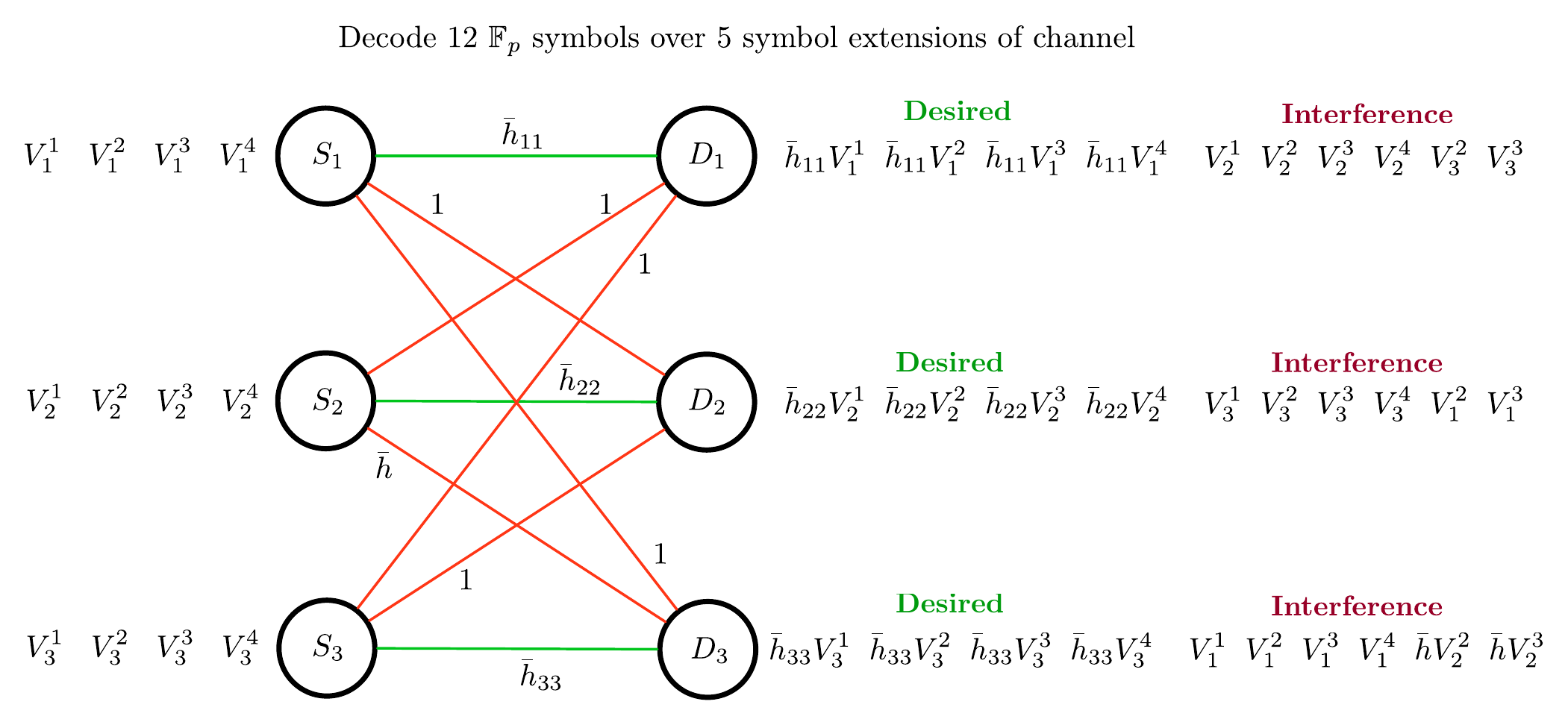}
\end{center}
\caption{3-user Interference channel over $\mathbb{F}_{p^2}$}\label{fig:int_chan_p2}
\end{figure}

We choose $V_1^1$ to be the vector of ones. Since $\bar{h}_{11},\bar{h}_{33} \notin \mathbb{F}_p$, vectors in $S_1: [\bar{h}_{11}V_1^1  \hspace{2mm}  V_1^1] $ are linearly independent and so are similar vectors in $S_3: [\bar{h}_{33}V_1^1  \hspace{2mm}  V_1^1] $. We now choose vector $V_1^2$ such that following conditions hold.

\begin{eqnarray}
\text{From $S_1$, } V_1^2\notin A&\triangleq&\left\{  \frac{(\alpha_1 \bar{h}_{11} + \alpha_2)V_1^1}{\beta_1 \bar{h}_{11}+\beta_2\frac{1}{\bar{h}}}: \alpha_1, \alpha_2, \beta_1, \beta_2\in\mathbb{F}_p, (\beta_1,\beta_2)\neq(0,0)\right\} \label{eqn_n2_21} \\
\text{From $S_2$, } V_1^2\notin B&\triangleq&\left\{  \frac{\alpha_1V_1^1}{\beta_1+\beta_2\frac{\bar{h}_{22}}{\bar{h}}}: \alpha_1,\beta_1, \beta_2\in\mathbb{F}_p, (\beta_1,\beta_2)\neq(0,0)\right\} \label{eqn_n2_22} \\
\text{From $S_3$, } V_1^2\notin C&\triangleq&\left\{  (\alpha_1 \bar{h}_{33} + \alpha_2)V_1^1: \alpha_1, \alpha_2\in\mathbb{F}_p \right\} \label{eqn_n2_23}
\end{eqnarray}
Now we note that
\begin{eqnarray}
|A| \leq \frac{(p^2-1)p^{2}}{p-1} = p^3+p^2, \hspace{10mm} |B| \leq \frac{(p^2-1)p}{p-1} = p^2+p, \hspace{10mm} |C| \leq p^2
\end{eqnarray}
\begin{eqnarray}
|A\cup B \cup C| \leq p^3 + 3p^2 + p
\end{eqnarray}
There are $p^{10}$ choices for $V_1^2\in\mathbb{F}_{p^2}^{5\times 1}$, and since 
\begin{eqnarray}
p^{10} &>& p^3 + 3p^2 + p
\end{eqnarray}
for all $p$, there exist choices for $V_1^2$ such that all 3 conditions \eqref{eqn_n2_21},\eqref{eqn_n2_22},\eqref{eqn_n2_23} hold. Choosing $V_1^2$ from those, we note that 4 columns of $S_1$ and 3 columns each of $S_2,S_3$ are linearly independent over $\mathbb{F}_p$.

Now we choose $V_2^1$ similarly such that following conditions hold
\begin{align}
V_2^1\notin A \triangleq &\{\frac{(\alpha_1 \bar{h}_{11} + \alpha_2)V_1^1 + (\alpha_3 \bar{h}_{11} + \frac{1}{\bar{h}}\alpha_4)V_1^2}{\beta_1 \bar{h}_{11}\bar{h}+\beta_2}:\alpha_1, \alpha_2, \alpha_3, \alpha_4, \beta_1, \beta_2\in\mathbb{F}_p, (\beta_1,\beta_2)\neq(0,0)\} \label{eqn_n2_31} \\ 
V_2^1\notin B \triangleq &\{\frac{\alpha_1V_1^1 + (\alpha_2 + \alpha_3\frac{\bar{h}_{22}}{\bar{h}})V_1^2}{\beta_1\bar{h}_{22}+\beta_2\bar{h}}:\alpha_1,\alpha_2,\alpha_3,\beta_1, \beta_2\in\mathbb{F}_p, (\beta_1,\beta_2)\neq(0,0)\} \label{eqn_n2_32} \\ 
V_2^1\notin C \triangleq &\{\frac{(\alpha_1 \bar{h}_{33} + \alpha_2)V_1^1 + \alpha_3V_1^2}{\bar{h}}: \alpha_1, \alpha_2, \alpha_3\in\mathbb{F}_p\} \label{eqn_n2_33}
\end{align}
Now we note that
\begin{center}
\begin{eqnarray}
|A| \leq \frac{(p^2-1)p^4}{p-1} = p^5+p^4, \hspace{10mm} |B| \leq \frac{(p^2-1)p^3}{p-1} = p^4+p^3, \hspace{10mm} |C| \leq p^3
\end{eqnarray}
\begin{eqnarray}
 |A\cup B \cup C| \leq p^5 + 2p^4 + 2p^3
\end{eqnarray}
\end{center}
There are $p^{10}$ choices for $V_2^1$, and since 
\begin{eqnarray}
p^{10} &>& p^5 + 2p^4 + 2p^3
\end{eqnarray}
for all $p$, there exist choices for $V_2^1$ such that all 3 conditions \eqref{eqn_n2_31},\eqref{eqn_n2_32},\eqref{eqn_n2_33} hold. Choosing $V_2^1$ from those, we note that 6 columns of $S_1$, 5 columns of $S_2$ and 4 columns of $S_3$ are linearly independent over $\mathbb{F}_p$.

Now we choose $V_2^2$ similarly such that following conditions hold
\begin{align}
V_2^2\notin A\triangleq & \{  (\alpha_1 \bar{h}_{11} + \alpha_2)V_1^1 + (\alpha_3 \bar{h}_{11} + \frac{1}{\bar{h}}\alpha_4)V_1^2 + (\alpha_5 \bar{h}_{11}\bar{h} + \alpha_6)V_2^1: \alpha_k \in\mathbb{F}_p, k\in\{1,\ldots,6\} \} \label{eqn_n2_41} \\ \nonumber
V_2^2\notin B\triangleq & \{ \frac{\alpha_1V_1^1 + (\alpha_2 + \alpha_3\frac{\bar{h}_{22}}{\bar{h}})V_1^2 + (\alpha_4\bar{h} + \alpha_5\bar{h}_{22})V_2^1 }{\beta_1\bar{h}_{22}+\beta_2}:  \\ & \qquad \alpha_k,\beta_1, \beta_2\in\mathbb{F}_p, k\in\{1,\ldots,5\}, (\beta_1,\beta_2)\neq(0,0) \} \label{eqn_n2_42} \\ 
V_2^2\notin C\triangleq & \{  \frac{(\alpha_1 \bar{h}_{33} + \alpha_2)V_1^1 + \alpha_3V_1^2 + \alpha_4\bar{h}V_2^1}{\beta_1\bar{h}_{33}+\beta_2\bar{h}}: \alpha_k,\beta_1,\beta_2\in\mathbb{F}_p, k\in\{1,\ldots,4\}, (\beta_1,\beta_2)\neq(0,0) \} \label{eqn_n2_43}
\end{align}
Now we note that
\begin{center}
\begin{eqnarray}
|A| \leq p^6, \hspace{10mm} |B| \leq \frac{(p^2-1)p^5}{p-1} = p^6+p^5, \hspace{10mm} |C| \leq \frac{(p^2-1)p^4}{p-1} = p^5+p^4
\end{eqnarray}
\begin{eqnarray}
 |A\cup B \cup C| \leq 2p^6 + 2p^5 + p^4
\end{eqnarray}
\end{center}
There are $p^{10}$ choices for $V_2^2$, and since 
\begin{eqnarray}
p^{10} &>& 2p^6 + 2p^5 + p^4
\end{eqnarray}
for all $p$, there exist choices for $V_2^2$ such that all 3 conditions \eqref{eqn_n2_41},\eqref{eqn_n2_42},\eqref{eqn_n2_43} hold. Choosing $V_2^2$ from those, we note that 7 columns each of $S_1, S_2$, and 6 columns of $S_3$ are linearly independent over $\mathbb{F}_p$.

Now we choose $V_3^1$ similarly such that following conditions hold
\begin{align}\nonumber
V_3^1\notin A\triangleq&\{  (\alpha_1 \bar{h}_{11} + \alpha_2)V_1^1 + (\alpha_3 \bar{h}_{11} + \frac{1}{\bar{h}}\alpha_4)V_1^2 + (\alpha_5 \bar{h}_{11}\bar{h} + \alpha_6)V_2^1 + \alpha_7 V_2^2:  \\ & \qquad \alpha_k \in\mathbb{F}_p, k\in\{1,\ldots,7\}  \} \label{eqn_n2_51} \\ \nonumber
V_3^1\notin B\triangleq&\{  \frac{\alpha_1V_1^1 + (\alpha_2 + \alpha_3\frac{\bar{h}_{22}}{\bar{h}})V_1^2 + (\alpha_4\bar{h} + \alpha_5\bar{h}_{22})V_2^1 + (\alpha_6\bar{h}_{22} + \alpha_7)V_2^2 }{\beta_1\bar{h}_{22}+\beta_2}:  \\ & \qquad \alpha_k,\beta_1, \beta_2\in\mathbb{F}_p, k\in\{1,\ldots,7\} , (\beta_1,\beta_2)\neq(0,0)\} \label{eqn_n2_52} \\ \nonumber
V_3^1\notin C\triangleq&\{  \frac{(\alpha_1 \bar{h}_{33} + \alpha_2)V_1^1 + \alpha_3V_1^2 + \alpha_4\bar{h}V_2^1 + (\alpha_5\bar{h}_{33} + \alpha_6\bar{h})V_2^2 }{\beta_1\bar{h}_{33}+\beta_2\bar{h}}:  \\ & \qquad \alpha_k,\beta_1,\beta_2\in\mathbb{F}_p, k\in\{1,\ldots,6\}, (\beta_1,\beta_2)\neq(0,0) \} \label{eqn_n2_53}
\end{align}
Now we note that
\begin{center}
\begin{eqnarray}
|A| \leq p^7, \hspace{10mm} |B| \leq \frac{(p^2-1)p^7}{p-1} = p^8+p^7, \hspace{10mm} |C| \leq \frac{(p^2-1)p^6}{p-1} = p^7+p^6
\end{eqnarray}
\begin{eqnarray}
 |A\cup B \cup C| \leq p^8 + 3p^7 + p^6
\end{eqnarray}
\end{center}
There are $p^{10}$ choices for $V_3^1$, and since 
\begin{eqnarray}
p^{10} &>& p^8 + 3p^7 + p^6
\end{eqnarray}
for all $p$, there exist choices for $V_3^1$ such that all 3 conditions \eqref{eqn_n2_51},\eqref{eqn_n2_52},\eqref{eqn_n2_53} hold. Choosing $V_3^1$ from those, we note that 8 columns each of $S_1,S_3$, and 9 columns of $S_2$ are linearly independent over $\mathbb{F}_p$.

Now we choose $V_3^2$ similarly such that following conditions hold
\begin{align} \nonumber
V_3^2\notin A\triangleq&\{\frac{(\alpha_1 \bar{h}_{11} + \alpha_2)V_1^1 + (\alpha_3 \bar{h}_{11} + \frac{1}{\bar{h}}\alpha_4)V_1^2 + (\alpha_5 \bar{h}_{11}\bar{h} + \alpha_6)V_2^1 + \alpha_7 V_2^2+ \alpha_8 V_3^1}{\beta_1\bar{h}_{11}+\beta_2}:  \\ & \qquad \alpha_k \in\mathbb{F}_p, k\in\{1,\ldots,8\}, (\beta_1,\beta_2)\neq(0,0) \} \label{eqn_n2_61} \\ \nonumber
V_3^2\notin B\triangleq&\{\alpha_1V_1^1 + (\alpha_2 + \alpha_3\frac{\bar{h}_{22}}{\bar{h}})V_1^2 + (\alpha_4\bar{h} + \alpha_5\bar{h}_{22})V_2^1 + (\alpha_6\bar{h}_{22} + \alpha_7)V_2^2+(\alpha_8\bar{h}_{22} + \alpha_9)V_3^1:  \\ & \qquad \alpha_k,\beta_1, \beta_2\in\mathbb{F}_p, k\in\{1,\ldots,9\}\} \label{eqn_n2_62} \\ \nonumber
V_3^2\notin C\triangleq&\{  \frac{(\alpha_1 \bar{h}_{33} + \alpha_2)V_1^1 + \alpha_3V_1^2 + \alpha_4\bar{h}V_2^1 + (\alpha_5\bar{h}_{33} + \alpha_6\bar{h})V_2^2 +(\alpha_7\bar{h}_{33} + \alpha_8\bar{h})V_3^1 }{\beta_1\bar{h}_{33}+\beta_2}:  \\ & \qquad \alpha_k,\beta_1,\beta_2\in\mathbb{F}_p, k\in\{1,\ldots,8\}, (\beta_1,\beta_2)\neq(0,0) \} \label{eqn_n2_63}
\end{align}
Now we note that
\begin{center}
\begin{eqnarray}
|A| \leq \frac{(p^2-1)p^8}{p-1} = p^9+p^8, \hspace{10mm} |B| \leq p^9, \hspace{10mm} |C| \leq \frac{(p^2-1)p^8}{p-1} = p^9+p^8
\end{eqnarray}
\begin{eqnarray}
 |A\cup B \cup C| \leq 3p^9 + 2p^8
\end{eqnarray}
\end{center}
There are $p^{10}$ choices for $V_3^2$, and since 
\begin{eqnarray}
p^{10} &>& 3p^9 + 2p^8
\end{eqnarray}
for $p>3$, there exist choices for $V_3^2$ such that all 3 conditions \eqref{eqn_n2_61},\eqref{eqn_n2_62},\eqref{eqn_n2_63} hold. Choosing $V_3^2$ from those, we note that all columns each of $S_1,S_2,S_3$ are linearly independent over $\mathbb{F}_p$.

Therefore, we have constructed beam forming vectors such that desired and interference signals are linearly independent at all destinations. This proves the achievability of linear-scheme capacity of $\frac{6}{5}$ for 3-user interference channel over $\mathbb{F}_{p^2}$ for all $p>3$ when the specified conditions are met. For p=2 and p=3, we are able to exhaustively solve all possible cases numerically using MATLAB, completing the achievability proof of sum-rate $\frac{6}{5}$ for channel over $\mathbb{F}_{p^2}$ for all $p$ under the conditions of Theorem \ref{theorem:int_n2}. 

The conditions can be also re-written in terms of the original channels as follows.
\begin{eqnarray} \nonumber
\bar{h}_{11} = \frac{h_{11}h_{23}}{h_{13}h_{21}} \notin \mathbb{F}_p, \hspace{8mm} {\bar{h}}{\bar{h}}_{11} = \frac{h_{11}h_{32}}{h_{12}h_{31}} \notin \mathbb{F}_p \\ \nonumber
\bar{h}_{22} = \frac{h_{22}h_{13}}{h_{23}h_{12}} \notin \mathbb{F}_p, \hspace{8mm}  \frac{\bar{h}}{{\bar{h}}_{22}} = \frac{h_{21}h_{32}}{h_{22}h_{31}} \notin \mathbb{F}_p \\ \nonumber
\bar{h}_{33} = \frac{h_{33}h_{21}}{h_{31}h_{23}} \notin \mathbb{F}_p, \hspace{8mm} \frac{\bar{h}}{{\bar{h}}_{33}} = \frac{h_{32}h_{13}}{h_{33}h_{12}} \notin \mathbb{F}_p
\end{eqnarray}
$\hfill\QED$
\end{proof}
\\

{\it Remark 5:} Note that these 6 conditions are equivalent to the 6 conditions on the phase differences between channel coefficients in the asymmetric complex signing scheme for wireless networks, as described in \cite{Cadambe_Jafar_Wang} to achieve DoF of $\frac{6}{5}$. 
\\

{\it Remark 6:} Each of the direct channels satisfy $\bar{h}_{ii} \notin \mathbb{F}_p, i\in \{1,2,3\}$
The fraction of channel realizations for which direct channels satisfy the 3 conditions is at least 
\begin{eqnarray} 
(\frac{p^2-p}{p^2})^3 = (1-\frac{1}{p})^3 \to \text{ 1 for large } p 
\end{eqnarray}
Further cross channel $\bar{h}$ should satisfy the conditions $\bar{h} \neq \frac{\alpha}{\bar{h}_{11}}, \bar{h} \neq \beta\bar{h}_{22}, \bar{h} \neq \gamma\bar{h}_{33}$ for $\alpha,\beta,\gamma \in \mathbb{F}_p$. There are atmost $3p$ channels such that one of these 3 conditions on $\bar{h}$ is violated. Hence there are at least $p^2-3p$ valid channel realizations for $\bar{h}$ for $p>3$.  Putting everything together, the fraction of all channels for which the scheme works for $p>3$ is at least
\begin{eqnarray} 
 (1-\frac{1}{p})^3 (\frac{p^2-3p}{p^2}) =  (1-\frac{1}{p})^3 (1-\frac{3}{p}) \to \text{ 1 for large }p 
\end{eqnarray}
\subsection{Linear outer bound} \label{lin_outer_bound_sec}
In this section, we will prove  the linear outer bounds. The proof follows along the lines of \cite{Bresler_Tse_Diversity} by showing that the alignment depth can be at most $D$, which is a function of channel diversity (in case of finite fields, $n$). 

\subsubsection{Linear outer bound over $ \mathbb{F}_{p^n}, n=2l+1$}
\begin{lemma} \label{outer_lemma1}
Alignment depth is at most $D = 2n-\lfloor\frac{n}{2}\rfloor-1$ for the normalized 3-user interference channel, wherein channels $\bar{h},\bar{h}_{kk}\in \mathbb{F}_{p^n}$ for odd $n=2l+1$ and satisfy
\begin{eqnarray}
\bar{h}_{11}\notin A&\triangleq&\left\{\frac{\alpha_0+\alpha_1\bar{h}+\ldots+\alpha_{l-1}\bar{h}^{l-1}}{\beta_0+\beta_1\bar{h}+\ldots+\beta_l\bar{h}^l}: \alpha_k,\beta_m\in\mathbb{F}_p, (\beta_0,\ldots,\beta_l)\neq(0,\ldots,0)\right\} \\
\bar{h}_{22}\notin B&\triangleq&\left\{\frac{\alpha_0+\alpha_1\bar{h}+\ldots+\alpha_l\bar{h}^l}{\beta_0+\beta_1\bar{h}+\ldots+\beta_{l-1}\bar{h}^{l-1}}: \alpha_k,\beta_m\in\mathbb{F}_p, (\beta_0,\ldots,\beta_{l-1})\neq(0,\ldots,0)\right\} \\
\bar{h}_{33}\notin C&\triangleq&\left\{\frac{\alpha_0+\alpha_1\bar{h}+\ldots+\alpha_l\bar{h}^l}{\beta_0+\beta_1\bar{h}+\ldots+\beta_{l-1}\bar{h}^{l-1}}: \alpha_k,\beta_m\in\mathbb{F}_p, (\beta_0,\ldots,\beta_{l-1})\neq(0,\ldots,0)\right\} \\
&&\beta_l\bar{h}^l+\ldots+\beta_1\bar{h}+\beta_0 \neq 0 : \beta_0,\ldots,\beta_l\in\mathbb{F}_p, (\beta_0,\ldots,\beta_l)\neq(0,\ldots,0)\label{eq:aaa}
\end{eqnarray}
\end{lemma}

\begin{proof}
Let us consider the normalized channel as described in section \ref{interf_norm_sec} for odd $n=2l+1$, and at source 1, denote a vector of dimension $m\times 1$ as $V$ with entries from $\mathbb{F}_{p^n}$. Since this is a converse proof, we assume that the desired symbols can be decoded at all the destinations. Here $m$ denotes the number of symbol extensions of the channel. This vector of source 1 needs to be aligned with a vector from source 3 at destination 2, we can denote the vector at source 3 as $\gamma_1V$ with $\gamma_1 \in \mathbb{F}_p$.  Vector $\gamma_1V$ aligns with a vector from source 2 at destination 1, say $\beta_1V$ with $\beta_1 \in \mathbb{F}_p$. Vector $\beta_1V$ aligns with a vector from source 1 at destination 3, say $\alpha_1\bar{h}V$ with $\alpha_1 \in \mathbb{F}_p$. So far, alignment chain length can be seen to be 4, and such an alignment chain can be extended upto length D when operating in field of order $p^n$. With $n=2l+1$ this results in source 1 using $l+1$ vectors, and sources 2 and 3 using $l$ vectors each such that the alignment chain length is $D=3l+1$. 
Then the vectors chosen so far at the 3 sources can be represented as
\begin{eqnarray} 
V_1 = [\alpha_{l}\bar{h}^{l}V \hspace{5mm} \alpha_{l-1}\bar{h}^{l-1}V \hspace{5mm} \ldots \hspace{5mm} \alpha_{1}\bar{h}V \hspace{5mm} V] \\
V_2 = [\beta_{l}\bar{h}^{l-1}V \hspace{5mm} \beta_{l-1}\bar{h}^{l-2}V \hspace{5mm} \ldots \hspace{5mm} \beta_{2}\bar{h}V \hspace{5mm} \beta_{1}V] \\
V_3 = [\gamma_{l}\bar{h}^{l-1}V \hspace{5mm} \gamma_{l-1}\bar{h}^{l-2}V \hspace{5mm} \ldots \hspace{5mm} \gamma_{2}\bar{h}V \hspace{5mm} \gamma_{1}V]
\end{eqnarray}
wherein $V$ is an $m \times 1$ vector with entries from $\mathbb{F}_{p^n}$ and $\alpha_i,\beta_i,\gamma_i \in \mathbb{F}_p$, $\forall i \in \{1,\ldots,l\}$. 
We will now argue that alignment chain length cannot be extended beyond $D$. Suppose on the contrary, alignment chain length was greater than $D$, say $D+1$. Then without loss of generality, we can choose additional vector at source 3 such that at destination 2, it aligns with the vector $\alpha_l\bar{h}^lV$ used at source 1. This additional vector at source 3 can be represented as $\gamma_{l+1}\bar{h}^lV$. Then the vectors sent by source 3 can be represented as
\begin{eqnarray} 
\bar{V}_3 = [\gamma_{l+1}\bar{h}^lV \hspace{5mm} \gamma_{l}\bar{h}^{l-1}V \hspace{5mm} \gamma_{l-1}\bar{h}^{l-2}V \hspace{5mm} \ldots \hspace{5mm} \gamma_{2}\bar{h}V \hspace{5mm} \gamma_{1}V]
\end{eqnarray}

Let us consider the signal space at destination 1, $S_1 = [\bar{h}_{11}V_1 \hspace{3mm} V_2 \hspace{3mm} \bar{V}_3]$. Since $l$ vectors from source 3 align with $l$ vectors from source 2, we can denote the signal space as $S_1 = [\bar{h}_{11}V_1 \hspace{3mm} V_2 \hspace{3mm} \gamma_{l+1}\bar{h}^lV]$. Now we claim that $\bar{h}_{11}V_1$ and $V_2$ spans the channel space, since all vectors are linearly independent. 
\begin{eqnarray} \nonumber
[\bar{h}_{11}V_1 \hspace{2mm} V_2] = [\alpha_{l}\bar{h}_{11}\bar{h}^{l}V \hspace{2mm} \alpha_{l-1}\bar{h}_{11}\bar{h}^{l-1}V \hspace{2mm} \ldots \hspace{2mm} \alpha_{1}\bar{h}_{11}\bar{h}V \hspace{2mm} \bar{h}_{11}V \hspace{2mm} \beta_{l}\bar{h}^{l-1}V \hspace{2mm} \beta_{l-1}\bar{h}^{l-2}V \hspace{2mm} \ldots \hspace{2mm} \beta_{2}\bar{h}V \hspace{2mm} \beta_{1}V] 
\end{eqnarray}
It can be noted that columns of above matrix are linearly independent when all entries listed below are linearly independent, since $V$ is scaled by different powers of $\bar{h}, \bar{h}_{11}$ and other coefficients.
\begin{eqnarray} \nonumber
[\alpha_{l}\bar{h}_{11}\bar{h}^{l} \hspace{2mm} \alpha_{l-1}\bar{h}_{11}\bar{h}^{l-1} \hspace{2mm} \ldots \hspace{2mm} \alpha_{1}\bar{h}_{11}\bar{h} \hspace{2mm} \bar{h}_{11} \hspace{2mm} \beta_{l}\bar{h}^{l-1} \hspace{2mm} \beta_{l-1}\bar{h}^{l-2} \hspace{2mm} \ldots \hspace{2mm} \beta_{2}\bar{h} \hspace{2mm} \beta_{1}] 
\end{eqnarray} 
This is  true when following conditions on $\bar{h},\bar{h}_{11}$ are met. 
\begin{eqnarray}
\bar{h}_{11}\notin A&\triangleq&\left\{\frac{\alpha_0+\alpha_1\bar{h}+\ldots+\alpha_{l-1}\bar{h}^{l-1}}{\beta_0+\beta_1\bar{h}+\ldots+\beta_l\bar{h}^l}: \alpha_k,\beta_m\in\mathbb{F}_p, (\beta_0,\ldots,\beta_l)\neq(0,\ldots,0)\right\} \\
&&\beta_l\bar{h}^l+\ldots+\beta_1\bar{h}+\beta_0 \neq 0 : \beta_0,\ldots,\beta_l\in\mathbb{F}_p, (\beta_0,\ldots,\beta_l)\neq(0,\ldots,0)
\end{eqnarray}
Since $n=2l+1$ columns of $[\bar{h}_{11}V_1 \hspace{2mm} V_2]$ are linearly independent,  additional vector chosen $\gamma_{l+1}\bar{h}^lV$ must lie in span of $[\bar{h}_{11}V_1 \hspace{2mm} V_2]$.
It cannot lie in the space spanned by $V_2$ because that would contradict (\ref{eq:aaa}). But if it does not lie in the space spanned by $V_2$ then the desired signal space $\bar{h}_{11}V_1$ cannot be resolvable from interference. This is a contradiction, since in the converse we assume that the desired signal is resolvable from interference. Therefore additional vector $\gamma_{l+1}\bar{h}^lV$ cannot be chosen at source 3 such that it aligns at destination 1, i.e., alignment depth cannot be greater than $D=3l+1$. This is illustrated in Fig. \ref{fig:align_depth}. Similarly alignment chains originating at other sources and ending at other destinations can be shown to be of depth not greater than D. Consolidating the linear independence conditions for all such chains, we note that alignment depth is at most D for channels satisfying following conditions.
\begin{eqnarray}
\bar{h}_{11}\notin A&\triangleq&\left\{\frac{\alpha_0+\alpha_1\bar{h}+\ldots+\alpha_{l-1}\bar{h}^{l-1}}{\beta_0+\beta_1\bar{h}+\ldots+\beta_l\bar{h}^l}: \alpha_k,\beta_m\in\mathbb{F}_p, (\beta_0,\ldots,\beta_l)\neq(0,\ldots,0)\right\} \\
\bar{h}_{22}\notin B&\triangleq&\left\{\frac{\alpha_0+\alpha_1\bar{h}+\ldots+\alpha_l\bar{h}^l}{\beta_0+\beta_1\bar{h}+\ldots+\beta_{l-1}\bar{h}^{l-1}}: \alpha_k,\beta_m\in\mathbb{F}_p, (\beta_0,\ldots,\beta_{l-1})\neq(0,\ldots,0)\right\} \\
\bar{h}_{33}\notin C&\triangleq&\left\{\frac{\alpha_0+\alpha_1\bar{h}+\ldots+\alpha_l\bar{h}^l}{\beta_0+\beta_1\bar{h}+\ldots+\beta_{l-1}\bar{h}^{l-1}}: \alpha_k,\beta_m\in\mathbb{F}_p, (\beta_0,\ldots,\beta_{l-1})\neq(0,\ldots,0)\right\} \\
&&\beta_l\bar{h}^l+\ldots+\beta_1\bar{h}+\beta_0 \neq 0 : \beta_0,\ldots,\beta_l\in\mathbb{F}_p, (\beta_0,\ldots,\beta_l)\neq(0,\ldots,0)
\end{eqnarray}
Thus, we have proved Lemma \ref{outer_lemma1}.
\end{proof}

\begin{figure}[h]
\begin{center}
\includegraphics[scale=0.7]{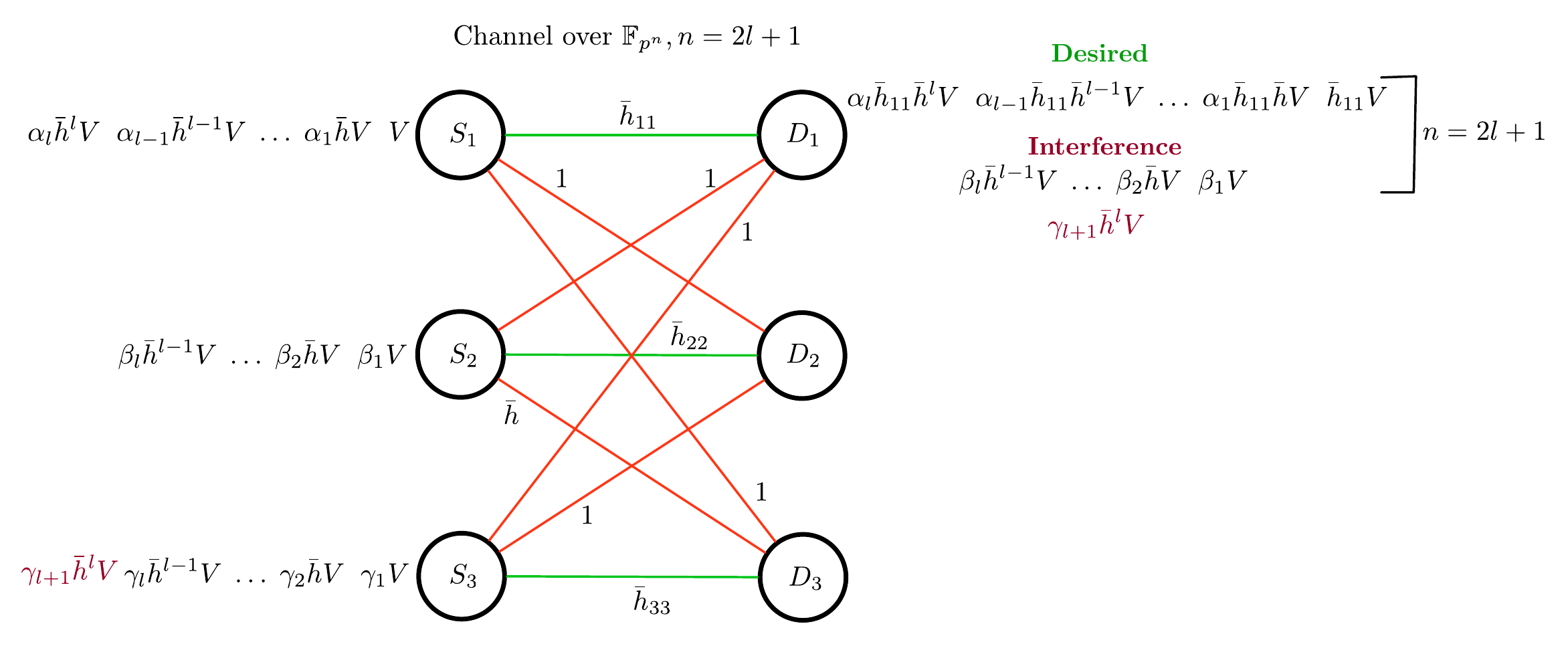}
\end{center}
\caption{Alignment depth in 3-user Interference channel}\label{fig:align_depth}
\end{figure}

We now show the outer bound on linear-scheme capacity for 3-user interference channel to be $\frac{3D}{2D+1}$. The proof of this part is almost identical to that in \cite{Bresler_Tse_Diversity}, so it is summarized  only for the sake of completeness.
\begin{theorem}\label{theorem:int_outer}
For the 3-user interference channel over $\mathbb{F}_{p^n}$,  outer bound on linear-scheme capacity is given by $\frac{3D}{2D+1}$, with $D = 2n-\lfloor\frac{n}{2}\rfloor-1$ for odd $n=2l+1$ wherein channels satisfy the following conditions
\begin{eqnarray}
\bar{h}_{11}\notin A&\triangleq&\left\{\frac{\alpha_0+\alpha_1\bar{h}+\ldots+\alpha_{l-1}\bar{h}^{l-1}}{\beta_0+\beta_1\bar{h}+\ldots+\beta_l\bar{h}^l}: \alpha_k,\beta_m\in\mathbb{F}_p, (\beta_0,\ldots,\beta_l)\neq(0,\ldots,0)\right\} \\
\bar{h}_{22}\notin B&\triangleq&\left\{\frac{\alpha_0+\alpha_1\bar{h}+\ldots+\alpha_l\bar{h}^l}{\beta_0+\beta_1\bar{h}+\ldots+\beta_{l-1}\bar{h}^{l-1}}: \alpha_k,\beta_m\in\mathbb{F}_p, (\beta_0,\ldots,\beta_{l-1})\neq(0,\ldots,0)\right\} \\
\bar{h}_{33}\notin C&\triangleq&\left\{\frac{\alpha_0+\alpha_1\bar{h}+\ldots+\alpha_l\bar{h}^l}{\beta_0+\beta_1\bar{h}+\ldots+\beta_{l-1}\bar{h}^{l-1}}: \alpha_k,\beta_m\in\mathbb{F}_p, (\beta_0,\ldots,\beta_{l-1})\neq(0,\ldots,0)\right\} \\
&&\beta_l\bar{h}^l+\ldots+\beta_1\bar{h}+\beta_0 \neq 0 : \beta_0,\ldots,\beta_l\in\mathbb{F}_p, (\beta_0,\ldots,\beta_l)\neq(0,\ldots,0)
\end{eqnarray}
\end{theorem}
\begin{proof}
Let $V_{i\uparrow k}$ denote the signal space of user $i$ (part of $V_i$) aligned to depth $k+1$ and $d_i = \dim(V_i), d_{i\uparrow k} = \dim(V_{i\uparrow k})$. Lemma 8 of \cite{Bresler_Tse_Diversity} follows since we have finite dimensional subspaces, i.e., $d_{i\uparrow k} \geq d_{i\uparrow k+a}+d_{i-b\uparrow k+b}-d_{i-b\uparrow k+a+b}$. For $a=-1,b=-1$, we have
\begin{eqnarray} \label{outer_eq1}
d_{i\uparrow k} \geq d_{i\uparrow k-1}+d_{i+1\uparrow k-1}-d_{i+1\uparrow k-2}
\end{eqnarray}
Since alignment depth is at most D (Lemma \ref{outer_lemma1}), $V_{i\uparrow D}=\{0\}$ for each i, and so similar to lemma 9 of \cite{Bresler_Tse_Diversity}, we have
\begin{eqnarray} \label{outer_eq2}
d_i \geq d_{i-1\uparrow 1} + d_{i\uparrow D-1}
\end{eqnarray}
Let us denote $c_k = \sum_{i=1}^{3} d_{i\uparrow k}$.  Then using \ref{outer_eq1}, we have $c_k \geq 2c_{k-1}-c_{k-2}$. Using induction, it can be deduced that $c_k \geq ic_{k-i+1}-(i-1)c_{k-i}$. For $i=k=D-1$, we have
\begin{eqnarray} \label{outer_eq3}
(D-2)c_0 \geq (D-1)c_{1}-c_{D-1}
\end{eqnarray}
Using \ref{outer_eq2}, it can be shown that $c_0 \geq c_1 + c_{D-1}$. Combining with \ref{outer_eq3}, we have $(D-1)c_0 \geq Dc_1$. Let total dimension at each destination be denoted by $N = mn$ where $m$ symbol extensions of the channel is considered with channels from $\mathbb{F}_{p^n}$.
Since interference span must be linearly independent of desired signal, and considering N dimensions at destination 1, we have
\begin{eqnarray} \label{outer_eq4}
\text{Destination 1:}& \hspace{5mm} \dim(\bar{h}_{11}V_1 + V_2 + V_3)  = &d_1 + d_2 + d_3 - d_{2\uparrow 1} \leq N \\
\text{Destination 2:}& \hspace{5mm} \dim(V_1 + \bar{h}_{22}V_2 + V_3)  = &d_1 + d_2 + d_3 - d_{3\uparrow 1} \leq N \\
\text{Destination 3:} &\hspace{5mm} \dim(V_1 + \bar{h}V_2 + \bar{h}_{33}V_3)  =  &d_1 + d_2 + d_3 - d_{1\uparrow 1} \leq N
\end{eqnarray}
Adding above inequalities and using $(D-1)c_0 \geq Dc_1$, we can deduce as in \cite{Bresler_Tse_Diversity} that 
\begin{eqnarray} 
\frac{d_1 + d_2 + d_3}{N} \leq \frac{3D}{2D+1}
\end{eqnarray}
Thus we have proved the outer bound on linear-scheme capacity for 3-user interference channel over $\mathbb{F}_{p^n}$ with channels satisfying aforementioned linear independence constraints. $\hfill\QED$
\end{proof}

\section{Conclusion}\label{section:conclusion} 
Linear capacity results are explored for the X channel and the 3 user interference channel over the finite field $\mathbb{F}_{p^n}$, by translating precoding based interference alignment schemes from corresponding DoF results for the wireless setting. The main insight is that the finite field $\mathbb{F}_{p^n}$ can be viewed as analogous to diagonal $n\times n$ wireless channels with diversity $n$.  This insight appears to be broadly true for linear precoding based schemes. While the linear capacity is fully characterized, the information theoretic capacity remains open for finite field networks over $\mathbb{F}_p$, i.e., for $n=1$, where diversity is only 1. We expect that signal level alignment schemes and combinatorial outer bound arguments such as those presented in \cite{Etkin_Ordentlich} should be useful in these cases.

\section{Appendix}\label{section:appendix} 

\subsection{Appendix I - X channel over $\mathbb{F}_{p^3}$ : Alternate proof}

Here we discuss an alternate proof for achievability of sum rate of $\frac{4}{3}$ for 2-user X channel over $\mathbb{F}_{p^3}$. Let us first state a lemma. \\

\noindent
\emph{Definition 1:} Let $\mathbb{F}_{p^n}$ be the field extension of $\mathbb{F}_{p}$, and $p(x)$ the ring of polynomials in $x$ over $\mathbb{F}_{p}$. The minimal polynomial of $h \in \mathbb{F}_{p^n}$ is the monic polynomial of least degree among all polynomials such that $p(h) = 0$.

\begin{lemma}\label{lemma:prime_lemma}
For $h \in \mathbb{F}_{p^n}$ with both $p, n$ being prime, $1, h, h^2, \ldots, h^{n-1}$ are linearly independent over $\mathbb{F}_p$ if and only if $h \notin \mathbb{F}_p$
\begin{eqnarray}
\alpha_0 + \alpha_1h + \alpha_2h^2 + \cdots + \alpha_{n-1}h^{n-1} \neq 0  \iff h \notin \mathbb{F}_p
\end{eqnarray}
wherein $\alpha_k \in \mathbb{F}_p, k \in \{0,1,\ldots,n-1\}$. 
\end{lemma}
\begin{proof}
When $1, h, h^2, \ldots, h^{n-1}$ are linearly independent over $\mathbb{F}_p$, it is trivial to note that $h \notin \mathbb{F}_p$.
For the other direction, let us consider $h \notin \mathbb{F}_p$. Suppose on the contrary, $1, h, h^2, \ldots, h^{n-1}$ were linearly dependent over $\mathbb{F}_p$, then there exist $\alpha_k \in \mathbb{F}_p, k \in \{0,1,\ldots,n-1\}$ such that
\begin{eqnarray}
\alpha_0 + \alpha_1h + \alpha_2h^2 + \cdots + \alpha_{n-1}h^{n-1} = 0 \label{minimal_poly}
\end{eqnarray}
If \eqref{minimal_poly} holds, one can identify a monic irreducible polynomial of degree $k \leq n-1$, which is then the minimal polynomial of $h \in \mathbb{F}_{p^n}$ according to definition 1. From Theorem 3.33 of \cite{FiniteFields}, we note that degree of minimal polynomial of an element $h \in \mathbb{F}_{p^n}$ (in our case, degree is $k$), divides $n$. As a result, since we consider only prime $n$, $k>1$ is not possible. $k$ cannot be 1 since $h \notin \mathbb{F}_p$. Hence \eqref{minimal_poly} is a contradiction, and so $1, h, h^2, \ldots, h^{n-1}$ are linearly independent over $\mathbb{F}_p$ when $h \notin \mathbb{F}_p$. $\hfill\QED$
\end{proof} 
\\

\noindent
\emph{Achievability for Theorem \ref{theorem:x_n} over $\mathbb{F}_{p^3}$:} \\
For the fully connected X channel over $\mathbb{F}_{p^3}$ if $h = \frac{h_{12}h_{21}}{h_{11}h_{22}} \notin \mathbb{F}_p$,
then $C=C_{\mbox{\tiny linear}}=\frac{4}{3}$. in units of $\mathbb{F}_{p^3}$ symbols per channel use. \\

Alternate 
\begin{proof}
Like in section \ref{sec:p3}, for the 2-user X channel, the received symbols after precoding using beamforming vectors $v_{ji} \in \mathbb{F}_{p^3}$ for input symbols $x_{ji} \in \mathbb{F}_p$, are expressed as
\begin{eqnarray*}
{ y}_1&=&{v_{11}{x}_{11}}+{v_{12}{x}_{12}}+{v_{22}{x}_{22}}+{v_{21}{x}_{21}} \\
{ y}_2&=&{v_{22}{x}_{22}}+{ {h}}{v_{21}{x}_{21}}+{ {h}}{v_{11}{x}_{11}}+{v_{12}{x}_{12}}
\end{eqnarray*}
wherein $h, y_j \in \mathbb{F}_{p^3}$.
Interference is aligned at each destination along one dimension by setting ${{v}_{22}} = {{v}_{21}}$ and ${{v}_{12}} = { h}{{ v}_{11}}$. At the destinations, signal spaces are represented using matrices $S_1$ and $S_2$.
\begin{eqnarray}
S_1 = [{ v}_{11} \hspace{3mm} { v}_{12} \hspace{3mm} { v}_{21}] = [{ v}_{11} \hspace{3mm} {h}}{{ v}_{11} \hspace{3mm} { v}_{21}] \\
S_2 = [{ v}_{22} \hspace{3mm} {{h}}{ v}_{21} \hspace{3mm} {v}_{12}] = [{v}_{21} \hspace{3mm} {{h}}{v}_{21} \hspace{3mm} {h}}{{ v}_{11} ]
\end{eqnarray}

\indent  Let us choose $v_{21}=1, v_{11}=h$.  Then $S_1$, $S_2$ are identical, given by
\begin{eqnarray}
S_1 = S_2 = [1 \hspace{3mm} h \hspace{3mm} {h^2}]
\end{eqnarray}
Using Lemma \ref{lemma:prime_lemma}, it follows that for all $h \in \mathbb{F}_{p^3}$, field elements $1, h, h^2$ are linearly independent over $\mathbb{F}_p$ if $h \notin \mathbb{F}_p$. Hence, desired and interfering symbols are linearly independent over $\mathbb{F}_p$ when $h \notin \mathbb{F}_p$.

Thus, we have proved the achievability of rate $\frac{1}{3}$ per message, and a sum-rate of $\frac{4}{3}$, which matches the capacity outer bound.  $\hfill\QED$
\end{proof}

\subsection{Appendix II - X Channel over $\mathbb{F}_{p^2}$}
$\mathbb{F}_{p^2}$  can be viewed as a 2-dimensional vector space over subfield $\mathbb{F}_{p}$, much like the field of complex numbers can be viewed as a 2-dimensional vector space over reals ($\mathbb{R}$), which is also the essential idea behind the asymmetric complex signaling scheme used in \cite{Cadambe_Jafar_Wang} to achieve $4/3$ DoF for the constant SISO wireless X channel with complex coefficients. We can represent each element of $\mathbb{F}_{p^2}$ as
\begin{eqnarray}
z = x+y\sqrt{c} \hspace{3mm} \text{or}\hspace{3mm} x+ys
\end{eqnarray}
wherein $z \in \mathbb{F}_{p^2}$, $x,y \in \mathbb{F}_{p}$ and $c$ is a quadratic non-residue (an element that does not have a square root in $\mathbb{F}_{p}$) similar to $-1$ (which does not have a square root over reals) in the field of  complex numbers. ($s=\sqrt{c}\equiv j$). 

For example, consider $\mathbb{F}_{3^2}$ with prime subfield $\mathbb{F}_{3}$ which has $c=-1$(mod 3) $=2$ as the quadratic non-residue, since $\sqrt{2}$ does not exist in $\mathbb{F}_{3}$.    
Field $\mathbb{F}_{3^2}$ contains $9$ elements and every element $a_1s+a_0$ can be written in a vector notation with coefficients $[a_1;a_0]$ wherein $a_1,a_0 \in \mathbb{F}_{3} =\{0,1,2\}$ and assigned a scalar integer label $\{0,1,\ldots,8\}$ as $3a_1+a_o$. For example, the field element labeled $a=7$ can be represented as $[2~;~1]$ in vector notation, as $2s+1$ in polynomial notation, or as $2\sqrt{2}+1$ in the quadratic non-residue notation. \\
Here,  product with ${h}$ can be represented using a $2\times 2$ linear transformation (MIMO equivalent). Let ${h} = h_1s+h_0, \hspace{3mm} {x} =  x_1s+x_0$ and $h_i,x_i \in \mathbb{F}_{3}$. Then the product ${y}={hx} \in \mathbb{F}_{3^2}$ can be written as
\begin{eqnarray}
{y=hx} = (h_1s+h_0)(x_1s+x_0) =
s^2(h_1x_1)+s(h_1x_0+h_0x_1)+(h_0x_0)
\end{eqnarray}
and in vector notation as
\begin{eqnarray}
\mathbf{{y}} = \mathbf{{H}{x}} = 
\begin{bmatrix}
h_0 & 2h_1  \\
h_1 & h_0
\end{bmatrix}
\begin{bmatrix}
x_1 \\
x_0
\end{bmatrix}
\end{eqnarray}
wherein $\mathbf{{x}} \in \mathbb{F}_{3}^{2\times 1}$ and $\mathbf{{H}} \in \mathbb{F}_{3}^{2\times 2}$. 
It can be noted that above $2\times 2$ linear transformation is equivalent to complex multiplication and stacking the resulting real and imaginary parts in a $2 \times 1$ vector. Note that $\mathbb{F}_{2}$ is a special case because there is no quadratic non-residue, where the scheme is equivalent to having a $2 \times 2$ MIMO channel, but not to asymmetric complex signaling. \\

\noindent
\textbf{Achievability proof for X-channel over $\mathbb{F}_{p^2}$} \\
\begin{proof} 
Now we prove that sum rate of $\frac{4}{3}$ is achievable (part of Theorem \ref{theorem:x_n} proof) for 2-user X-channel over $\mathbb{F}_{p^2}.$
We consider the X channel with $3$ symbol extensions, wherein we can represent the channel between source $i$ and destination $j$ as $H_{ji}=h_{ji}I_3$ where $I_3$ is the $3 \times 3$ identity matrix and $h_{ji}$ is the scalar channel coefficient from $\mathbb{F}_{p^2}$. 
The inputs $x_{ji}$ are chosen from $\mathbb{F}_{p}$ and outputs $y_j$ over $\mathbb{F}_{p^2}$ and three channel uses can be seen as a $6$ dimensional vector space over $\mathbb{F}_{p}$ within which $4$ desired symbols and $4$ interference symbols are present at each destination. In order to achieve  capacity, interference should be aligned within 2 dimensions at each destination. To this end, we will construct beamforming vectors at each source such that interference is aligned. Received symbols at the destinations, in vector notation, are given by
\begin{eqnarray*}
{\bf  Y}_1&=&{{\bf V}_{11}{X}_{11}}+{{\bf V}_{12}{X}_{12}}+{{\bf V}_{22}{X}_{22}}+{{\bf V}_{21}{X}_{21}} \\
{ \bf Y}_2&=&{{\bf V}_{22}{X}_{22}}+{ \mathbf{\bar H}}{{\bf V}_{21}{X}_{21}}+{ \mathbf{\bar H}}{{\bf V}_{11}{X}_{11}}+{{\bf V}_{12}{X}_{12}}
\end{eqnarray*}
Here ${\bf Y}_j \in \mathbb{F}_{p}^{6\times 1}, {\bf V}_{ji} \in \mathbb{F}_{p}^{6\times 2}$, and $X_{ji} \in \mathbb{F}_{p}^{2\times 1}$ represents the symbols sent by source $i$ for destination $j$.  $\bar{\mathbf{H}} \in \mathbb{F}_{p}^{6\times 6}$ is the linear transformation which is equivalent to multiplication by $h \in \mathbb{F}_{p^2}$.
Over 3 symbol extensions of the channel, linear transformation $\bar{\mathbf{H}}$ for $p>2$, is given by
\begin{eqnarray}
\bar{\mathbf{H}} =
\begin{bmatrix}
h_0 & 0 & 0 & ch_1 & 0 & 0 \\
0 & h_0 & 0 & 0 & ch_1 & 0 \\
0 & 0 & h_0 & 0 & 0 & ch_1 \\
h_1 & 0 & 0 & h_0 & 0 & 0 \\
0 & h_1 & 0 & 0 & h_0 & 0 \\
0 & 0 & h_1 & 0 & 0 & h_0
\end{bmatrix}
\end{eqnarray}
\\
\indent Note that above matrix is the 3-symbol extension of the linear transformation $H = [h_0 \hspace{2mm} ch_1; h_1 \hspace{2mm} h_0]$. Here, $c$ is the quadratic non-residue which exists for all $p>2$. In order to achieve sum rate of $\frac{4}{3}$, interference should be aligned at both destinations:
\begin{eqnarray}
\mbox{span}({{\bf V}_{22}}) \equiv \mbox{span}({{\bf V}_{21}}) \hspace{3mm} \& \hspace{3mm} 
\mbox{span}({{\bf V}_{12}}) \equiv \mbox{span}({ \bar{\mathbf{H}}}{{\bf V}_{11}}) 
\end{eqnarray}
For every choice of ${\bf V}_{21}, {\bf V}_{11}$, we set 
\begin{eqnarray}
{{\bf V}_{22}} = {{\bf V}_{21}} \hspace{6mm} \& \hspace{6mm} 
{{\bf V}_{12}} = { \bar{\mathbf{H}}}{{\bf V}_{11}}
\end{eqnarray}
At each destination, the two desired signal vectors and the aligned interference vector can be represented using $6\times 6$ matrices, $S_1$ and $S_2$.
\begin{eqnarray}
S_1 = [{\bf V}_{11} \hspace{3mm} {\bf V}_{12} \hspace{3mm} {\bf V}_{21}] = [{\bf V}_{11} \hspace{3mm} \bar{\mathbf{H}}}{{\bf V}_{11} \hspace{3mm} {\bf V}_{21}] \\
S_2 = [{\bf V}_{22} \hspace{3mm} \bar{\mathbf{H}}{\bf V}_{21} \hspace{3mm} {\bf V}_{12}] = [{\bf V}_{21} \hspace{3mm} \bar{\mathbf{H}}{\bf V}_{21} \hspace{3mm} \bar{\mathbf{H}}{\bf V}_{11}]
\end{eqnarray}
\noindent We now choose ${\bf V}_{11}$ and ${\bf V}_{21}$ as follows.
\begin{eqnarray}
{\bf V}_{11} = 
\begin{bmatrix}
1 & 1 \\
1 & 0 \\
0 & 0 \\
1 & 1 \\
0 & 1 \\
0 & 1
\end{bmatrix} \hspace{40mm}
{\bf V}_{21} = 
\begin{bmatrix}
1 & 0 \\
1 & 0 \\
1 & 1 \\
0 & 0 \\
1 & 1 \\
1 & 1
\end{bmatrix}
\end{eqnarray}
\\
With above choice of beamforming matrices, matrices $S_1$ and $S_2$ can be written as \\
\small
\[S_1= \left[\begin{array}{cccccc}
1 & 1 & h_0+ch_1 & h_0+ch_1 & 1 & 0 \\
1 & 0 & h_0 & ch_1 & 1 & 0 \\
0 & 0 & 0 & ch_1 & 1 & 1 \\
1 & 1 & h_0+h_1 & h_0+h_1 & 0 & 0 \\
0 & 1 & h_1 & h_0 & 1 & 1 \\
0 & 1 & 0 & h_0 & 1 & 1
\end{array} \right] \hspace{6mm}
S_2= \left[\begin{array}{cccccc}
h_0 & 0 & h_0+ch_1 & h_0+ch_1 & 1 & 0 \\
h_0+ch_1 & ch_1 & h_0 & ch_1 & 1 & 0 \\
h_0+ch_1 & h_0+ch_1 & 0 & ch_1 & 1 & 1 \\
h_1 & 0 & h_0+h_1 & h_0+h_1 & 0 & 0 \\
h_0+h_1 & h_0 & h_1 & h_0 & 1 & 1 \\
h_0+h_1 & h_0+h_1 & 0 & h_0 & 1 & 1
\end{array} \right]
\]
\normalsize
\\
Evaluating determinant of the above two matrices, we get the following polynomials
\begin{eqnarray}
|S_1| = ch_1^2\\
|S_2| = h_1^2(ch_1^2-h_0^2)
\end{eqnarray}
Determinant of matrix $S_1$ is non-zero since $h_1 \neq 0$ when $h\notin \mathbb{F}_{p}$, and a non-zero quadratic non-residue exists for all $p>2$, i.e., $c \neq 0$. When considering determinant polynomial of matrix $S_2$, term $h_1^2 \neq 0$ when $h\notin\mathbb{F}_{p}$. Therefore, $|S_2| = 0$ only when $c=\frac{h_0^2}{h_1^2}$. But this is clearly not possible since the quadratic non-residue, c cannot be a square of any element in $\mathbb{F}_{p}$ ($\frac{h_0}{h_1} \in \mathbb{F}_{p}$). Hence, columns of matrices $S_1$ and $S_2$ are linearly independent over $\mathbb{F}_{p}$, implying that the desired and interference signals do not overlap. 

For the case of p=2, we are able to solve all possible cases numerically using MATLAB by constructing beamforming matrices ${\bf V}_{11}$ and ${\bf V}_{21}$ such that the columns of matrices $S_1$ and $S_2$ are linearly independent.
Thus, when $h\notin \mathbb{F}_{p}$, we have shown that the desired signals are resolvable, and sum rate of $\frac{4}{3}$ is achievable for channels over $\mathbb{F}_{p^2}$ for all $p$.
\end{proof}

\subsection{Appendix III - Zero Channels in 3-user Interference channel}
Here, we deal with realizations of the 3-user interference channel where some of the channel coefficients are zero. 

\begin{theorem}\label{theorem:zero_int}
For the 3 user interference channel over $\mathbb{F}_{p^n}$, if one or more of the channel coefficients $h_{ji}$ is equal to zero, the capacity results are given as follows:
\begin{enumerate}
\item If all three direct channels are zero, then $C = C_{linear} = 0$.
\item If any two direct channels are zero, then $C = C_{linear} = 1$.
\item If exactly one direct channel is zero, then $C = C_{linear} = 1$ or $C = C_{linear} = 2$, depending on whether any of the cross-channels between the other two users takes a non-zero value or they are all zero, respectively.
\item If all direct channels are non-zero and all 6 cross channels are zero, then $C = C_{linear} = 3$.
\item If all direct channels are non-zero and either 4 or 5 cross channels are zero, then $C = C_{linear} = 2$.
\item If all direct channels are non-zero and either 2 or 3 cross channels are zero, and if $h_{ij} = h_{ji} = 0$ for any one $\{i, j\} \in \{1,2,3\}$, then $C = C_{linear} = 2$.

\item In all other cases, the linear capacity is either $1$ or $1.5$ for channels over $\mathbb{F}_{p^n}$ with $p>3$ (the specific cases for each are identified in the proof).
\end{enumerate}
\end{theorem}
\begin{proof}
Cases 1,2,3,4,6 are trivial. The remaining cases are discussed below.

\noindent
\textbf{Case 5:} For all these channel structures, it can be shown that there always exists at least one $\{i, j\} \in \{1,2,3\}$  such that $h_{ij} = h_{ji} = 0$, and so only the sources $\{i, j\}$ can be used for transmission, leading to a sum rate of $2$ being achievable. Outer bound of $2$ follows by removing all but one non-zero cross-link. \\

\noindent
\textbf{Case 7:} \\
\noindent
For the achievability of sum rate of $1.5$,  consider the following:
\begin{enumerate}
\item
All channels are from $\mathbb{F}_{p^n}$. For even $n=2l$, we choose beamforming matrices $V \in \mathbb{F}_{p^n}^{1 \times l}$ at some of the sources and $V' \in \mathbb{F}_{p^n}^{1 \times l}$ at others, and precode $\frac{n}{2}=l$ symbols $x_k^1, x_k^2, \ldots, x_k^{l} \in \mathbb{F}_{p}$ for each channel use, at all 3 sources. We denote the $l$ columns of $V$ as $v_1, v_2, \ldots, v_l$ and those of $V'$ as $v'_1, v'_2, \ldots, v'_{l}$. These beam forming matrices would be chosen such that desired and interference symbols are linearly independent over $ \mathbb{F}_{p}$ at the destinations.
\item
When $n$ is odd,  2 symbol extensions are used wherein the beamforming matrix $V \in \mathbb{F}_{p^n}^{2 \times n}$ is used at some of the sources and $V' \in \mathbb{F}_{p^n}^{2 \times n}$ at others. Over 2 channel uses, $n$ input symbols are precoded at each source. Columns of $V$ and $V'$ are then chosen such that desired and interference symbols are linearly independent over $ \mathbb{F}_{p}$ at all destinations. Linear independence arguments follow similar to case of even $n$. 
\end{enumerate}

 We describe only even $n$ for various channel structures, for brevity.

Let us first consider the setting where 3 cross channels are zero. There are 5 distinct channel structures corresponding to any three cross channels being zero, and all other channel structures (${6 \choose 3}-5 = 15$) are isomorphic to them. These 5 channel structures are shown in Fig. \ref{fig:user3_cross3_5}. Of these, A, B, C belong to Case 5, and are therefore trivial.

\begin{figure}[h]
\begin{center}
\includegraphics[scale=0.5]{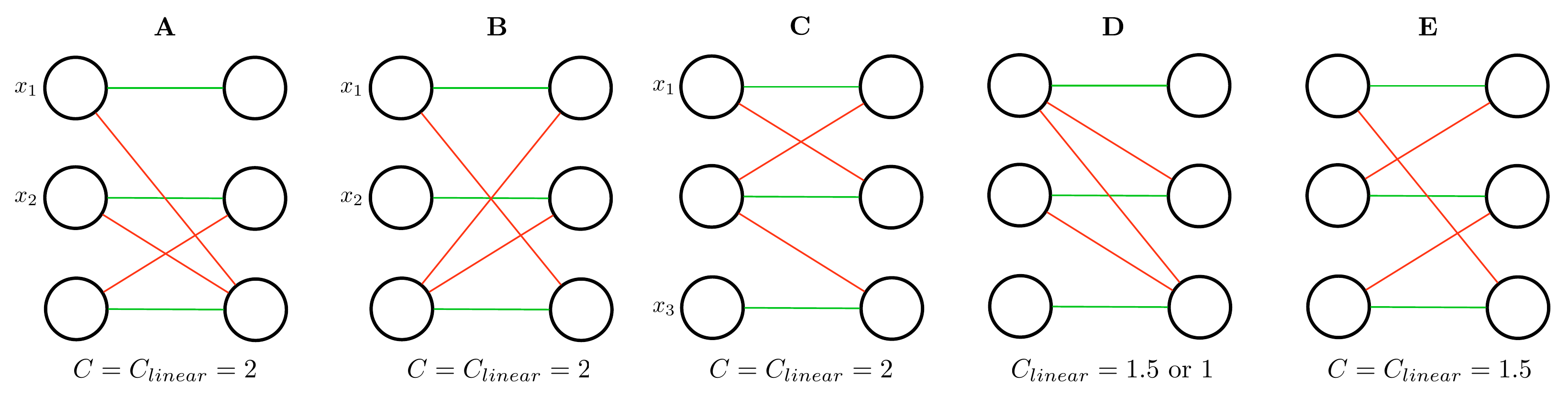}
\end{center}
\caption{Distinct channel structures with 3 cross channels as 0}\label{fig:user3_cross3_5}
\end{figure}

\noindent
\textbf{Structure D:} \\
For this structure, interference from sources 1 and 2 need to be aligned at destination 3. The normalized channel for this structure is illustrated in Fig. \ref{fig:str_d_norm}.
\begin{figure}[h]
\begin{center}
\includegraphics[scale=0.7]{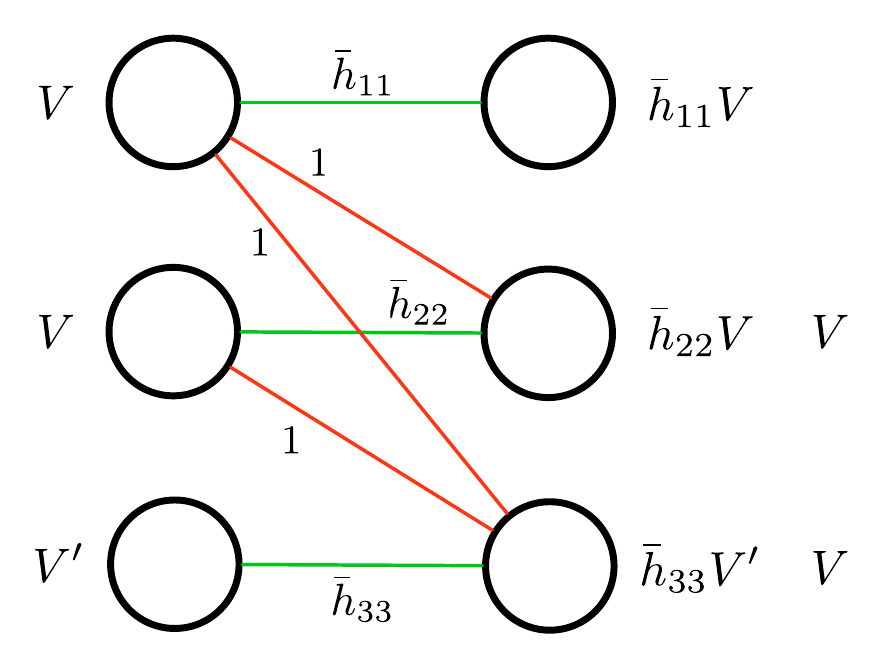}
\end{center}
\caption{Normalized channel of structure D}\label{fig:str_d_norm}
\end{figure}
\\
Beam forming matrix $V$ is used at sources 1 and 2, and $V'$ is used at source 3. Signal spaces at 3 destinations are then given by
\begin{eqnarray}
S_1 = [\bar{h}_{11}V] = [\bar{h}_{11}v_1, \hspace{3mm} \bar{h}_{11}v_2, \hspace{3mm}  \ldots, \hspace{3mm} \bar{h}_{11}v_l] \\
S_2 = [\bar{h}_{22}V \hspace{5mm} V] = [\bar{h}_{22}v_1, \hspace{3mm} \bar{h}_{22}v_2, \hspace{3mm}  \ldots, \hspace{3mm}  \bar{h}_{22}v_l, \hspace{3mm} v_1, \hspace{3mm} v_2, \hspace{3mm}  \ldots, \hspace{3mm}  v_l] \\
S_3 = [\bar{h}_{33}V' \hspace{5mm} V] = [\bar{h}_{33}v'_1, \hspace{3mm} \bar{h}_{33}v'_2, \hspace{3mm}  \ldots, \hspace{3mm}  \bar{h}_{33}v'_l \hspace{3mm} v_1, \hspace{3mm} v_2, \hspace{3mm}  \ldots, \hspace{3mm}  v_l] 
\end{eqnarray}
Consider signal space at destination 2. Let us choose $v_1$ as 1, then if $\bar{h}_{22} \notin \mathbb{F}_{p}$, $[\bar{h}_{22}v_1 \hspace{3mm} v_1]$ are linearly independent over $\mathbb{F}_{p}$. Now let us construct $v_2$ such that 4 columns of $S_2$, $[\bar{h}_{22}v_1 \hspace{3mm} v_1 \hspace{3mm} \bar{h}_{22}v_2 \hspace{3mm} v_2]$ are linearly independent over $\mathbb{F}_p$.   
\begin{eqnarray}
\text{From $S_2$, } \hspace{3mm}  v_2\notin A&\triangleq&\left\{  \frac{(\alpha_1 \bar{h}_{22} + \alpha_2)v_1}{\beta_1 \bar{h}_{22}+\beta_2}: \alpha_1, \alpha_2, \beta_1, \beta_2\in\mathbb{F}_p, (\beta_1,\beta_2)\neq(0,0)\right\} \label{eqn_d1_1}
\end{eqnarray}
Now we note that
\begin{eqnarray}
|A| &\leq& \frac{(p^2-1)p^{2}}{p-1} = p^3+p^2
\end{eqnarray}

There are $p^{n}$ choices for $v_2$, and since $p^n > (p^3 + p^2)$ for all $p$, there exist choices for $v_2$ such that condition \eqref{eqn_d1_1} holds. Choosing $v_2$ from those, we note that 4 columns of $S_2$ are linearly independent over $\mathbb{F}_p$. We proceed recursively in a similar manner, for choosing columns $v_3, v_4, \ldots, v_{l-1}$ such that $6, 8, \ldots, 2(l-1)$ columns are linearly independent over $\mathbb{F}_p$ respectively, in $S_2$. \\
\indent Let us now discuss the last iteration wherein we choose column $v_l$ such that all $n=2l$ columns are linearly independent over $\mathbb{F}_p$ in $S_2$, given that $2l-2$ columns are already linearly independent with appropriate choices of $v_1, v_2, \ldots, v_{l-1}$.
\begin{align} \nonumber
\text{From $S_2$, }  \hspace{3mm} v_l \notin A \triangleq \{  \frac{(\alpha_1 \bar{h}_{22} + \alpha_2)v_1 + (\alpha_3 \bar{h}_{22} + \alpha_4)v_2 + \cdots + (\alpha_{2l-3} \bar{h}_{22} + \alpha_{2l-2})v_{l-1}}{\beta_1 \bar{h}_{22}+\beta_2}: \\
\alpha_i, \beta_1, \beta_2\in\mathbb{F}_p, i\in \{1,\ldots,2l-2\}, (\beta_1,\beta_2)\neq(0,0)\} \label{eqn_d1_2}
\end{align}
Now we note that
\begin{eqnarray}
|A| &\leq& \frac{(p^2-1)p^{2l-2}}{p-1} = p^{2l-1}+p^{2l-2}
\end{eqnarray}
\indent There are $p^{n} = p^{2l}$ choices for $v_l$, and since $p^{2l} > (p^{2l-1}+p^{2l-2})$ for all $p$, there exist choices for $v_l$ such that condition \eqref{eqn_d1_2} holds. Choosing $v_l$ from those, we note that all $n$ columns of $S_2$ are linearly independent over $\mathbb{F}_p$. Also, it can be noted that $l = \frac{n}{2}$ columns of $V$ in $S_1$ and $S_3$ are linearly independent over $\mathbb{F}_p$. Destination 1 does not receive any interference and so desired symbols are resolvable. \\
\indent Let us now consider destination 3 where interference is aligned in $\frac{n}{2}=l$ linearly independent columns of $V$. Since source 3 does not cause interference anywhere, $V'$ is trivially chosen to be $\frac{1}{\bar{h}_{33}}$ times the remaining $n/2$ basis vectors. 
Hence, desired and interference symbols are linearly independent at all destinations. Thus, sum rate of $\frac{3}{2}$ is achieved for structure D in Fig. \ref{fig:str_d_norm}, with channels over $\mathbb{F}_{p^n}$ for all even $n$, if $\bar{h}_{22} \notin \mathbb{F}_p$. \\
\indent Fraction of channels for which scheme achieves $\frac{3}{2}$ sum rate is given by
\begin{eqnarray}
\frac{p^n-p}{p^n} = 1-\frac{1}{p^{n-1}} \to 1 \text{ for large } p,n
\end{eqnarray}
$\frac{3}{2}$  is also an information theoretic outer bound on sum rate for structure D  because the sum-rate of any two users is bounded by 1. However, when $\bar{h}_{22} =1$,  then arguing along the lines  of \cite{Cadambe_Jafar_inseparable} we find that destination 3 can decode all three messages, so that the information theoretic sum-capacity bound = 1. For all other cases where  $\bar{h}_{22}\in \mathbb{F}_p$ but $\bar{h}_{22}\notin \{0,1\}$,  the linear capacity is still 1 (because the linear capacity does not depend on the scaling of channel coefficients by non-zero $\mathbb{F}_p$ elements) but the information theoretic capacity is unknown.

Thus, structure D has linear capacity of $1.5$ if $\bar{h}_{22} \notin \mathbb{F}_p$, and $1$ otherwise. \\

\noindent
\textbf{Structure E:} For structure E, the sum rate of $1.5$ is achieved  even without channel knowledge at the transmitters. For example, transmitter 1 sends an $\mathbb{F}_{p^n}$ symbol only over the first channel use and stays quiet over the second channel use, transmitter 2 sends a $\mathbb{F}_{p^n}$ symbol over the second channel use and remains quiet over the first channel use, and transmitter 3 repeats its $\mathbb{F}_{p^n}$ symbol over both channel uses. This allows each receiver to decode its desired symbols. The outer bound of $1.5$ applies because the sum-capacity of any two users is 1. Thus, structure E has $C=C_{linear}=1.5$. \\

\indent Next let us consider cases where 2 cross channels are 0, shown in Fig. \ref{fig:user3_cross2_5}. Structure F belongs to Case 5, so it is trivial. \\
\begin{figure}[h]
\begin{center}
\includegraphics[scale=0.5]{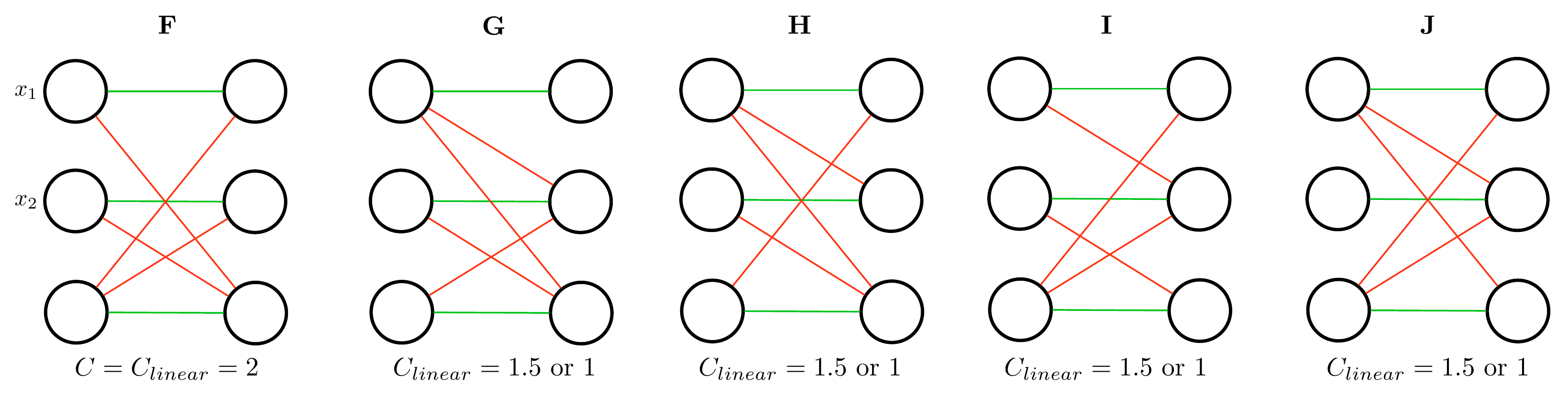}
\end{center}
\caption{Distinct channel structures with 2 cross channels as 0}\label{fig:user3_cross2_5}
\end{figure}
\\

\noindent
\textbf{Structure G:} The normalized channel for this structure is illustrated in Fig. \ref{fig:str_g_norm}. For this structure, signals from sources 1 and 2 need to be aligned at destination 3 and remain resolvable at destination 2. Following the proof for structure $D$, this can be done if $\bar{h}_{22}\notin\mathbb{F}_p$. Similarly, signals from sources 1 and 3 need to align at destination 2 and remain resolvable at destination 3. This can be done if $\bar{h}_{33}\notin\mathbb{F}_p$. We choose $V$ such that both $S_2 = [\bar{h}_{22}V \hspace{3mm} V]$  and $S_3 = [\bar{h}_{33}V \hspace{3mm} V]$ are linearly independent over $\mathbb{F}_p$, which can be shown to be possible for all $p>2$. Thus, sum rate of $\frac{3}{2}$ is achieved for structure G in Fig. \ref{fig:str_g_norm}, with channels over $\mathbb{F}_{p^n}$ for all even $n$, if $\bar{h}_{22}, \bar{h}_{33} \notin \mathbb{F}_p$. The outer bound of $\frac{3}{2}$ follows from the pair-wise bounds. If all non-zero channels are equal to 1, then the argument of \cite{Cadambe_Jafar_inseparable}  shows that one destination can decode all messages, i.e., $C=C_{linear}=1$. In all other cases with non-zero $\bar{h}_{kk} \in \mathbb{F}_p$ for any $k=2,3$, the linear capacity is still one because the linear capacity is not affected by a scaling of channel coefficients by non-zero constants in $\mathbb{F}_p$. Thus structure G has linear-scheme capacity of $\frac{3}{2}$ if $\bar{h}_{kk} \notin \mathbb{F}_p, k\in \{2,3\}$, and $1$ otherwise. \\

\begin{figure}[h]
\begin{minipage}[b]{0.45\linewidth}
\begin{center}
\includegraphics[scale=0.7]{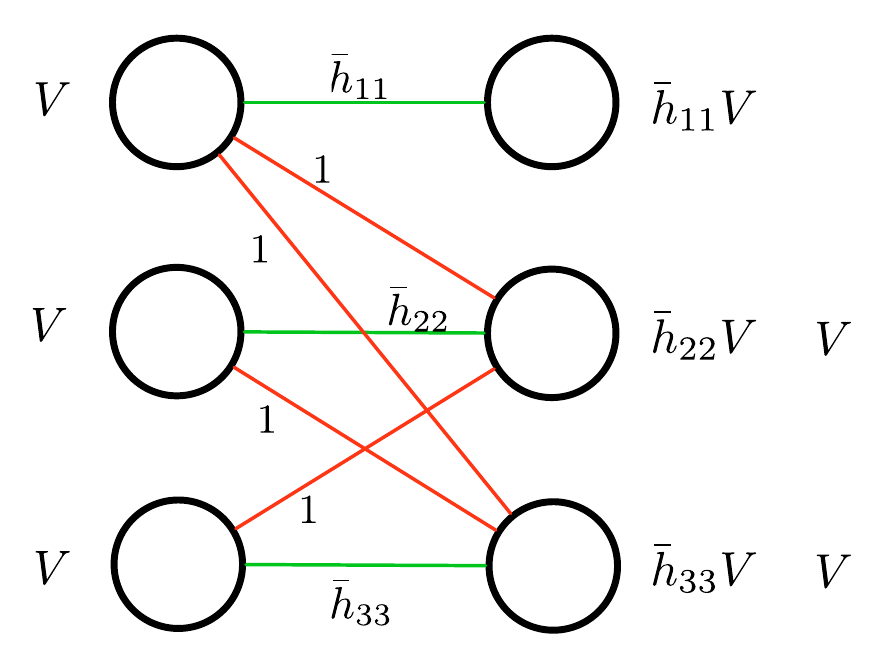}
\end{center}
\caption{Normalized channel - structure G}\label{fig:str_g_norm}
\end{minipage}
\hspace{0.5cm}
\begin{minipage}[b]{0.45\linewidth}
\begin{center}
\includegraphics[scale=0.7]{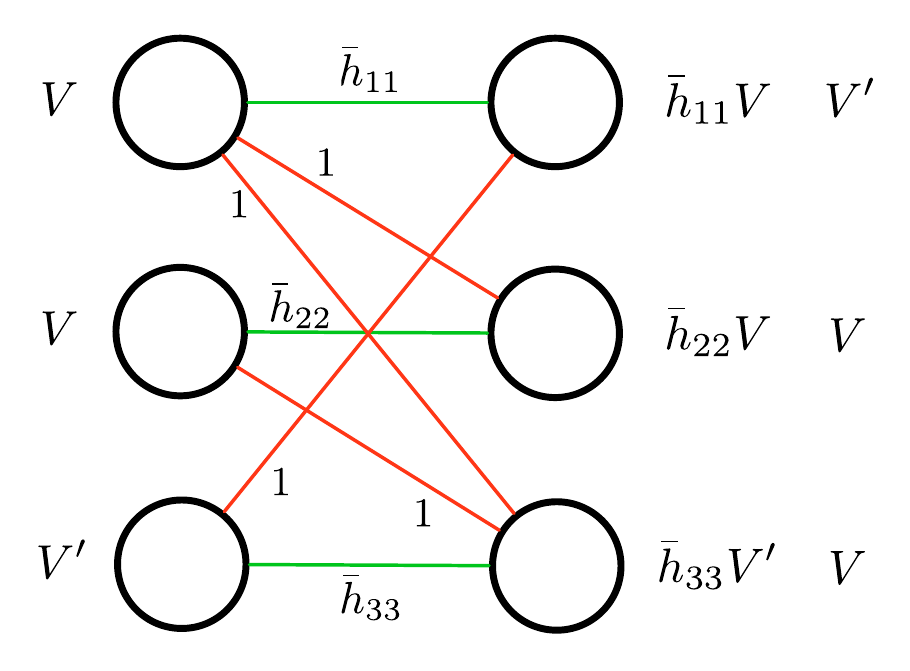}
\end{center}
\caption{Normalized channel - structure H}\label{fig:str_h_norm}
\end{minipage}
\end{figure}

\noindent
\textbf{Structure H:} \\
The normalized channel for this structure is illustrated in Fig. \ref{fig:str_h_norm}. For this structure, signals from sources 1 and 2 need to be aligned at destination 3 and remain resolvable at destination 2. Following the proof for structure $D$, this can be done if $\bar{h}_{22}\notin\mathbb{F}_p$. We choose $V'$ such that both $S_1 = [\bar{h}_{11}V \hspace{3mm} V']$  and $S_3 = [\bar{h}_{33}V' \hspace{3mm} V]$ are linearly independent over $\mathbb{F}_p$, which can be shown to be possible for all $p>2$. Thus, sum rate of $\frac{3}{2}$ is achieved for structure H in Fig. \ref{fig:str_h_norm}, with channels over $\mathbb{F}_{p^n}$ for all even $n$, if $\bar{h}_{22} \notin \mathbb{F}_p$.  The outer bound of $\frac{3}{2}$ follows from the pair-wise bounds. If all non-zero channels are equal to 1, then the argument of \cite{Cadambe_Jafar_inseparable}  shows that one destination can decode all messages, i.e., $C=C_{linear}=1$. In all other cases with non-zero $\bar{h}_{22} \in \mathbb{F}_p$, the linear capacity is still one because the linear capacity is not affected by a scaling of channel coefficients by non-zero constants in $\mathbb{F}_p$. Thus structure H has linear-scheme capacity of $\frac{3}{2}$ if $\bar{h}_{22} \notin \mathbb{F}_p$, and $1$ otherwise. \\

\noindent
\textbf{Structure I:} \\
The normalized channel for this structure is illustrated in Fig. \ref{fig:str_i_norm}. For this structure, signals from sources 1 and 3 need to be aligned at destination 2 and remain resolvable at destination 1. Following the proof for structure $D$, this can be done if $\bar{h}_{11}\notin\mathbb{F}_p$.  We choose $V'$ such that both $S_2 = [\bar{h}_{22}V' \hspace{3mm} V]$  and $S_3 = [\bar{h}_{33}V \hspace{3mm} V']$ are linearly independent over $\mathbb{F}_p$, which can be shown to be possible for all $p>2$. Thus, sum rate of $\frac{3}{2}$ is achieved for structure H in Fig. \ref{fig:str_i_norm}, with channels over $\mathbb{F}_{p^n}$ for all even $n$, if $\bar{h}_{11} \notin \mathbb{F}_p$.  The outer bound of $\frac{3}{2}$ follows from the pair-wise bounds. If all non-zero channels are equal to 1, then the argument of \cite{Cadambe_Jafar_inseparable}  shows that one destination can decode all messages, i.e., $C=C_{linear}=1$. In all other cases with non-zero $\bar{h}_{11} \in \mathbb{F}_p$, the linear capacity is still one because the linear capacity is not affected by a scaling of channel coefficients by non-zero constants in $\mathbb{F}_p$. Thus structure I has linear-scheme capacity of $\frac{3}{2}$ if $\bar{h}_{11} \notin \mathbb{F}_p$, and $1$ otherwise. \\
\begin{figure}[h]
\begin{minipage}[b]{0.45\linewidth}
\begin{center}
\includegraphics[scale=0.7]{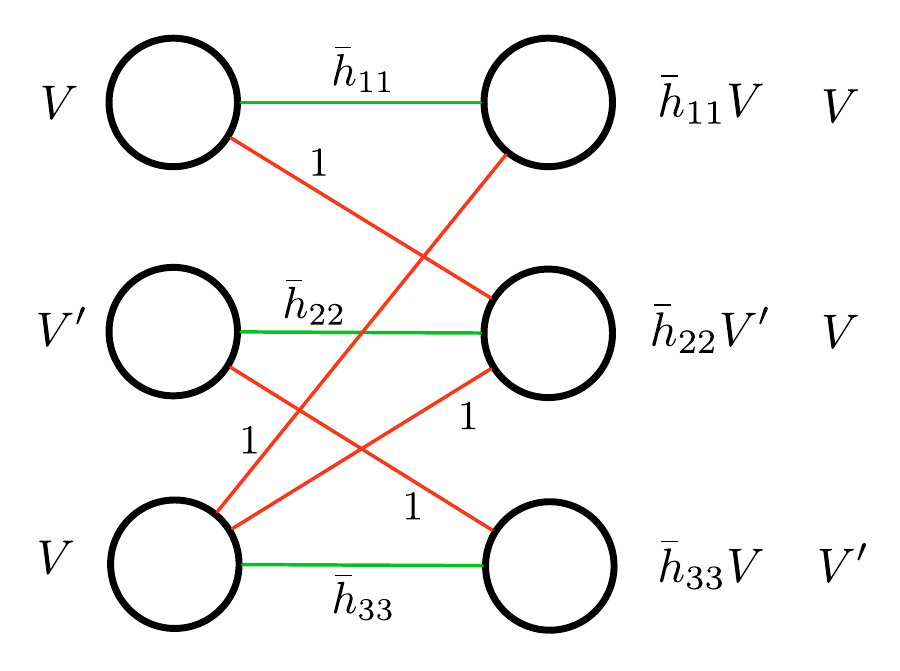}
\end{center}
\caption{Normalized channel - structure I}\label{fig:str_i_norm}
\end{minipage}
\hspace{0.5cm}
\begin{minipage}[b]{0.45\linewidth}
\begin{center}
\includegraphics[scale=0.7]{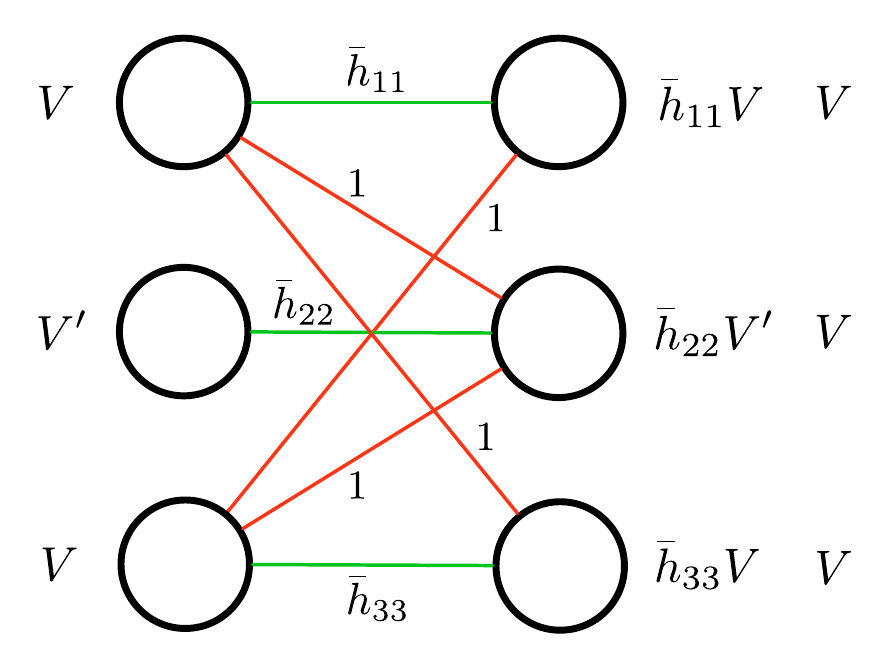}
\end{center}
\caption{Normalized channel - structure J}\label{fig:str_j_norm}
\end{minipage}
\end{figure}

\noindent
\textbf{Structure J:} \\
The normalized channel for this structure is illustrated in Fig. \ref{fig:str_j_norm}. For this structure, signals from sources 1 and 3 need to be aligned at destination 2 but remain resolvable at destinations 1 and 3. Following the proof for structure $D$, this can be done if $\bar{h}_{11}, \bar{h}_{33}\notin\mathbb{F}_p$.  We choose $V$ such that both $S_1 = [\bar{h}_{11}V \hspace{3mm} V]$  and $S_3 = [\bar{h}_{33}V \hspace{3mm} V]$ are linearly independent over $\mathbb{F}_p$, which can be shown to be possible for all $p>2$. Thus, sum rate of $\frac{3}{2}$ is achieved for structure J in Fig. \ref{fig:str_j_norm}, with channels over $\mathbb{F}_{p^n}$ for all even $n$, if $\bar{h}_{11}, \bar{h}_{33} \notin \mathbb{F}_p$.  The outer bound of $\frac{3}{2}$ follows from the pair-wise bounds. If all non-zero channels are equal to 1, then the argument of \cite{Cadambe_Jafar_inseparable}  shows that one destination can decode all messages, i.e., $C=C_{linear}=1$. In all other cases with non-zero $\bar{h}_{kk} \in \mathbb{F}_p$ for any $k=1,3$, the linear capacity is still one because the linear capacity is not affected by a scaling of channel coefficients by non-zero constants in $\mathbb{F}_p$. Thus structure J has linear-scheme capacity of $\frac{3}{2}$ if $\bar{h}_{kk} \notin \mathbb{F}_p, k\in \{1,3\}$, and $1$ otherwise. \\

Finally, let us now consider the setting where only one cross channel is zero. \\

\noindent\textbf{Structure K:} \\
The normalized channel for this structure is illustrated in Fig. \ref{fig:str_k_norm}. For this single channel structure, interference from sources 2 and 3 need to be aligned at destination 1, and interference from sources 1 and 3 need to be aligned at destination 2.

Beam forming matrix $V$ is used at all 3 sources. Signal spaces at 3 destinations are then given by
\begin{eqnarray}
S_1 = [\bar{h}_{11}V \hspace{5mm} V] = [\bar{h}_{11}v_1, \hspace{3mm} \bar{h}_{11}v_2, \hspace{3mm}  \ldots, \hspace{3mm} \bar{h}_{11}v_l, \hspace{3mm} v_1, \hspace{3mm} v_2, \hspace{3mm}  \ldots, \hspace{3mm}  v_l] \\
S_2 = [\bar{h}_{22}V \hspace{5mm} V] = [\bar{h}_{22}v_1, \hspace{3mm} \bar{h}_{22}v_2, \hspace{3mm}  \ldots, \hspace{3mm}  \bar{h}_{22}v_l, \hspace{3mm} v_1, \hspace{3mm} v_2, \hspace{3mm}  \ldots, \hspace{3mm}  v_l] \\
S_3 = [\bar{h}_{33}V \hspace{5mm} V] = [\bar{h}_{33}v_1, \hspace{3mm} \bar{h}_{33}v_2, \hspace{3mm}  \ldots, \hspace{3mm}  \bar{h}_{33}v_l, \hspace{3mm} v_1, \hspace{3mm} v_2, \hspace{3mm}  \ldots, \hspace{3mm}  v_l] 
\end{eqnarray}
Let us choose $v_1$ as 1, then if $\bar{h}_{11}, \bar{h}_{22}, \bar{h}_{33} \notin \mathbb{F}_{p}$, $[\bar{h}_{11}v_1 \hspace{3mm} v_1]$, $[\bar{h}_{22}v_1 \hspace{3mm} v_1]$ and $[\bar{h}_{33}v_1 \hspace{3mm} v_1]$ are linearly independent over $\mathbb{F}_{p}$. Now let us construct $v_2$ such that 4 columns of $S_k, k\in \{1,2,3\}$ are linearly independent.   
\begin{figure}[h]
\begin{center}
\includegraphics[scale=0.7]{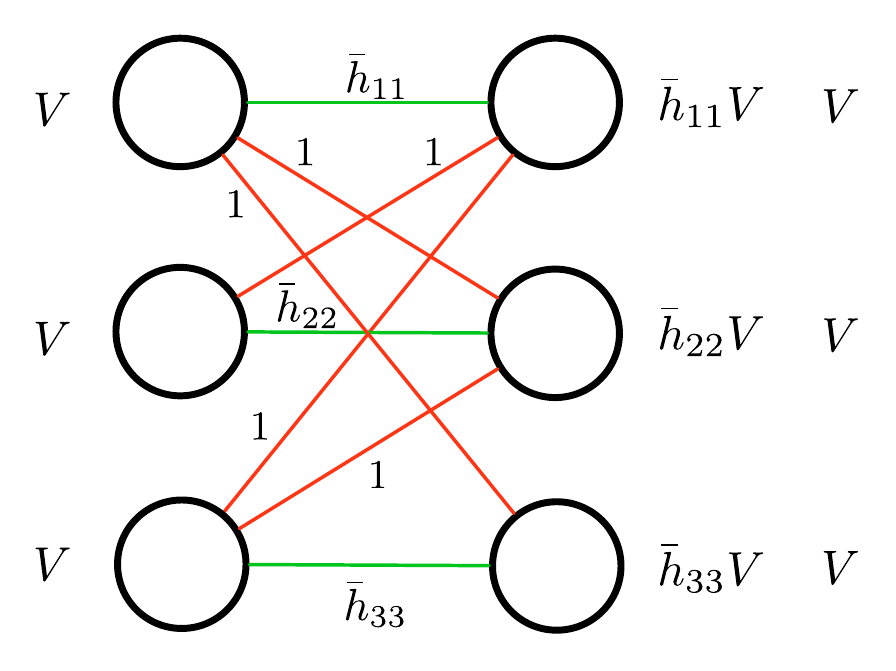}
\end{center}
\caption{Normalized channel of structure K}\label{fig:str_k_norm}
\end{figure}
\\
\begin{align}
\text{From $S_k$, } \hspace{3mm}  v_2\notin A_k \triangleq \left\{  \frac{(\alpha_1 \bar{h}_{kk} + \alpha_2)v_1}{\beta_1 \bar{h}_{kk}+\beta_2}: \alpha_1, \alpha_2, \beta_1, \beta_2\in\mathbb{F}_p, (\beta_1,\beta_2)\neq(0,0)\right\} , k \in \{1,2,3\} \label{eqn_k1_1}
\end{align}
Now we note that
\begin{eqnarray}
|A_k| &\leq& \frac{(p^2-1)p^{2}}{p-1} = p^3+p^2 \\
|A_1 \cup A_2 \cup A_3| &\leq& 3(p^3+p^2)
\end{eqnarray}
There are $p^{n}$ choices for $v_2$, and since $p^n > 3(p^3 + p^2)$ for all $p>3$, there exist choices for $v_2$ such that all 3 conditions of \eqref{eqn_k1_1} hold. Choosing $v_2$ from those, we note that 4 columns of $S_k, k\in \{1,2,3\}$ are linearly independent over $\mathbb{F}_p$. We proceed recursively in a similar manner, for choosing columns $v_3, v_4, \ldots, v_{l-1}$ such that $6, 8, \ldots, 2(l-1)$ columns are linearly independent over $\mathbb{F}_p$ respectively, in $S_k, k\in \{1,2,3\}$. \\
\indent For the last iteration, we choose column $v_l$ such that all $n=2l$ columns are linearly independent over $\mathbb{F}_p$ in $S_k, k\in \{1,2,3\}$, given that $2l-2$ columns are already linearly independent with appropriate choices of $v_1, v_2, \ldots, v_{l-1}$.
\begin{align} \nonumber
\text{From $S_k$, }  \hspace{3mm} v_l \notin A_k \triangleq \{  \frac{(\alpha_1 \bar{h}_{kk} + \alpha_2)v_1 + (\alpha_3 \bar{h}_{kk} + \alpha_4)v_2 + \cdots + (\alpha_{2l-3} \bar{h}_{kk} + \alpha_{2l-2})v_{l-1}}{\beta_1 \bar{h}_{kk}+\beta_2}: \\
\alpha_i, \beta_1, \beta_2\in\mathbb{F}_p, i\in \{1,\ldots,2l-2\}, (\beta_1,\beta_2)\neq(0,0)\}, k \in \{1,2,3\}  \label{eqn_k1_2}
\end{align}
Now we note that
\begin{eqnarray}
|A_k| &\leq& \frac{(p^2-1)p^{2l-2}}{p-1} = p^{2l-1}+p^{2l-2} \\
|A_1 \cup A_2 \cup A_3| &\leq& 3(p^{2l-1}+p^{2l-2})
\end{eqnarray}
\indent There are $p^{n} = p^{2l}$ choices for $v_l$, and since $p^{2l} > 3(p^{2l-1}+p^{2l-2})$ for all $p>3$, there exist choices for $v_l$ such that conditions of \eqref{eqn_k1_2} hold. Choosing $v_l$ from those, we note that all $n$ columns of $S_1, S_2, S_3$ are linearly independent over $\mathbb{F}_p$.  \\
\indent Hence, desired and interference symbols are linearly independent at all destinations. Thus, sum rate of $\frac{3}{2}$ is achieved for structure K in Fig. \ref{fig:str_k_norm}, with channels over $\mathbb{F}_{p^n}$ for all even $n$, if $\bar{h}_{11},\bar{h}_{22}, \bar{h}_{33} \notin \mathbb{F}_p$. \\
\indent Fraction of channels for which scheme achieves $\frac{3}{2}$ sum rate is given by
\begin{eqnarray}
(\frac{p^n-p}{p^n})^3 = (1-\frac{1}{p^{n-1}})^3 \to 1 \text{ for large }p,n
\end{eqnarray}
The outer bound of $\frac{3}{2}$ follows from the pair-wise bounds. If all channels are equal to 1, then the argument of \cite{Cadambe_Jafar_inseparable}  shows that one destination can decode all messages, i.e., $C=C_{linear}=1$. In all other cases with non-zero $\bar{h}_{kk} \in \mathbb{F}_p$ for any $k$, the linear capacity is still one because the linear capacity is not affected by a scaling of channel coefficients by non-zero constants in $\mathbb{F}_p$. Thus structure K has linear-scheme capacity of $\frac{3}{2}$ if $\bar{h}_{kk} \notin \mathbb{F}_p, k\in \{1,2,3\}$, and $1$ otherwise.
\hfill\QED \\
\end{proof}

\bibliographystyle{ieeetr}
\bibliography{Thesis}

\end{document}